\documentclass[11pt]{article}
\pdfoutput=1
\usepackage[T1]{fontenc}
\usepackage{lmodern}
\usepackage[protrusion=true,expansion=true]{microtype}
\usepackage{amsmath,amssymb,amsfonts,amsthm}
\usepackage{subcaption}
\usepackage{graphicx}
\usepackage{setspace}
\usepackage{fullpage}
\usepackage[backref=page]{hyperref}
\usepackage{color}
\usepackage{wrapfig}
\usepackage{tikz}
\usetikzlibrary{decorations.pathreplacing}
\usepackage{algorithm}
\usepackage[noend]{algpseudocode}
\usepackage[framemethod=tikz]{mdframed}
\usepackage{xspace}
\usepackage{pgfplots}
\usepackage{framed}
\usepackage{thmtools}
\usepackage{thm-restate}
\usepackage{tabu}
\usepackage{fancyhdr}
\usepackage{comment}
\pgfplotsset{compat=1.5}

\newtheorem{theorem}{Theorem}[section]

\newtheorem{lemma}[theorem]{Lemma}

\newtheorem{definition}[theorem]{Definition}

\newtheorem{fact}[theorem]{Fact}

\newenvironment{proofof}[1]{\begin{trivlist} \item {\bf Proof
#1:~~}}
  {\qed\end{trivlist}}

\newcommand{\namedref}[2]{\hyperref[#2]{#1~\ref*{#2}}}
\newcommand{\thmlab}[1]{\label{thm:#1}}
\newcommand{\thmref}[1]{\namedref{Theorem}{thm:#1}}
\newcommand{\lemlab}[1]{\label{lem:#1}}
\newcommand{\lemref}[1]{\namedref{Lemma}{lem:#1}}

\newcommand{\seclab}[1]{\label{sec:#1}}
\newcommand{\secref}[1]{\namedref{Section}{sec:#1}}

\newcommand{\factlab}[1]{\label{fact:#1}}
\newcommand{\factref}[1]{\namedref{Fact}{fact:#1}}

\newcommand{\figlab}[1]{\label{fig:#1}}
\newcommand{\figref}[1]{\namedref{Figure}{fig:#1}}
\newcommand{\alglab}[1]{\label{alg:#1}}
\renewcommand{\algref}[1]{\namedref{Algorithm}{alg:#1}}

\newcommand{\deflab}[1]{\label{def:#1}}
\newcommand{\defref}[1]{\namedref{Definition}{def:#1}}

\newcommand{\PPr}[1]{\ensuremath{\mathbf{Pr}\left[#1\right]}}

\newcommand{\Ex}[1]{\ensuremath{\mathbb{E}\left[#1\right]}}

\renewcommand{\O}[1]{\ensuremath{\mathcal{O}\left(#1\right)}}
\newcommand{\tO}[1]{\ensuremath{\tilde{\mathcal{O}}\left(#1\right)}}
\newcommand{\eps}{\varepsilon}

\def \RoughSens {\mdef{\textsc{RoughSens}}}
\def \BatchSens {\mdef{\textsc{BatchSens}}}

\def \dist    {\mdef{\text{dist}}}

\def \ba    {\mdef{\mathbf{a}}}

\def \calC    {\mdef{\mathcal{C}}}

\def \calE    {\mdef{\mathcal{E}}}

\def \calG    {\mdef{\mathcal{G}}}

\def \calN    {\mdef{\mathcal{N}}}

\def \calS    {\mdef{\mathcal{S}}}

\def \bA    {\mdef{\mathbf{A}}}
\def \bB    {\mdef{\mathbf{B}}}

\def \bF    {\mdef{\mathbf{F}}}
\def \bG    {\mdef{\mathbf{G}}}

\def \bM    {\mdef{\mathbf{M}}}
\def \bP    {\mdef{\mathbf{P}}}
\def \bR    {\mdef{\mathbf{R}}}

\def \bQ    {\mdef{\mathbf{Q}}}

\def \bU    {\mdef{\mathbf{U}}}

\def \bZ    {\mdef{\mathbf{Z}}}
\def \bW    {\mdef{\mathbf{W}}}
\def \bb    {\mdef{\mathbf{b}}}

\def \br    {\mdef{\mathbf{r}}}

\def \bv    {\mdef{\mathbf{v}}}
\def \bx    {\mdef{\mathbf{x}}}
\def \by    {\mdef{\mathbf{y}}}
\def \bz    {\mdef{\mathbf{z}}}
\def \bg    {\mdef{\mathbf{g}}}

\newcommand{\mdef}[1]{{\ensuremath{#1}}\xspace}  

\DeclareMathOperator*{\polylog}{polylog}
\DeclareMathOperator*{\poly}{poly}

\DeclareMathOperator*{\TreeDist}{TreeDist}

\DeclareMathOperator*{\Cost}{cost}

\DeclareMathOperator*{\nnz}{nnz}

\newcommand{\ignore}[1]{}

\newif\ifnotes\notestrue 
\ifnotes
\newcommand{\david}[1]{\textcolor{blue}{{\bf (David:} {#1}{\bf ) }} \marginpar{\tiny\bf
             \begin{minipage}[t]{0.5in}
               \raggedright S:
            \end{minipage}}}
\newcommand{\samson}[1]{\textcolor{red}{{\bf (Samson:} {#1}{\bf ) }} \marginpar{\tiny\bf
             \begin{minipage}[t]{0.5in}
               \raggedright S:
            \end{minipage}}}
\newcommand{\vincent}[1]{\textcolor{purple}{{\bf (Vincent:} {#1}{\bf ) }} \marginpar{\tiny\bf
             \begin{minipage}[t]{0.5in}
               \raggedright V:
            \end{minipage}}} 
\newcommand{\liudeng}[1]{\textcolor{magenta}{{\bf (Liudeng:} {#1}{\bf ) }} \marginpar{\tiny\bf
             \begin{minipage}[t]{0.5in}
               \raggedright L:
            \end{minipage}}}
\else
\newcommand{\samson}[1]{}
\newcommand{\liudeng}[1]{}
\newcommand{\david}[1]{}
\newcommand{\vincent}[1]{}
\fi

\makeatletter
\renewcommand*{\@fnsymbol}[1]{\textcolor{mahogany}{\ensuremath{\ifcase#1\or *\or \dagger\or \ddagger\or
 \mathsection\or \triangledown\or \mathparagraph\or \|\or **\or \dagger\dagger
   \or \ddagger\ddagger \else\@ctrerr\fi}}}
\makeatother

\providecommand{\email}[1]{\href{mailto:#1}{\nolinkurl{#1}\xspace}}

\definecolor{mahogany}{rgb}{0.75, 0.25, 0.0}
\definecolor{darkblue}{rgb}{0.0, 0.0, 0.55}
\definecolor{darkpastelgreen}{rgb}{0.01, 0.75, 0.24}
\definecolor{darkgreen}{rgb}{0.0, 0.2, 0.13}
\definecolor{darkgoldenrod}{rgb}{0.72, 0.53, 0.04}
\definecolor{darkred}{rgb}{0.55, 0.0, 0.0}
\definecolor{forestgreenweb}{rgb}{0.13, 0.55, 0.13}
\definecolor{greencss}{rgb}{0.0, 0.5, 0.0}
\definecolor{bleudefrance}{rgb}{0.19, 0.55, 0.91}
\definecolor{darkpastelpurple}{rgb}{0.59, 0.44, 0.84}
\definecolor{darkcerulean}{rgb}{0.03, 0.27, 0.49}

\hypersetup{
     colorlinks   = true,
     citecolor    = darkcerulean,
	 linkcolor	  = darkcerulean
}

\fancypagestyle{pg}
{
\lhead{}
\rhead{}
\cfoot{--\ \thepage\ --}

}

\AtBeginDocument{%
  \DeclareFontShape{T1}{lmr}{m}{scit}{<->ssub*lmr/m/scsl}{}%
}
\title{Fast, Space-Optimal Streaming Algorithms for Clustering and Subspace Embeddings}
\author{
Vincent Cohen-Addad\thanks{Google Research. E-mail: \email{cohenaddad@google.com}.} 
\and
Liudeng Wang\thanks{Texas A\&M University. E-mail: \email{eureka@tamu.edu}}
\and
David P. Woodruff\thanks{Carnegie Mellon University. E-mail: \email{dwoodruf@andrew.cmu.edu}. 
}
\and
Samson Zhou\thanks{Texas A\&M University. E-mail: \email{samsonzhou@gmail.com}. 
}
}
\date{}
\begin{document}
\pagestyle{pg}

\maketitle
\thispagestyle{empty}

\begin{abstract}
We show that both clustering and subspace embeddings can be performed in the streaming model with the same asymptotic efficiency as in the central/offline setting.

For $(k, z)$-clustering in the streaming model, we achieve a number of words of memory which is independent of the number $n$ of input points and the aspect ratio $\Delta$, yielding an optimal bound of $\tilde{\mathcal{O}}\left(\frac{dk}{\min(\varepsilon^4,\varepsilon^{z+2})}\right)$ words for accuracy parameter $\varepsilon$ on $d$-dimensional points. Additionally, we obtain amortized update time of $d\,\log(k)\cdot\text{polylog}(\log(n\Delta))$, which is an exponential improvement over the previous $d\,\text{poly}(k,\log(n\Delta))$. Our method also gives the fastest runtime for $(k,z)$-clustering even in the offline setting. 

For subspace embeddings in the streaming model, we achieve $\mathcal{O}(d)$ update time and space-optimal constructions,  using $\tilde{\mathcal{O}}\left(\frac{d^2}{\varepsilon^2}\right)$ words for $p\le 2$ and $\tilde{\mathcal{O}}\left(\frac{d^{p/2+1}}{\varepsilon^2}\right)$ words for $p>2$, showing that streaming algorithms can match offline algorithms in both space and time complexity. 
\end{abstract}

\clearpage
\setcounter{page}{1}

\allowdisplaybreaks

\section{Introduction}
The \emph{streaming model} of computation is a powerful framework for processing large-scale datasets that are too vast to be stored in memory. 
In the streaming setting, algorithms must quickly make decisions based on a sequence of updates that are irrevocably discarded after processing, without the possibility of revisiting past data. 
The primary challenge is to design algorithms that can efficiently approximate, compute, or detect properties of the underlying dataset, while ensuring both fast update time and using memory that is sublinear in both the dataset size and the length of the data stream.
This model is well-suited for applications in data summarization, where the goal is to extract meaningful insights from large, evolving datasets. 
Such applications include clustering, which is widely applied in domains such as customer segmentation, anomaly detection in network traffic, and pattern recognition in sensor networks, as well as subspace embeddings, which reduce dimensionality in areas including image processing, recommendation systems, and natural language processing. Despite significant progress, it remained open whether such fundamental tasks inherently require more space or time in the streaming model compared to their offline counterparts. 

\vspace{-0.1in}
\paragraph{Clustering and data streams.}
Clustering is a fundamental problem that seeks to partition a dataset so that similarly structured points are grouped together and differently structured points are separated. 
Variations of clustering are used across a wide range of applications, including bioinformatics, combinatorial optimization, computational geometry, computer graphics, data science, and machine learning. 
In the Euclidean $(k,z)$-clustering problem, the input is a set $X$ of $n$ points $x_1,\ldots,x_n\in\mathbb{R}^d$, a cluster parameter $k>0$, and a positive integer exponent $z>0$, and the goal is to minimize the clustering cost across all sets $\calC$ of at most $k$ centers:
\vspace{-0.1in}
\[
\min_{\calC\subset\mathbb{R}^d, |\calC|\le k}\Cost(X,\calC):=\min_{\calC\subset\mathbb{R}^d, |\calC|\le k}\sum_{i=1}^n\min_{c\in\calC}\|x_i-c\|_2^z.
\vspace{-0.05in}
\]
For $z=1$ and $z=2$, respectively, the $(k,z)$-clustering problem corresponds to the well-known $k$-median and $k$-means clustering problems. 

For clustering in the standard insertion-only streaming model, the points $x_1,\ldots,x_n$ of $X$ arrive sequentially along with an input accuracy parameter $\eps>0$. 
The goal is for each time $t\in[n]$ to find a clustering $\calC_t$ on the prefix $X_t$ of the first $t$ points of $X$ with cost that is a $(1+\eps)$-multiplicative approximation of the optimal clustering of $X_t$. 
To that end, note that a set $\calC$ of $k$ centers implicitly defines the overall clustering, since each point is assigned to the closest cluster center. 
In fact, it is not feasible to directly output a label for every point in $X$ using space sublinear in $n=|X|$, and so in some sense, outputting $\calC$ is the most reasonable objective we can hope for. 
We also remark that due to bit representation, input points cannot be represented in arbitrary precision and thus we generally rescale the points in $X$ to lie within the grid $\{1,\ldots,\Delta\}^d$. 

\vspace{-0.1in}
\paragraph{Subspace embeddings and data streams.} 
Subspace embeddings are a fundamental tool in dimensionality reduction, enabling the preservation of algebraic and geometric structure, while significantly reducing data dimensionality. 
The problem has been extensively studied, with applications spanning theoretical computer science, machine learning, randomized linear algebra, scientific computing, and data science. 
In the subspace embedding problem, the input is once again a set $A$ of $n$ points $x_1,\ldots,x_n\in\mathbb{R}^d$, which induces a matrix $\bA\in\mathbb{R}^{n\times d}$, where typically $n\gg d$. 
Given an accuracy parameter $\eps>0$, the goal is to construct a matrix $\bM$ such that for all vectors $\bx\in\mathbb{R}^d$,  
\[(1-\eps)\|\bM\bx\|_p\le\|\bA\bx\|_p\le(1+\eps)\|\bA\bx\|_p,\]
where $\|\cdot\|_p$ denotes the $L_p$ norm of a vector. 
In the standard insertion-only streaming model, the rows of $\bA$ arrive sequentially and at each time $t\in[n]$, the goal is to find a subspace embedding for the matrix $\bA_t$ consisting of the first $t$ rows of $\bA$. 

\vspace{-0.1in}
\paragraph{Streaming algorithms and coresets.}
The most common approach for both fast and space-efficient streaming algorithms for data summarization problems such as $(k,z)$-clustering and subspace embeddings is to collect a small weighted subset of the points from the input dataset. 
For clustering, this forms a (strong) coreset $Z_t$ at time $t$, which is a data structure that preserves the clustering cost for any choice $\calC$ of $k$ centers, i.e., $\Cost(X_t,\calC)\approx\Cost(Z_t,\calC)$. 
We can then perform an efficient $(k,z)$-clustering algorithm on each coreset $Z_t$ to identify an optimal or near-optimal set of cluster centers. 
For subspace embeddings, this forms a coreset $\bM_t$ at time $t$, from which we can approximately extract $\|\bA_t\bx\|_2$ for any $\bx\in\mathbb{R}^d$, i.e., $\|\bM_t\bx\|_2\approx\|\bA_t\bx\|_2$. 
Naturally, smaller coreset constructions correspond to both faster and more space-efficient algorithms. 

In the offline setting where the input dataset $X$ is given upfront and there are no space restrictions, there exist coreset constructions for both clustering and subspace embeddings that sample a number of weighted points that is independent of the size of the input. 
For $(k,z)$-clustering, there exist coreset constructions~\cite{Cohen-AddadSS21,Cohen-AddadLSS22,huang2023optimal,nips/Cohen-AddadLSSS22} that sample $\tO{\min\left(\frac{1}{\eps^2}\cdot k^{2-\frac{z}{z+2}},\frac{1}{\min(\eps^4,\eps^{2+z})}\cdot k\right)}$ weighted points of $X$\footnote{Note that the bound is actually $\tO{\min\left(\frac{k^{3/2}}{\eps^2},\frac{k}{\min(\eps^4,\eps^{2+z})}\right)}$ for $k$-means and $\tO{\min\left(\frac{k^{4/3}}{\eps^2},\frac{k}{\min(\eps^4,\eps^{2+z})}\right)}$ for $k$-median.}, where we use $\tO{f}$ for a function $f$ to denote $f\cdot\polylog(f)$.  
Similarly for subspace embeddings, there exist coreset constructions that sample $\tO{\frac{d}{\eps^2}}$ weighted rows of $\bA$~\cite{DrineasMM06a,DrineasMM06b,MagdonIsmail10,Woodruff14}. 
These coreset constructions have size that is \emph{independent} of the size $n$ of the input dataset, which is especially important as modern datasets often consist of hundreds of millions of points. 

On the other hand, despite a long line of work on clustering in the streaming model~\cite{Har-PeledM04,Har-PeledK07,Chen09,FeldmanL11,BravermanFLR19,Cohen-AddadWZ23,WoodruffZZ23}, all known results achieved space overheads with dependencies in $n$, compared to the offline setting. 
A natural question is thus: 
\begin{quote}
\vspace{-0.05in}
Is there an inherent \emph{space complexity} cost for data summarization problems, such as $(k,z)$-clustering and subspace embeddings, in the streaming model?
\vspace{-0.15in}
\end{quote}

\paragraph{Input size and runtime.} 
We remark that even in the offline setting, where the entire input is given, $\Omega(nd)$ runtime is necessary simply to process a dataset of size $n$ in $d$-dimensional space. 
Thus, $\Omega(nd)$ runtime is necessary for both $(k,z)$-clustering and subspace embeddings without additional assumptions on the input. 
The lower bound is matched by an offline algorithm that uses $\O{nd}$ runtime for subspace embeddings by \cite{ClarksonW13} and nearly matched, up to a logarithmic factor, by an offline algorithm that uses $\O{nd\log(n\Delta)}$ runtime for $(k,z)$-clustering by \cite{DraganovSS24}. 
We ask:
\begin{quote}
\vspace{-0.05in}
Is there an inherent \emph{time complexity} cost for data summarization problems, such as $(k,z)$-clustering and subspace embeddings, in the streaming model?
\vspace{-0.1in}
\end{quote}

\subsection{Our Contributions}
We present a unified technique that resolves the above questions for both $(k,z)$-clustering and subspace embeddings, showing there is no overhead for working in the streaming model for either space complexity or time complexity. 
At a high level, the technique is summarized in \figref{fig:outline}. 
However, efficiently implementing each of the steps is challenging and requires structural properties to each specific problem. 
We provide a technical overview of these steps in \secref{sec:overview}. 

\begin{figure*}[!htb]
\begin{mdframed}
\begin{enumerate}
\item
Given a stream $\calS$ of length $n$, use a crude sampling scheme to generate a stream $\calS'$ of length $n^{1-\Omega(1)}$. 
\item
Use $\calS'$ and a refined sampling scheme that takes longer to compute but generates a stream $\calS''$ of length $\polylog(n)$, omitting dependencies in other parameters. 
\item
Finally, run merge-and-reduce on $\calS''$ to produce a coreset $\calC$. 
\item
Produce an efficient encoding of $\calC$ using a constant-factor global encoding.
\end{enumerate}
\end{mdframed}
\vspace{-0.2in}
\caption{High-level summary of our approach}
\vspace{-0.2in}
\figlab{fig:outline}
\end{figure*}

\paragraph{Clustering.}
We describe the applications of our technique in \figref{fig:outline} in the context of $(k,z)$-clustering. 
Our first main result is the following: 
\vspace{-0.11in}
\begin{mdframed}[backgroundcolor=lightgray!40,topline=false,rightline=false,leftline=false,bottomline=false,innertopmargin=-4pt]
\vspace{-0.05in}
\begin{theorem}[Fast and space-optimal clustering]
\thmlab{thm:main}
Given a set $X$ of $n$ points on $[\Delta]^d$ and an accuracy parameter $\eps\in(0,1)$, there is a one-pass insertion-only streaming algorithm that uses $\tO{\frac{dk}{\min(\eps^4,\eps^{2+z})}}$ words of space and $d\log(k)\cdot\polylog(\log(n\Delta))$ amortized update time and outputs a $(1+\eps)$-strong coreset of $X$. 
\end{theorem}
\vspace{-0.05in}
\end{mdframed}
\vspace{-0.11in}
\thmref{thm:main} has a number of implications. 
First, the amortized runtime of \thmref{thm:main} is an \emph{exponential} improvement for dependencies in $k$ and $\log n$ over existing work~\cite{HenzingerK20,BhattacharyaCLP23,BhattacharyaCGL24}. 
Indeed, the fastest existing algorithm due to \cite{BhattacharyaCLP23,BhattacharyaCGL24} uses amortized update time $k\cdot\polylog(n)$ for constant $d$, whereas we use $\log(k)\cdot\polylog(\log(n\Delta))$ amortized update time. 
Thus, \thmref{thm:main} is the first result to show that $(k,z)$-clustering can be performed in the insertion-only streaming model with amortized update time sublinear in the number $k$ of clusters. 

In general, it is known that $\Omega(nk)$ runtime is necessary for even a constant-factor approximation to the metric $(k,z)$-clustering problem~\cite{BhattacharyaCLP23}, which would imply an $\Omega(k)$ lower bound for update time amortized over the stream. 
However, inputs to data streams require bounded bit precision and thus are intrinsically more appropriate for Euclidean $(k,z)$-clustering, where it was previously unclear whether the $\Omega(nk)$ total runtime lower bounds held. 
Our result in \thmref{thm:main} shows that they do not. 

\begin{figure}[!htb]
\centering
\small
{
\tabulinesep=1.1mm
\begin{tabu}{|c|c|}\hline
Streaming algorithm & Amortized Update Time \\\hline\hline
\cite{HenzingerK20} & $k^2\cdot\polylog(n\Delta)$ \\\hline\hline
\cite{BhattacharyaCLP23,BhattacharyaCF24} & $k\cdot\polylog(n\Delta)$ \\\hline\hline
\thmref{thm:main} (this work) & $\log(k)\cdot\polylog(\log(n\Delta))$ \\\hline
\end{tabu}
}
\vspace{-0.05in}
\caption{Table of $(k,z)$-clustering algorithms on data streams, omitting linear dependencies in the dimension $d$. We remark that \cite{HenzingerK20,BhattacharyaCLP23,BhattacharyaCF24} can handle the fully-dynamic setting, whereas ours cannot. 
However, our algorithm uses sublinear space while theirs does not.}
\vspace{-0.15in}
\figlab{fig:time:summary}
\end{figure}

Additionally, \thmref{thm:main} is the first result to achieve $(k,z)$-clustering on insertion-only streams using $\mathcal{O}_{k,d,\eps}(1)$ words of space, i.e., space usage independent of $n$. 
Despite a long line of work on clustering in the streaming model~\cite{Har-PeledM04,Har-PeledK07,Chen09,FeldmanL11,BravermanFLR19,stoc/BecchettiBC0S19,WoodruffZZ23}, the majority of previous results achieved space overheads that were polylogarithmic in $n$, compared to the offline setting. 
In a recent breakthrough work, \cite{Cohen-AddadWZ23} gave a streaming algorithm for $(k,z)$-clustering that uses $o_{k,d,\eps}(\log(n\Delta))$ words of space. 
However, their algorithm requires exponential time to process each stream item. 
Hence, the question remains whether we can achieve space complexity independent of $n$ altogether or even if we can achieve efficient update time with space overhead that is even polylogarithmic in $n$. 
\thmref{thm:main} shows that we can simultaneously achieve both, i.e., space complexity (in words) entirely independent of $n$, while also achieving $\log(k)\cdot\polylog(\log(n\Delta))$ amortized update time. 
In particular, our results match the best known \emph{offline} coreset constructions~\cite{Cohen-AddadSS21,Cohen-AddadLSS22} in terms of dependencies on $d$, $k$, and $\frac{1}{\eps}$ in the space, and thus match or improve upon the space usage of all previous streaming algorithms for insertion-only streams across all factors. 
See \thmref{thm:coreset:main} for a more formal statement and \figref{fig:space:summary} and \figref{fig:time:summary} for a comprehensive summary for the time and space complexity of existing results in the streaming model. 

\begin{figure}[!htb]
\centering
\small
{
\tabulinesep=1.1mm
\begin{tabu}{|c|c|}\hline
Streaming algorithm & Words of Memory \\\hline\hline
\cite{Har-PeledK07}, $z\in\{1,2\}$ & $\tO{\frac{dk^{1+z}}{\eps^{\O{d}}}\log^{d+z} n}$ \\\hline
\cite{Har-PeledM04}, $z\in\{1,2\}$ &
$\tO{\frac{dk}{\eps^d}\log^{2d+2} n}$ \\\hline
\cite{Chen09}, $z\in\{1,2\}$ &
$\tO{\frac{d^2k^2}{\eps^2}\log^8 n}$ \\\hline
\cite{FeldmanL11}, $z\in\{1,2\}$ &
$\tO{\frac{d^2k}{\eps^{2z}}\log^{1+2z} n}$ \\\hline
Sensitivity and rejection sampling~\cite{BravermanFLR19} & $\tO{\frac{d^2k^2}{\eps^2}\log n}$ \\\hline
Online sensitivity sampling & $\tO{\frac{d^2k^2}{\eps^2}\log^2 n}$ \\\hline
Merge-and-reduce with coreset of~\cite{Cohen-AddadLSS22} & $\tO{\frac{dk}{\min(\eps^4,\eps^{2+z})}\log^4 n}$ \\\hline
\cite{Cohen-AddadWZ23} & $\tO{\frac{dk}{\min(\eps^4,\eps^{2+z})}}\cdot\polylog(\log n)$ \\\hline\hline
\thmref{thm:main} (this work) & $\tO{\frac{dk}{\min(\eps^4,\eps^{2+z})}}$ \\\hline
\end{tabu}
}
\vspace{-0.05in}
\caption{Table of $(k,z)$-clustering algorithms on insertion-only streams. We summarize existing results with $z=\O{1}$, $\Delta=\poly(n)$, and the assumption that $k>\frac{1}{\eps^z}$ for the purpose of presentation.}
\vspace{-0.1in}
\figlab{fig:space:summary}
\end{figure}

Moreover, we note that \thmref{thm:main} also implies that we can efficiently compute the approximately optimal centers for $(k,z)$-clustering in the incremental setting as well. 
In particular, given an insertion-only stream of $n$ points that defines a dataset $X$ on $[\Delta]^d$, our one-pass streaming algorithm uses $d\log(k)\cdot\polylog(\log(n\Delta))$ amortized update time and $\tO{dk\log(n\Delta)}$ bits of space, and outputs an $\O{z}$-approximation to $(k,z)$-clustering at all times in the stream. 
We show this in \thmref{thm:cluster:main}. 
Additionally, \thmref{thm:main} surprisingly even gives the best known runtime for $(k,z)$-clustering in the \emph{offline} setting. 
\vspace{-0.05in}
\begin{mdframed}[backgroundcolor=lightgray!40,topline=false,rightline=false,leftline=false,bottomline=false,innertopmargin=-4pt]
\vspace{-0.05in}
\begin{restatable}{theorem}{thmofflineruntime}
\thmlab{thm:offline:runtime}
Given an dataset $X$ of $n$ in $[\Delta]^d$, there is an algorithm that uses $nd\log(k)\cdot\polylog(\log(n\Delta))$ runtime and outputs an $\O{z}$-approximation to the $(k,z)$-clustering problem. 
\end{restatable}
\vspace{-0.05in}
\end{mdframed}
\vspace{-0.05in}
Previously, the best existing $(k,z)$-clustering algorithms use $(nd+nk)\cdot\polylog(n\Delta)$ runtime, by using bicriteria approximations~\cite{ArthurV07} in conjunction with dimensionality reduction~\cite{CohenEMMP15,stoc/BecchettiBC0S19,MakarychevMR19,IzzoSZ21} and importance sampling~\cite{HuangV20,Cohen-AddadSS21,Cohen-AddadLSS22} to construct a coreset, which is then given as input to any polynomial-time approximation algorithm such as local search~\cite{GuptaT08}. 
By comparison, \thmref{thm:offline:runtime} has exponentially better dependencies on both $k$ and $\log(n\Delta)$. 

In concurrent and independent works, techniques by \cite{DraganovSS24,TourHS24} can be used to achieve offline algorithms for $(k,z)$-clustering that use $\O{nd\log(n\Delta k)}$ runtime, while \cite{DupreS24} achieved an algorithm that uses $\O{n^{1+\Omega(1)}d\log(n\Delta k)}$ runtime, for the setting of $\Delta=\poly(n)$. 
In the same spirit as \thmref{thm:offline:runtime}, these works observed that the $\Omega(nk)$ lower bound of \cite{BhattacharyaCLP23} is not necessarily for Euclidean clustering. 
However, our result is faster across all settings. 

Finally, we remark on the implications of our techniques to communication complexity. 
Our approach in \thmref{thm:main} uses a black-box reduction that utilizes an efficient encoding for any coreset construction. 
For instance, if $k<\frac{1}{\eps^2}$, it may be desirable to use existing coreset constructions of size $\tO{\frac{k^2}{\eps^2}}$ rather than $\tO{\frac{k}{\eps^{z+2}}}$; our algorithm transitions to this case smoothly. 
In fact, our algorithm actually uses $\tO{dk\log(n\Delta)}+\frac{dk}{\min(\eps^4,\eps^{2+z})}\cdot\polylog\left(k,\frac{1}{\eps},\log(n\Delta)\right)$ bits of space, i.e., the $\frac{1}{\eps}$ factors do not multiply the $\log(n\Delta)$ factors. 
Therefore, our results are near-optimal with a recent lower bound on the communication complexity of clustering~\cite{zhu2024space}. 

\vspace{-0.1in}
\paragraph{Subspace embeddings.}
We next show the applications of our technique in \figref{fig:outline} to the problem of $L_p$ subspace embeddings. 
In addition to structural results about efficient encodings for $L_p$ subspace embeddings, our main result is the following:
\vspace{-0.1in}
\begin{mdframed}[backgroundcolor=lightgray!40,topline=false,rightline=false,leftline=false,bottomline=false,innertopmargin=-4pt]
\vspace{-0.05in}
\begin{theorem}
\thmlab{thm:lp:inf}
Given a matrix $\bA$ of $n$ rows on $[-M,\ldots,-1,0,1,\ldots,M]^d$, with $M=\poly(n)$ and online condition number $\kappa$, there exists a one-pass streaming algorithm in the row arrival model that outputs a $(1+\eps)$-strong coreset of $\bA$ and amortized runtime is $\O{d}$ per update. 
For $p\in[1,2]$, the algorithm uses $\tO{\frac{d^2}{\eps^2}}$ words of space, while for $p>2$, the algorithm uses $\tO{\frac{d^{p/2+1}}{\eps^2}}$ words of space.
\end{theorem}
\vspace{-0.05in}
\end{mdframed}
\vspace{-0.1in}
We first emphasize that the amortized runtime of \thmref{thm:lp:inf} translates to an offline algorithm with $\O{nd}$ runtime. 
Since any offline algorithm requires $\O{nd}$ time to read the input, without additional assumptions, our result is tight and shows that there is no inherent time complexity separation between the offline setting and the streaming model for $L_p$ subspace embeddings. 

We next remark on the differing behaviors across the ranges of $p$. 
This is not a coincidence, as \cite{LiWW21} showed that any constant-factor $L_p$ subspace embedding for $p\ge 2$ requires $\Omega(d^{\max(p/2,1)}+1)$ bits. 
Thus, our result has tight dependencies in terms of the dimension $d$. 

Moreover, existing algorithms based on using merge-and-reduce with the optimal coreset constructions for $L_p$ subspace embeddings require $\poly(d,\log n)$ bits of space, e.g., $\tO{\frac{d^2}{\eps^2}}\cdot\polylog(n)$ for $p\le 2$~\cite{DrineasMM06a,DrineasMM06b,Sarlos06} and $\tO{\frac{d^{p/2+1}}{\eps^2}}\cdot\polylog(n)$ for $p> 2$~\cite{CohenP15,WoodruffY23}.
Similarly, schemes based on online sampling also require $\poly(d,\log n)$ bits of space, e.g., $\tO{\frac{d^2}{\eps^2}}\cdot\polylog(n)$ for $p\le 2$~\cite{CohenMP20} and $\tO{\frac{d^{p/2+1}}{\eps^2}}\cdot\polylog(n)$ for $p>2$~\cite{CohenP15,BravermanDMMUWZ20,WoodruffY23}. 
Hence, it was unknown whether there is an inherent separation between the offline setting and the streaming model for coreset constructions for $L_p$ subspace embeddings. 
\thmref{thm:lp:inf} resolves this question, showing that there is no overhead for the streaming model -- we can perform $L_p$ subspace embedding in $\mathcal{O}_{d,\eps}(1)$ bits of space. 
See \figref{fig:lp:space:summary} and \thmref{thn:lp:main} for more details. 

\begin{figure}[!htb]
\centering
\small
{
\tabulinesep=1.1mm
\begin{tabu}{|c|c|}\hline
Streaming algorithm & Words of Memory \\\hline\hline
Merge-and-reduce with coreset of~\cite{DrineasMM06a,DrineasMM06b,Sarlos06}, $p\le 2$ & $\tO{\frac{d^2}{\eps^2}}\cdot\polylog(n)$ \\\hline
Online leverage score sampling \cite{CohenMP20}, $p=2$ & $\tO{\frac{d^2}{\eps^2}\log n}$ \\\hline
Online sensitivity sampling \cite{BravermanDMMUWZ20}, $p\le2$ & $\tO{\frac{d^3}{\eps^2}\log n}$ \\\hline
Online Lewis weight sampling \cite{WoodruffY23}, $p\le2$ & $\tO{\frac{d^2}{\eps^2}\log n}$ \\\hline
\thmref{thm:main} (this work), $p\le 2$ & $\tO{\frac{d^2}{\eps^2}}$ \\\hline
\hline
Merge-and-reduce with coreset of~\cite{DasguptaDHKM09}, $p\ge 1$ &$\tO{\frac{d^{\max(p/2+1,p)+2}}{\eps^2}\log^3 n}$ \\\hline
Merge-and-reduce with coreset of~\cite{CohenP15}, $p\ge 2$ &$\tO{\frac{d^{p/2+1}}{\poly(\eps)}}\cdot\polylog(n)$ \\\hline
Online Lewis weight sampling \cite{WoodruffY23}, $p>2$ & $\tO{\frac{d^{p/2+1}}{\eps^2}}\cdot\polylog(n)$ \\\hline
\thmref{thm:main} (this work), $p>2$ & $\tO{\frac{d^{p/2+1}}{\eps^2}}$ \\\hline
\end{tabu}
}
\vspace{-0.05in}
\caption{Table of $L_p$ subspace embedding algorithms on insertion-only streams. We summarize existing results with $\kappa=\poly(n)$ for the purpose of presentation.}
\vspace{-0.15in}
\figlab{fig:lp:space:summary}
\end{figure}

\subsection{Technical Overview}
\seclab{sec:overview}
In this section, we provide a high-level intuition behind our framework and how to efficiently implement each step in \figref{fig:outline}. 
Recall that given a stream $\calS$ of length $n$, we first use a crude sampling scheme to generate a stream $\calS'$ of length $n^{1-\Omega(1)}$. 
Then, we use a refined sampling scheme on $\calS'$, which provides a significantly more efficient compression of the input, but requires longer time to compute, resulting in a stream $\calS''$. 
However, because these operations are performed on the input stream $\calS'$ of length $o(n)$, then these slower operations can be amortized into lower-order terms that do not multiply a linear term in $n$. 
Finally, we run merge-and-reduce on $\calS''$ to produce a coreset $\calC$, which we store using an efficient encoding to optimize the final space complexity. 

\subsubsection{Fast Update Time for Clustering}
For any point $x$ in a fixed dataset $X$, the \emph{$(k,z)$-clustering sensitivity} of $x$  is defined by $s(x)=\max_{C\subset\mathbb{R}^d:|C|\le k}\frac{\Cost(x,C)}{\Cost(X,C)}$, where the cost function is a sum of the $z$-th power of the distances~\cite{FeldmanL11,FeldmanS12,BravermanFLSZ21,Cohen-AddadWZ23,WoodruffZZ23}. 
When the dataset is evolving, the \emph{online sensitivity} of $x_t$ is the $(k,z)$-clustering sensitivity of $x$ with respect to $X_t=\{x_1,\ldots,x_t\}$, so that it quantifies the ``importance'' of $x_t$ with respect to the points of the stream that have already arrived. 
A standard approach in streaming algorithms for $(k,z)$-clustering is to sample each point $x$ with probability proportional to its online sensitivity, resulting in $\poly\left(k,\frac{1}{\eps},\log(n\Delta)\right)$ samples. 

To efficiently implement \figref{fig:outline}, we first compute a crude but fast approximation to the sensitivity of each point. 
Note that since the total sensitivity sums to $\O{k}$, then even $n^\alpha$-approximation to the sensitivities, where $\alpha\in(0,1)$, will sample $n^\alpha\cdot\poly\left(k,\frac{1}{\eps},\log(n\Delta)\right)=o(n)$ points, forming the new insertion-only stream $\calS'$ of length $o(n)$. 
Unfortunately, it is not clear how to efficiently acquire even crude approximations to the online sensitivities; this is the main runtime bottleneck in achieving $o(k)$ amortized update time. 

\vspace{-0.1in}
\paragraph{$(k,z)$-clustering sensitivity to $(k,z)$-medoids sensitivity.}
To that end, we first define the $(k,z)$-medoids sensitivity of $x$ with respect to the dataset $X$ for $(k,z)$-clustering sensitivity by $\tau(x)=\max_{C\subset X:|C|\le k}\frac{\Cost(x,C)}{\Cost(X,C)}$; we define the online sensitivities analogously.  
Note that for the medoids formulation, the $k$ centers must be among the input set $X$, which for our purposes will be the coreset maintained by the algorithm, rather than the original input. 
Hence, we can assume $|X|=\poly\left(k,d,\log(n\Delta),\frac{1}{\eps}\right)$. 
We show that the $(k,z)$-medoids sensitivity $\tau(x)$ is a constant-factor approximation to the $(k,z)$-clustering sensitivity $s(x)$ as follows. 

First, note that the maximization of the ratio $\frac{\Cost(x,C)}{\Cost(X,C)}$ is over a smaller search space for the possible values of $C$ in the medoids formulation and thus $\tau(x)\le s(x)$.  
To show $s(x)\lesssim\tau(x)$, we first recall that the optimal $(k,z)$-clustering and $(k,z)$-medoids clustering costs are within a factor of $2^{\O{z}}$ due to the generalized triangle inequality. 
Hence if we fix $C$ to be a set of centers that maximizes $s(x)$, let $c$ be the closest center of $C$ to $x$ and let $R$ be the set of points of $X$ served by $c$, then it suffices to show that there exists $c'\in X$ such that $\frac{\Cost(x,c)}{\Cost(R,c)}\approx\frac{\Cost(x,c')}{\Cost(R,c')}$, as we can find a $2^{\O{z}}$-approximate clustering of $X\setminus R$ using $k-1$ other medoids.  

For $k$-medoids, we perform casework on the size of the set $P$ of points that are closer to $x$ than to $c$. 
If $|P|<0.99|R|$, then by a simple Markov-type argument, we can show there exists some $y\in R\setminus P$ such that $\Cost(y,c)\le\frac{100}{r}\cdot\Cost(R,c)$, but since $y\notin P$, then $y$ is closer to $c$ than $y$ is to $x$. 
Thus by setting $y$ to be a $k$-medoids center, it follows from the triangle inequality that the cost of clustering $R$ with $y$ instead of $c$ cannot change by much. Moreover, the distance between $x$ and $y$ is similar to the distance between $x$ and $c$, so that $\tau(x)\approx s(x)$. 
On the other hand, if $|P|\ge 0.99|R|$, then there is a large number of points which are closer to $x$ than to $c$ by definition of $P$ and hence $\frac{\Cost(x,c)}{\Cost(R,c)}=\O{1}{|R|}$. 
However, we can choose the point $c'$ of $P$ farthest away from $x$, so that $\frac{\Cost(x,c')}{\Cost(R,c')}\ge\frac{1}{2|R|}$ and thus $\tau(x)\approx s(x)$ again. 

\vspace{-0.1in}
\paragraph{Constrained $(k,z)$-medoids clustering with a fixed center.}
We next show how to approximately solve the constrained medoids clustering problem $\min\Cost(X,C)$ across all sets $C\subset X$ of $k$ centers containing a center at $p$. 
Suppose $C$ is the optimal such constrained clustering. 
Let $S$ be a constant-factor approximation to the optimal unconstrained $(k,z)$-medoids clustering on $X$, which can be efficiently computed by local search~\cite{GuptaT08}. 
We show that the best clustering obtainable from swapping a center in $S$ with $\{p\}$ is a good approximation to $C$. 

Specifically, let $W$ be the set of $k-1$ centers $W\subseteq S$ such that $W\cup\{p\}$ has the minimum $(k,z)$-medoids clustering cost on $X$, and let $C'=W\cup\{p\}$. 
Let $Q$ be the set of at most $k$ centers of $S\cup\{p\}$ consisting of the nearest neighbors to $C$.  
For each $y\in X$, let $\pi_{C}(y)$ be the center in $C$ closest to $y$ 
Then we can use the triangle inequality to charge $\Cost(y,Q)$ in terms of $\Cost(y,\pi_C(y))+\Cost(\pi_C(y),Q)$. 
Since $\Cost(\pi_C(y),Q)\le\Cost(\pi_C(y),S\cup\{p\})\le\Cost(\pi_C(y),S)$, we can use the triangle inequality to charge $\Cost(\pi_C(y),Q)$ to  $\Cost(y,C)+\Cost(y,S)$. 
Finally, because $C'$ is the $k$ centers among $S\cup\{p\}$, summing across all $y\in X$ we can upper bound $\Cost(X,C')\le\Cost(X,Q)$ in terms of $\Cost(X,C)$ and $\Cost(X,S)$ and hence, simply in terms of $\Cost(X,C)$. 

However, finding $C'$ requires finding $W$, the best set of $k-1$ centers $W\subseteq S$, which does not immediately seem like an easy algorithmic task. 
To that end, we approximate $\Cost(X,C')$ by first performing the following bookkeeping on $S$. 
For each center $s\in S$, we compute the number $n_s$ (or weight) of the points of $X$ served by $s$. 
We then compute the distance $r_s$ from $s$ to its closest center in $S\setminus\{s\}$ and show the clustering cost after removing $s$ would approximately increase by $n_s\cdot(r_s)^z$.

Next, we need to update this information for each query $p$ that we guess to serve $x$ in the optimal constrained solution. 
Rather than updating the information for each center $s\in S$, we instead show that it suffices to update the information for just the nearest neighbor $u\in S$ to $p$. 
We then remove the center $s$ with the smallest $n_s\cdot(r_s)^z$ to find $W$ and thus $C'$. 
Finally, we show that $\Cost(X,S)+n_c\cdot(r_c)^z$ is both efficiently computable given our pre-processing information and also a good approximation to $\Cost(X,C')$.

\vspace{-0.1in}
\paragraph{Constrained $(k,z)$-medoids clustering with a fixed closest/serving center.}
Although $C'$ is a set of $k$ centers that contains our guess $p$ for the center that realizes the sensitivity of $x$, it may not hold that $p$ is the closest center in $C'$ to $x$. 
Thus we must further utilize $C'$ to approximately solve the constrained clustering problem $\min\Cost(X,C)$ across all sets $C\subset X$ of $k$ centers containing a center at $p$ \emph{and no other centers that are closer to $x$ than $p$}. 
Let $C$ be the optimal solution to this constrained problem and let $r:=\|x-p\|_2$. 
Our first observation is that any points served by some center $c\in C'$ in the ball $B_{r/2}(x)$ of radius $\frac{r}{2}$ around $x$ must be served in $C$ by a center outside the interior of the ball $B_r(x)$ of radius $r$ around $x$. 
In fact, if $n_b$ is the number (or weight) of points served by centers of $C'$ in $B_{r/2}(x)$, then we would expect these points to incur additional cost $n_b\cdot r^z$ in $\Cost(X,C)$. 

Unfortunately, it is not true that $\Cost(X,C')+n_b\cdot r^z$ is a good approximation to $\Cost(X,C)$ because there may be a large number of points served by a single center $q$ at some distance $d_q\in\left[\frac{r}{2},r\right)$ from $x$.  
However, since we did not include these points in $n_b$, then they are effectively not moved to outside of $B_r(x)$. 
Hence, it is possible that $\Cost(X,C)$ is significantly larger than $\Cost(X,C')+n_b\cdot r^z$. 

The natural approach would be to include these points in the computation of $n_b$, so that perhaps we instead compute $n'_b$ to be the number of points served by centers of $C'$ inside $B_r(x)$ and then effectively move them to a center $B_r(x)$ via $\Cost(X,C')+n'_b\cdot r^z$. 
This still does not work because there can be a large number of points served by a center $q$ inside but arbitrarily close to the boundary of $B_r(x)$; in this case, $\Cost(X,C')+n'_b\cdot r^z$ is a significant over-estimate of $\Cost(X,C)$. 

The crucial observation is that we do not need to handle either of these cases because in both of these cases, there exists a different point $p'$ in the annulus between $B_r(x)$ and $B_{r/2}(x)$ that approximately realizes the $(k,z)$-medoids sensitivity of $x$, such that these bad cases do not hold for $p'$. 
In particular, these problematic centers would be outside the ball of radius $\|p'-x\|_2$ centered at $x$ and so the resulting clustering is valid for the constraint of $p'$ serving $x$. 

On the other hand, we must still ensure our approximation $\Cost(X,C')+n'_b\cdot r^z$ does not significantly overestimate the sensitivity. 
Thus if $n_a$ is number of points served by centers in the annulus between $B_r(x)$ and $B_{r/2}(x)$, we show that if $n_b\ge n_a$, then $\frac{r^z}{\Cost(X,C')}$ is a good estimate to $\frac{\Cost(x,C)}{\Cost(X,C)}$. 
Otherwise if $n_b<n_a$, then we show that $\frac{r^z}{\Cost(X,C')}$ is an overestimate of $\frac{\Cost(x,C)}{\Cost(X,C)}$ but still upper bounded by the desired $(k,z)$-medoids sensitivity $\tau(x)$, up to a constant factor.

\vspace{-0.1in}
\paragraph{Crude quadtree for crude approximations.}
Finally, to implement our crude sampling procedure in \figref{fig:outline}, we create a quadtree for the input set, e.g.~\cite{indyk120oct,BackursIRW16,Cohen-AddadLNSS20}. 
Traditionally, each cell of a level of a quadtree has half the side length of a cell of an adjacent level in the quadtree. 
However, because we only seek $n^\alpha$-approximation, we permit each cell of a level to have side length $\frac{1}{n^\alpha}$ fraction of the side length of the adjacent level. 
This significantly decreases the number of levels in the quadtree to a constant number of levels overall, provided that $\Delta=\poly(n)$, rather than $\O{\log\Delta}$ levels in the standard quadtree. 
We then run the same algorithm above searching up the levels of the quadtree rather than the radially outward ball. 

Specifically, we estimate the distance of the center serving the point $x$ whose sensitivity we wish to approximate. 
Each estimate of a possible distance corresponds to a separate level in the quadtree. 
For a fixed level $\beta$ in the quadtree, we first take a constant-factor approximation to the unconstrained $(k,z)$-clustering problem to find a crude-approximation to the optimal $(k,z)$-clustering constrained to having a center in the same cell as $x$ at level $\beta$. 
Our algorithm provides an estimate $\Psi$ of this cost by deleting the center that increases the overall cost the least, among the centers in the constant-factor solution adjoined with a center in the same cell as $x$ at level $\beta$. 
We then ensure that $x$ is served by a center at level $\beta$ by moving all centers in the same cell as $x$ before level $\beta$ ``upward'' until level $\beta$ by computing the number (or weight) $n_\beta$ of points served by these centers in $\Psi$ and again increasing $\Psi$ by $n_\beta\cdot r^z$, where $r$ is the distance estimated by the quadtree for a center serving $x$ at level $\beta$. 

\subsubsection{Optimal Space Clustering in the Streaming Model}
The two most common approaches for $(k,z)$-clustering on insertion-only streams are merge-and-reduce and online sensitivity sampling. 
Informally, merge-and-reduce partitions the stream into consecutive blocks, builds a binary tree on the sequence of blocks, and then maintains a coreset for the points corresponding to each node in the tree. 
On the other hand, the online sensitivity approach samples each point $x_t$ of the stream with probability proportional to its online sensitivity, a quantity that measures the importance of $x_t$ with respect to the points of the stream that have already arrived. 
However, it is known these approaches cannot immediately be used to achieve our goal of streaming algorithms with space (in words) independent of $n$. 
we describe these approaches and others, as well as their shortcomings in more detail in \secref{sec:overview:space}.  

Recent insight by \cite{Cohen-AddadWZ23} observed that merge-and-reduce requires space that is polylogarithmic in the length of the input stream, while running online sensitivity sampling induces an insertion-only stream $\calS'$ of length $\poly\left(k,d,\log(n\Delta),\frac{1}{\eps}\right)$ of weighted points that forms a $(1+\eps)$-coreset of the input points. 
By running merge-and-reduce on $\calS'$, we acquire a $(1+\O{\eps})$-coreset for the original stream using an algorithm that stores $\poly_{k,d,\eps}(\log\log(n\Delta))$ weighted points. 

\vspace{-0.1in}
\paragraph{Efficient encoding of coreset points.}
To achieve our goal of $\mathcal{O}_{k,d,\eps}(1)$ space (in words), we thus need a more efficient encoding of each point, rather than representing each point using $\O{d\log(n\Delta)}$ bits of space. 
Suppose we have a set $C'$ of $\O{k}$ centers that is a constant-factor approximation to the optimal $(k,z)$-clustering on $X$. 
We can noisily encode $X$ by writing each $x\in X$ as $x=\pi_{C'}(x)+(x-\pi_{C'}(x))$, where $\pi_{C'}(x)$ is the closest center of $C'$ to $x$. 
We then round each coordinate of $x-\pi_{C'}(x)$ by rounding each coordinate to a power of $(1+\eps')$ for $\eps'=\poly\left(\eps,\frac{1}{d},\frac{1}{\log(n\Delta)}\right)$ to form a vector $y'$. 
Hence, given $C'$, each vector $x'=\pi_{C'}(x)+y'$ can be encoded using $\O{\log k+d\log\frac{1}{\eps'}}$, since we can store the \emph{exponents} of the offsets. 
It is known~\cite{Cohen-AddadWZ23,zhu2024space} that if $X'$ is the set of all points of $X$ rounded in this manner, then $X'$ is a $(1+\eps)$-coreset for $X$. 
Moreover, we can round the weights of $X'$ to powers of $(1+\eps')$, thereby efficiently encoding each coreset in the merge-and-reduce framework to improve the overall space complexity of the problem. 

Unfortunately, storing $C'$ itself uses $\O{kd\log n}$ bits and the merge-and-reduce tree on $\calS'$ has height $\O{\log_{k,d,\eps}(|S|)}=\polylog_{k,d,\eps}(\log(n\Delta))$ and therefore requires at least the same number of coresets to be simultaneous stored. 
Thus the total space remains $\O{kd\log n}\cdot\polylog_{k,d,\eps}(\log(n\Delta))$. 

\vspace{-0.1in}
\paragraph{Global encoding.}
To overcome this issue, we provide a single constant-factor approximation to the global dataset, rather than providing a constant-factor approximation for each level of the merge-and-reduce tree. 
Namely, we observe that whenever we need to do a merge operation at some time $t$, we use the stored coreset to recompute a constant-factor approximation $C'$ to $X_t$. 
For each coreset for a subset $S_v$ of data points representing a node $v$ in the tree, we can then perform the efficient encoding with respect to $C'$ instead of a separate constant-factor approximation. 
Although each encoding no longer guarantees $(1+\eps)$-multiplicative approximation to $\Cost(S_v,C)$ for a query set $C$ of $k$ centers, it still guarantees $\eps'\cdot(X_t,C)$ additive error to $\Cost(S_v,C)$, which suffices for a $(1+\eps)$-approximation to $\Cost(X_t,C)$ after summing across all subsets $S_v$ at time $t$.  

\subsubsection{Subspace Embeddings}
Finally, we briefly describe how our framework in \figref{fig:outline} can achieve fast space-efficient algorithms for subspace embeddings. 
However, it is not immediately obvious how to utilize a constant-factor subspace embedding $\bM\in\mathbb{R}^{m\times d}$ of the input matrix $\bA\in\mathbb{R}^{n\times d}$ to achieve an efficient encoding, especially since the dimensions are different. 
One possible approach is to project $\bA$ onto $\bM$ and round the resulting rows, but the rounding process could result in large error to due cancellation of rounded entries. 
For example, suppose that after the projection, there is a row $\br=(A, B)$ for some large values of $A,B>0$. 
By rounding each coordinate of $\br$ to the nearest power of $(1+\eps')$ for some $\eps'=\poly\left(\frac{\eps}{d,\log(nM)}\right)$, we obtain a row $\br'=(A',B')$. 
Now for $\bv=(B,-A)$, we have $\langle\bv,\br\rangle=0$, but $\langle\bv,\br'\rangle$ can be as large as $\eps'\cdot(|A|+|B|)$. 
Moreover, this error can compound across all rows of the matrix $\bA$ so that the resulting additive error to the estimate $\|\bA\bv\|_p^p$ could be as large as $\eps'^p\cdot\|\bA\|_p^p$. 

We instead observe that the correct procedure is to first multiply the matrix $\bA$ with a preconditioner so that none of its rows can contribute a large amount to the error. 
To that end, we first use a well-conditioned basis to compute a preconditioner $\bP\in\mathbb{R}^{d\times d}$ and multiply each row of $\bA$, so that for all vectors $\bx\in\mathbb{R}^d$ with $\|\bA\bx\|_p=1$, we have that $\frac{1}{\poly(d)}\le\|\bA\bP\bx\|_p\le\poly(d)$. 
In particular, the maximum possible additive error $\langle\bv,\br'\rangle$ for each row $\br$ can be charged to the sensitivity of $\br$, and it follows that the sum of the errors due to the rounding can be at most $\eps'\cdot\poly(d)$. 
Since $\eps'=\poly\left(\frac{\eps}{d,\log(nM)}\right)$, then the overall error is $\O{\eps}$, which achieves a subspace embedding because $\|\bA\bx\|_p=1$. 
For a row $\bb_t=\ba_t\bP$, we then round each entry of $\bb_t$ to the nearest power of $(1+\eps')$ and store the exponent as before. 
We can then also achieve a streaming algorithm by compressing the entire merge-and-reduce tree, as before. 

\vspace{-0.1in}
\paragraph{Fast runtime.}
Finally, we implement the approach in \figref{fig:outline} to achieve fast runtime. 
We first produce crude but fast $n^\alpha$-approximations to the $L_p$ sensitivities using the \emph{root} leverage scores. 
These can be quickly produced using a quadratic form of the constant-factor approximation that we maintain at each time. 
In particular, given a constant-factor subspace embedding $\bB$ to $\bA$, the leverage score of a row $\ba_i$ with respect to $\bA$ is a constant factor multiple of $\ba_i^\top(\bB^\top\bB)^{-1}\ba_i$. 
Now we can approximate this quantity using $\|\bg\bZ\ba_i\|_2^2$, where $\bg$ is a random Gaussian vector and $\bZ=(\bB^\top\bB)^{-1/2}$. 
To relate this quantity to the $L_p$ sensitivities, we show that the square root of the leverage score is within $\poly(n)$ factors of the true $L_p$ sensitivities. 
Thus by sampling each row with the crude approximations, we then induce a stream of length $o(n)$ for which we can amortize existing approaches that use $\poly(d)$ update time. 
For more details, see the intuition in \secref{sec:subspace:embed}. 

\subsection{Preliminaries}
\seclab{sec:prelims}
Given an integer $n>0$, we use the notation $[n]$ to represent the set $\{1,\ldots,n\}$. 
We use $\poly(n)$ to represent a fixed polynomial in $n$ and $\polylog(n)$ to represent $\poly(\log n)$. 
We say an event occurs with high probability if it occurs with probability at least $1-\frac{1}{\poly(n)}$. 

In this paper, we will focus on the Euclidean clustering, as opposed to inputs from a general metric space. 
As is often standard in the streaming literature, we assume $\Delta=\poly(n)$ and thus we allow each word of space to use $\Theta(\log(nd\Delta))$ bits of storage, so that we can store the weight and each coordinate of each point using $\O{1}$ words of space. 
Thus for vectors $x,y\in\mathbb{R}^d$, we use $\dist(x,y)$ to represent the Euclidean distance $\|x-y\|_2$, so that $\|x-y\|_2^2=\sum_{i=1}^d(x_i-y_i)^2$. 
More generally, given a point $x$ and a set $S$, we abuse notation by representing $\dist(x,S):=\min_{s\in S}\dist(x,s)$. 
We also recall the $L_z$ norm of $x$ is defined by $\|x\|_z$, where $\|x\|_z^z=\sum_{i=1}^d x_i^z$. 
For a matrix $\bA\in\mathbb{R}^{n\times d}$, we use $\|\bA\|_F$ to denote the Frobenius norm, so that 
\[\|\bA\|_F^2=\sum_{i\in[n]}{j\in[d]}A_{i,j}^2.\]
For a fixed $z\ge 1$ and sets $X,C\subset\mathbb{R}^d$ with $X=\{x_1,\ldots,x_n\}$ we use $\Cost(X,C)$ to denote $\sum_{i=1}^n\dist(x_i,C)^z$. 

We first recall the generalized triangle inequality:
\begin{fact}[Generalized triangle inequality]
\factlab{fact:triangle}
For any $z\ge 1$ and $x,y,z\in\mathbb{R}^d$, we have
\[\dist(x,y)^z\le 2^{z-1}(\dist(x,w)^z+\dist(w,y)^z).\]
\end{fact}

Next, we recall the definition of a strong coreset for $(k,z)$-clustering.  
\begin{definition}[Coreset]
\deflab{def:coreset}
Given an approximation parameter $\eps>0$, and a set $X$ of points $x_1,\ldots,x_n\in\mathbb{R}^d$ with distance function $\dist$, a \emph{coreset} for $(k,z)$ clustering is a set $S$ with weight function $w$ such that for any set $C$ of $k$ points, we have
\[(1-\eps)\sum_{t=1}^n\dist(x_t,C)^z\le\sum_{q\in S}w(q)\dist(q,S)^z\le(1+\eps)\sum_{t=1}^n\dist(x_t,C)^z.\]
\end{definition}
We use the following coreset construction for $(k,z)$-clustering:
\begin{theorem}\cite{Cohen-AddadLSS22,huang2023optimal,nips/Cohen-AddadLSSS22}
\thmlab{thm:offline:coreset:size}
Given an accuracy parameter $\eps\in(0,1)$, there exists a coreset construction for Euclidean $(k,z)$-clustering that samples $\tO{\min\left(\frac{1}{\eps^2}\cdot k^{2-\frac{z}{z+2}},\frac{1}{\min(\eps^4,\eps^{2+z})}\cdot k\right)}$ weighted points of the input dataset. 
\end{theorem}

We next recall the standard Johnson-Lindenstrauss transformation:
\begin{theorem}[Johnson-Lindenstrauss lemma]
\cite{johnson1984extensions}
\thmlab{thm:jl}
Let $\eps\in\left(0,\frac{1}{2}\right)$ and $m=\O{\frac{1}{\eps^2}\log n}$. 
Let $X\subset\mathbb{R}^d$ be a set of $n$ points. 
There exists a family of random linear maps $\Pi:\mathbb{R}^d\to\mathbb{R}^m$ such that with high probability over the choice of $\pi\sim\Pi$, 
\[(1-\eps)\|x-y\|_2\le\|\pi x-\pi y\|_2\le(1+\eps)\|x-y\|_2,\]
for all $x,y\in X$. 
\end{theorem}

We next recall the following standard concentration inequality:
\begin{theorem}[Hoeffding's inequality]
\thmlab{thm:hoeffding's:inequality}
Let $X_1, \cdots, X_n$ be independent random variables such that $a_i \leq X_i \leq b_i$, and let $S_n = \sum_{i=1}^n X_i$. 
Then
\[\PPr{|S_n-\Ex{S_n}|>t}\le2\exp\left(-\frac{t^2}{\sum_{i=1}^n(b_i-a_i)^2}\right).\]
\end{theorem}

\section{Clustering}
In this section, we describe the efficient encoding for coresets for $(k,z)$-clustering, as well as a global encoding that be utilized to achieve a one-pass streaming algorithm on insertion-only streams for $(k,z)$-clustering using $\mathcal{O}_{k,d,\eps}(1)$ words of space. 
We begin with a number of preliminaries. 

We first formally define the sensitivity and online sensitivity of each point $x$ in a dataset $X$ for $(k,z)$-clustering. 
\begin{definition}[Sensitivities for $(k,z)$-clustering]
\deflab{def:sens:cluster}
The \emph{sensitivity} of a point $x\in X$ for $(k,z)$-clustering in a metric space equipped with metric $\dist$ is 
\[\max_{C:|C|\le k}\frac{\Cost(x,C)}{\Cost(X,C)}=\max_{C:|C|\le k}\frac{\dist(x,C)^z}{\sum_{x\in X}\dist(x,C)^z}.\]
\end{definition}

\begin{definition}[Online sensitivity for $(k,z)$-clustering]
\deflab{def:online:sens:cluster}
Let $x_1,\ldots,x_n$ be a sequence of points with metric $\dist$ and define $X_t:=\{x_1,\ldots,x_t\}$ for all $t\in[n]$. 
The \emph{online sensitivity} of $x_t$ for $(k,z)$-clustering is
\[\max_{C:|C|\le k}\frac{\Cost(x_t,C)}{\Cost(X_t,C)}=\max_{C:|C|\le k}\frac{\dist(x_t,C)^z}{\sum_{i=1}^t\dist(x_i,C)^z}.\]
\end{definition}
We have the following upper bound on the sum of the online sensitivities, i.e., the total online sensitivity. 
\begin{theorem}
\thmlab{thm:total:online:sens}
\cite{Cohen-AddadWZ23}
Let $X=\{x_1,\ldots,x_n\}\subset[\Delta]^d$ be a sequence of $n$ points and let $\sigma(x_t)$ denote the online sensitivity of $x_t$ for $t\in[n]$ for $(k,z)$-clustering, where $z\ge 1$. 
Then
\[\sum_{t=1}^n\sigma(x_t)=\O{2^{2z}k\log^2(nd\Delta)}.\]
\end{theorem}
We next recall the guarantees of online sensitivity sampling, in terms of both correctness and sample complexity. 
\begin{theorem}[Online sensitivity sampling]
\thmlab{thm:online:sens:cluster}
\cite{Cohen-AddadWZ23}
Given a sequence $x_1,\ldots,x_n$ of points, suppose each point $x_t$ is sampled with probability $p_t\ge\min(1,\gamma\cdot\sigma(x_t))$, where $\sigma_t$ is the online sensitivity of $x_t$ and $\gamma=\O{\frac{dk}{\eps^2}\log\frac{n\Delta}{\eps}}$ and weighted $\frac{1}{p_t}$ if $x_t$ is sampled. 
Then with high probability, the weighted sample is a $(1+\eps)$-strong coreset for $(k,z)$-clustering that contains at most $\O{\frac{dk^2}{\eps^2}\log^3\frac{n\Delta}{\eps}}$ points. 
\end{theorem}

\subsection{Background for Optimal Space Clustering in the Streaming Model}
\seclab{sec:overview:space}
We first describe common approaches for clustering in the streaming setting. 
Even though they incur logarithmic overheads in the size $n$ of the dataset compared to offline coreset constructions, their intuition will nevertheless be helpful towards our main algorithm. 

\paragraph{Merge-and-reduce.}
Merge-and-reduce~\cite{BentleyS80,Har-PeledM04} is a standard approach on insertion-only streams for $(k,z)$-clustering, as well as many other problems for which there exist coreset constructions. 
Given a dataset $X\subset[\Delta]^d$ of size $n$, the number $k$ of clusters, an accuracy parameter $\eps$, a failure probability $\delta$, suppose there exists a coreset construction algorithm for $(k,z)$-clustering that samples and reweights $f(n,d,k,\eps,\delta)$ points of $X$. 
The merge-and-reduce framework first partitions the stream into consecutive blocks of size $f(n,d,k,\eps',\delta')$, where $\eps'=\frac{\eps}{\O{\log n}}$ and $\delta'=\frac{\delta}{\poly(n)}$. 
It then creates a coreset with accuracy $(1+\eps')$ and failure probability $\delta'$ for each block of the stream. 
Each of these coresets uses space $f(n,d,k,\eps',\delta')$ and can be viewed as the leaves of a binary tree with height $\O{\log n}$. 
For each node in the tree at depth $\ell$, the merge-and-reduce framework then takes the points that are in the coresets in its children nodes at depth $\ell+1$ and constructs a coreset with accuracy $(1+\eps')$ for these points. 

Observe that the coreset at the root of the tree is a coreset for the entire dataset of the stream. 
Moreover, each node is the coreset of the coresets of its children nodes, so the entire merge-and-reduce process can be performed on-the-fly during the evolution of an insertion-only stream. 
Although each level of the merge-and-reduce tree induces a multiplication distortion of $(1+\eps')$, the accuracy of the root node is  $(1+\eps')^{\O{\log n}}=(1+\eps)$ due to the setting of $\eps'=\frac{\eps}{\O{\log n}}$. 
Unfortunately, the optimal coreset constructions sample $\tO{\frac{k}{\eps^2}}\cdot\min\left(k,\frac{1}{\min(\eps^2,\eps^z)}\right)$ points~\cite{Cohen-AddadSS21,Cohen-AddadLSS22} and in fact for certain regimes of $\eps$, coreset constructions for $(k,z)$-clustering provably require $\Omega\left(\frac{k}{\eps^{z+2}}\right)$ points~\cite{huang2023optimal,nips/Cohen-AddadLSSS22}. 
Therefore setting $\eps'=\frac{\eps}{\O{\log n}}$ in the merge-and-reduce approach would incur $\log n$ factors that are prohibitive for our goal. 

\paragraph{Offline sensitivity sampling.}
Another common approach is to adapt offline coreset constructions to the streaming model. 
In particular, recent efforts~\cite{Cohen-AddadWZ23,WoodruffZZ23} have been made to adapt the sensitivity sampling framework \cite{FeldmanL11,FeldmanS12,BravermanFLSZ21} to data streams. 
The \emph{sensitivity} of each point $x$ among a dataset $X\subset\mathbb{R}^d$ of $n$ points measures the ``importance'' of the point $x$ with respect to $X$, for the purposes of $(k,z)$-clustering and, as per \defref{def:sens:cluster}, is defined by
\[\max_{C\subset\mathbb{R}^d: |C|\le k}\frac{\Cost(x,C)}{\Cost(X,C)}=\max_{C\subset\mathbb{R}^d: |C|\le k}\frac{\dist(x,C)^z}{\sum_{x\in X}\dist(x,C)^z}.\]
There exist many variants of the sensitivity sampling framework, but the most relevant approach is sampling each point independently without replacement with probability proportional to its sensitivity, a process that can be shown to admit a coreset construction with high probability. 

In fact, the argument is relatively straightforward. 
The analysis first fixes a set $C$ of $k$ centers and shows that sensitivity sampling preserves the cost of clustering $X$ with $C$ in expectation. 
It then upper bounds the variance of the cost of clustering $C$ on the sampled points to show concentration. 
Although there can be arbitrary number of choices for $C$, for $k$ centers among $\mathbb{R}^d$, it suffices to only show correctness-of-approximation on a net of size $\left(\frac{n}{\eps}\right)^{\O{kd}}$, which can be handled by adjusting the probability of failure for each choice of $C$ and then applying a union bound. 

Unfortunately, the sensitivity of a point $x$ is defined with respect to the entire dataset $X$, which is unknown at the time of arrival of $x$ in a data stream, and thus the sensitivity of $x$ cannot be computed or even well-approximated at that point. 

\paragraph{Online sensitivity sampling.}
Instead, recent works have focused on \emph{online} sensitivity sampling, where the elements of the dataset $X$ are ordered by the time of their arrival in the data stream, so that $X=\{x_1,\ldots,x_n\}$, where $x_1$ is the first item in the stream and $x_n$ is the last item in the stream. 
The online sensitivity of a point $x_t$ is then, as per \defref{def:online:sens:cluster},  defined by 
\[\max_{C\subset\mathbb{R}^d: |C|\le k}\frac{\Cost(x_t,C)}{\Cost(X_t,C)}=\max_{C\subset\mathbb{R}^d: |C|\le k}\frac{\dist(x_t,C)^z}{\sum_{i=1}^t\dist(x_i,C)^z},\]
where $X_t$ is the subset consisting of the first $t$ points of $X$. 
Note that while the sensitivity of a point $x_t$ measures the importance of $x_t$ with respect to the entire dataset $X$, the online sensitivity of $x_t$ measures the importance of $x_t$ with respect to the prefix of size $t$, a quantity that can be computed at the time of arrival of $x_t$. 
In particular, since online sensitivity sampling produces a $(1+\eps)$-coreset of $X_{t-1}$, then the previously sampled points can be used to approximate the online sensitivity of $x_t$. 

The online sensitivity framework then samples each point with probability proportional to its online sensitivity. 
We remark that although this process seemingly induces dependencies in the analysis, a standard martingale and coupling argument shows correctness of online sensitivity sampling.  
Unfortunately, as a result of the sampling probability, the total number of points sampled is proportional to the sum of the online sensitivities of the points, which can be shown to be at least $\Omega(k\log(n\Delta))$. 
Thus online sensitivity sampling would again incur $\log n$ factors that are prohibitive for our goal.

\subsection{Efficient Encoding for Coreset Construction for \texorpdfstring{$(k,z)$}{(k,z)}-Clustering}
We now give our efficient encoding for a given coreset for $(k,z)$-clustering. 
Given a dataset $X$, which can be viewed as either the original input or a set of weighted points that forms a coreset of some underlying dataset, we first acquire a constant-factor approximation $C'$ for $(k,z)$-clustering on $X$. 
For each $x\in X$, let $\pi_{C'}(X)$ be the closest center of $C'$ to $x$. 
We then write each $x\in X$ as $x=\pi_{C'}(x)+(x-\pi_{C'}(x))$, thereby decomposing $x$ into the closest center $\pi_{C'}(x)$ and its offset $x-\pi_{C'}(x)$ from the center. 
We show that for the purposes of $(1+\eps)$-approximation for $(k,z)$-clustering, we can afford to round each coordinate of $x-\pi_{C'}(x)$ to a power of $(1+\eps')$, where $\eps'=\poly\left(\eps,\frac{1}{d},\frac{1}{\log(n\Delta)}\right)$, forming a vector $y'$. 
Finally, we can store $C'$ and thus encode the vector $x'=\pi_{C'}(x)+y'$ using $\O{\log k+d\log\left(\frac{1}{\eps},d,\log(n\Delta)\right)}$ bits, by storing the identity of $\pi_{C'}(x)$ and the exponent of the offset for each of the $d$ coordinates. 
We give the algorithm in full in \algref{alg:coreset:cluster}. 

\begin{algorithm}[!htb]
\caption{Efficient Encoding for Coreset Construction for $(k,z)$-Clustering}
\alglab{alg:coreset:cluster}
\begin{algorithmic}[1]
\Require{Data set $X\subset[\Delta]^d$ with weight $w(\cdot)$, accuracy parameter $\eps\in(0,1)$, number of clusters $k$, parameter $z\ge 1$, failure probability $\delta\in(0,1)$}
\Ensure{$(1+\eps)$-coreset for $(k,z)$-clustering}
\State{$\eps'\gets\frac{\poly(\eps^z)}{\poly(k,\log(nd\Delta)}$}
\State{Find a set $C'$ of $k$ centers that is a constant-factor approximation to $(k,z)$-clustering on $X$}
\For{each $x\in X$}
\State{Let $c'(x)$ be the closest center of $C'$ to $x$}
\State{Let $y'$ be the offset $x-c'(x)$}
\State{Let $y$ be $y'$ with coordinates rounded to a power of $(1+\eps')$}
\State{Let $x'=(c'(x),y)$, storing the \emph{exponent} for each coordinate of $y$}
\State{$X'\gets X'\cup\{x'\}$}
\EndFor
\State{\Return $(C',X')$}
\end{algorithmic}
\end{algorithm}

To show correctness of our encoding, we first recall the following fact.
\begin{fact}[Claim 5 in \cite{SohlerW18}]
\factlab{fact:gen:tri}
Suppose $z\ge 1$, $x,y\ge 0$, and $\eps\in(0,1]$.  
Then
\[(x+y)^z\le(1+\eps)\cdot x^z+\left(1+\frac{2z}{\eps}\right)^z\cdot y^z.\]
\end{fact}

Given a constant-factor approximation $C'$, we now show that the rounding $X'$ of the dataset $X$ by representing each point as a rounded offset from its closest center in $C'$ is a strong coreset of $X$. 
\begin{lemma}
\lemlab{lem:x:to:xprime:cluster}
Let $\eps\in\left(0,\frac{1}{2}\right)$ and let $X'$ be the weighted dataset $S$ defined by their offsets from the set $C'$ of centers from \algref{alg:coreset:cluster}. 
Then for all $C\subset[\Delta]^d$ with $|C|\le k$, 
\[(1-\eps)\cdot\Cost(C,X)\le\Cost(C,X')\le(1+\eps)\cdot\Cost(C,X).\]
\end{lemma}
\begin{proof}
Let $C$ be any fixed set of at most $k$ centers. 
We abuse notation so that for each $x\in X$, we use $x'$ to denote the corresponding point $c(x)+y$, where $y'=x-c'(x)$ is the offset, and $y$ is $y'$ with its coordinates rounded to a power of $(1+\eps')$. 
Similarly, we use $c'$ to denote $c(x)\in C'$. 
Thus, we have $\|x-x'\|_2\le\eps'\cdot\|c'-x\|_2$ for all $x\in X$. 

For ease of presentation, we first write $X'$ so that the weights are not rounded to a power of $(1+\eps')$. 
By the triangle inequality,
\[\Cost(C,X')=\sum_{x'\in X'}(\dist(x',C))^z\le\sum_{x'\in X'}(\dist(x',x)+\dist(x,C))^z.\]
Since $\dist(x,x')=\|x-x'\|_2\le\eps'\cdot\|c'-x\|_2=\eps'\cdot\dist(x,C')$, then
\begin{align*}
\Cost(C,X')&\le\sum_{x'\in X'}\left(\eps'\cdot\dist(x,C')+\dist(x,C)\right)^z.
\end{align*}
By \factref{fact:gen:tri},
\begin{align*}
\Cost(C,X')&\le\sum_{x'\in X'}\left(\left(1+\frac{\eps}{2}\right)\cdot(\dist(x,C))^z+\left(1+\frac{4z}{\eps}\right)^z\cdot(\eps'\cdot\dist(x,C'))^z\right).
\end{align*}
Thus,
\[\Cost(C,X')\le\left(1+\frac{\eps}{2}\right)\cdot\Cost(C,X)+\left(1+\frac{4z}{\eps}\right)^z(\eps')^z\cdot\Cost(C',X).\]
Since $C'$ is a constant-factor approximation to the optimal $(k,z)$-clustering of $X$, then there exists a constant $\gamma\ge 1$ such that $\Cost(C,X')\le\gamma\cdot\Cost(C,X)$. 
Hence, 
\[\Cost(C,X')\le\left(\left(1+\frac{\eps}{2}\right)+\left(1+\frac{4z}{\eps}\right)^z(\eps')^z\cdot\gamma\right)\cdot\Cost(C,X).\]
Then for $\eps'=\frac{\poly(\eps^z)}{\poly(k,\log(nd\Delta)}$, we have that
\[\Cost(C,X')\le\left(1+\frac{2}{3}\eps\right)\cdot\Cost(C,X),\]
which almost gives the right hand side of the inequality (due to the notation $X'$ representing the output prior to the rounding of the weights). 

Similarly, by triangle inequality and $\dist(x,x')\le\eps'\cdot\dist(x,C')$,
\begin{align*}
\Cost(C,X)&=\sum_{x\in X}(\dist(x,C))^z\\
&\le\sum_{x\in X}(\dist(x,x')+\dist(x',C))^z\\
&\le\sum_{x'\in X'}\left(\eps'\cdot\dist(x,C')+\dist(x',C)\right)^z.
\end{align*}
By \factref{fact:gen:tri},
\begin{align*}
\Cost(C,X)&\le\left(1+\frac{\eps}{2}\right)\cdot\Cost(C,X')+\left(1+\frac{4z}{\eps}\right)^z(\eps')^z\cdot\Cost(C',X)\\
&\le\left(1+\frac{\eps}{2}\right)\cdot\Cost(C,X')+\left(1+\frac{4z}{\eps}\right)^z(\eps')^z\cdot\gamma\cdot\Cost(C,X).
\end{align*}
Then for $\eps'=\frac{\poly(\eps^z)}{\poly(k,\log(nd\Delta)}$, we have that
\[\Cost(C,X)\le\left(1+\frac{\eps}{2}\right)\cdot\Cost(C,X')+\frac{\eps}{100}\cdot\Cost(C,X),\]
so that $\left(1-\frac{2}{3}\cdot\eps\right)\Cost(C,X)\le\Cost(C,X')$, which almost gives the left-hand side of the inequality (due to the notation $X'$ representing the output prior to the rounding of the weights). 

Finally, since the weights are rounded to a power of $(1+\eps')$, then the cost can only change by $(1+\eps')$. 
Since $\eps'=\frac{\poly(\eps^z)}{\poly(k,\log(nd\Delta)}$, then it follows that
\[(1-\eps)\cdot\Cost(C,X)\le\Cost(C,X')\le(1+\eps)\cdot\Cost(C,X).\]
\end{proof}
\noindent
We now give the full guarantees of \algref{alg:coreset:cluster}, which gives an efficient encoding of an input set $X$.  
\begin{lemma}
\lemlab{lem:efficient:coreset:cluster}
Let $X$ be a coreset construction with weights bounded by $[1,\poly(nd\Delta)]$. 
Then $X'$ is a $(1+\eps)$-strong coreset for $X$ that uses 
$\O{dk\log(n\Delta)}+|X|\cdot\polylog\left(k,\frac{1}{\eps},\log(nd\Delta),\log\frac{1}{\delta}\right)$ bits of space. 
\end{lemma}
\begin{proof}
\algref{alg:coreset:cluster} then outputs a set $(C',X')$ that encodes $X'$ using $C'$. 
By \lemref{lem:x:to:xprime:cluster}, we have that for all $C\subset[\Delta]^d$ with $|C|\le k$, 
\[(1-\eps)\cdot\Cost(C,X)\le\Cost(C,X')\le(1+\eps)\cdot\Cost(C,X).\]
Thus, $X'$ is a strong coreset for $X$. 
 
It remains to analyze the space complexity of $(C',X')$. 
We first require $C'$ to encode $S$. 
Since $C'$ is a set of $k$ centers, then $C'$ can be represented using $\O{dk\log(n\Delta)}$ bits of space. 
Additionally, each point $x'$ in $X'$ is encoded using $\O{d\log\left(\frac{1}{\eps}+\log n+\log\Delta\right)}$ bits of space due to storing the $d$ coordinates of the rounded offset $y$. 
In particular, we do not store the explicit coordinates, but rather their exponents. 
Moreover, we use $\log k$ bits to store the closest center $c'(x)$. 
Finally, each weight has weight $\poly(nd\Delta)$ and thus can be approximated to $(1+\eps)$-multiplicative factor by storing the exponent, using $\O{\log\log(nd\Delta)}$ bits of space. 
Therefore, the output of the algorithm is encoded using $\O{dk\log(n\Delta)}+|X|\cdot\polylog\left(k,\frac{1}{\eps},\log(nd\Delta),\log\frac{1}{\delta}\right)$ total bits of space.
\end{proof}

\subsection{Clustering in the Streaming Model}
In this section, we show how to use a global encoding combined with the efficient encoding from the previous section in order to achieve our main algorithm for $(k,z)$-clustering on insertion-only data streams. 

The main shortcoming of immediately applying the previous efficient encoding to each node of a merge-and-reduce tree is that the constant-factor approximation requires $\O{kd\log(n\Delta)}$ bits of storage per efficient encoding, and there can be $\polylog(\log(n\Delta))$ such efficient encodings due to the height of the merge-and-reduce tree on a data stream produced by online sensitivity sampling. 
Instead, we provide a single constant-factor approximation to the global dataset and show that the error of maintaining such a global constant-factor approximation does not compound too much over the course of the stream. 

The algorithm appears in full in \algref{alg:cluster:insert:stream}. 

\begin{algorithm}[!htb]
\caption{$(k,z)$-Clustering on Insertion-Only Stream}
\alglab{alg:cluster:insert:stream}
\begin{algorithmic}[1]
\Require{Data set $X=\{x_1,\ldots,x_n\}\subset\mathbb{R}^d$ that arrives as a data stream, $\eps\in(0,1)$, number of clusters $k$, parameter $z\ge 1$}
\Ensure{$(1+\eps)$-coreset for $(k,z)$-clustering}
\State{$Z\gets\emptyset$, $\lambda\gets\O{\frac{k}{\eps^2}\cdot\log k\log n}$}
\For{each $t\in[n]$}
\State{Let $\widehat{x_t}$ be the image of $x_t$ after applying a JL transform into $\O{\log n}$ dimensions}
\State{Use $Z\cup\{\widehat{x_t}\}$ to compute a $\O{1}$-approximation $\widehat{\sigma(x_t)}$ to the online sensitivity of $x_t$}
\State{$p(x_t)\gets\min(1,\lambda\cdot\widehat{\sigma(x_t)})$}
\State{With probability $p(x_t)$, sample $x_t$ into a stream $\calS'$ with weight $\frac{1}{p(x_t)}$}
\State{Let $Z$ be the running output of merge-and-reduce on $\calS'$ using the efficient encoding for coreset construction in \algref{alg:coreset:cluster} with a \emph{global} constant-factor approximation, for accuracy $\frac{\eps}{\poly(\log\log n\Delta)}$ and failure probability $\frac{1}{\poly\left(\frac{1}{\eps},\log(n\Delta)\right)}$} 
\EndFor
\State{\Return $Z$}
\end{algorithmic}
\end{algorithm}
We first show that given a partition $X_1\sqcup\ldots\sqcup X_m$ of $X$, we can apply \algref{alg:coreset:cluster} and although the output $X'_i$ from our efficient encoding will no longer necessarily be a strong coreset for each corresponding $X_i$, the resulting error will only be an additive $\eps'\cdot\Cost(X,C)$, which suffices for our purposes. 
Specifically, we will have $\eps'=\frac{\eps}{\poly(\log\log(nd\Delta)}$ and ultimately have $m=\poly(\log\log(nd\Delta))$. 
\begin{lemma}
\lemlab{lem:each:coreset:cluster}
Let $X_1\sqcup\ldots\sqcup X_m=X$ be a partition of $X$ for $m=\poly(\log\log(nd\Delta))$. 
Let $C'$ be a set of $k$ centers that is a set $C'$ of $k$ centers that is a constant-factor approximation to $(k,z)$-clustering on $X$. 
For each $i\in[m]$, let $X'_i$ be the set of rounded points corresponding to $X_i$ from \algref{alg:coreset:cluster}. 
Then for all $i\in[m]$ and all $C\subset[\Delta]^d$ with $|C|\le k$, 
\[|\Cost(X'_i,C)-\Cost(X_i,C)|\le\frac{\eps}{\poly(\log\log nd\Delta)}\cdot\Cost(X,C).\]
\end{lemma}
\begin{proof}
Note that by \lemref{lem:x:to:xprime:cluster} with accuracy $\frac{\eps}{\poly(\log\log nd\Delta)}$, we have
\[|\Cost(X',C)-\Cost(X,C)|\le\frac{\eps}{\poly(\log\log nd\Delta)}\cdot\Cost(X,C).\]
Since 
\[|\Cost(X'_i,C)-\Cost(X_i,C)|\le\sum_{i=1}^m|\Cost(X'_i,C)-\Cost(X_i,C)|\le\eps\cdot\Cost(X,C),\]
then the claim follows. 
\end{proof}

It remains to give the full guarantees of \algref{alg:cluster:insert:stream} by showing correctness of the global encoding and in particular, that there is no compounding error across the $\poly(\log\log(nd\Delta))$ iterations. 
\begin{theorem}
\thmlab{thm:cluster:stream}
Given a set $X$ of $n$ points on $[\Delta]^d$, let $f\left(n,d,\Delta,k,\frac{1}{\eps},z\right)$ be the number of points of a coreset construction with weights $[1,\poly(nd\Delta)]$ for $(k,z)$-clustering.  
Then \algref{alg:cluster:insert:stream} outputs a $(1+\eps)$-strong coreset of $X$ and uses $\O{dk\log(n\Delta)}+f\left(n,d,\Delta,k,\frac{\polylog(\log(nd\Delta))}{\eps},z\right)\cdot\polylog\left(\frac{1}{\eps},\log(nd\Delta)\right)$ bits of space.
\end{theorem}
\begin{proof}
We first show that with high probability, the following invariants are satisfied at each time:
\begin{enumerate}
\item
$\widehat{\sigma(x_1)},\ldots,\widehat{\sigma(x_t)}$ are constant-factor approximations to the online sensitivities $\sigma(x_1),\ldots,\sigma(x_t)$ with respect to $(k,z)$-clustering
\item
At most $\poly\left(k,\frac{1}{\eps},\log(nd\Delta)\right)$ points have been sampled into $\calS'$
\item
$Z$ at time $t$ is a $(1+\eps)^2$-coreset to $X_t:=\{x_1,\ldots x_t\}$
\end{enumerate}
We prove correctness by induction on $t$. 
The online sensitivity of the first point is $1$ and it will be correctly computed by the data structure. 
Thus we have that after the first step, we have $\calS'=Z=\{x_1\}$ and the base case is complete. 

Let $\calE_t$ be the event that the invariants hold at time $t$.  
Conditioning on $\calE_{t-1}$ and in particular the correctness of $Z$ at time $t-1$ being a $(1+\eps)$-strong coreset of $X_{t-1}:=\{x_1,\ldots,x_{t-1}\}$, we have that $Z\cup\{x_t\}$ is a strong coreset of $X_t:=\{x_1,\ldots,x_t\}$. 
Thus we can use $Z\cup\{x_t\}$ to compute a constant-factor approximation to the online sensitivity $\sigma(x_t)$ of $x_t$. 
By \thmref{thm:jl}, it follows that with high probability, we can instead use $Z\cup\{\widehat{x_t}\}$ to compute a constant-factor approximation $\widehat{\sigma(x_t)}$ to the online sensitivity $\sigma(x_t)$. 
Then by \thmref{thm:online:sens:cluster}, $\calS'$ after time $t$ will be a $(1+\eps)$-strong coreset for $X_t$ with high probability. 

By \thmref{thm:online:sens:cluster} through \thmref{thm:total:online:sens}, we have that by sampling points with probabilities that are constant-factor approximations to their online sensitivities, then the total number of sampled points into $\calS'$ is $\O{\frac{dk^2}{\eps^2}\log^3\frac{n\Delta}{\eps}}$ with high probability. 
Thus, $\PPr{\calE_t}\ge 1-\frac{1}{\poly(n)}$. 

Observe that merge-and-reduce will be correct if $\calS'$ is a stream with length at most $k\cdot\poly\left(\frac{1}{\eps},\log(nd\Delta)\right)$, since we set the failure probability to be $\frac{1}{\poly\left(\frac{1}{\eps},\log(n\Delta)\right)}$ and merge-and-reduce will consider at least $k$ stream updates before applying a new coreset construction. 
Thus conditioned on (1) the correctness of $\calS'$ after time $t$ being a $(1+\eps)$-strong coreset for $X_t$ and (2) the correctness of merge-and-reduce on $\calS'$, we have that $Z$ after time $t$ will be a $(1+\eps)$-strong coreset for $\calS'$, and thus a $(1+\eps)^2$-strong coreset for $X_t$, which completes the induction. 
That is, $\PPr{\calE_t\,\mid\,\calE_1,\ldots,\calE_{t-1}}\ge 1-\frac{1}{\poly(n)}$. 
The correctness guarantee then follows by rescaling $\eps$ and a union bound over all $t\in[n]$. 

To analyze the space complexity, note that by the guarantees of the invariants, the total number of sampled points into $\calS'$ is $\O{\frac{dk^2}{\eps^2}\log\frac{nd\Delta}{\eps}}$. 
However, $\calS'$ is not maintained explicitly, but instead given as input to the merge-and-reduce subroutine, which uses the efficient encoding for coreset construction given in \algref{alg:coreset:cluster} with accuracy $\frac{\eps}{\poly(\log\log(nd\Delta))}$ and failure probability $\frac{1}{\poly\left(\frac{1}{\eps},\log(nd\Delta)\right)}$. 
Thus by \lemref{lem:efficient:coreset:cluster}, the total representation of $Z$ uses $\O{dk\log(n\Delta)}+\frac{dk^2}{\eps^2}\cdot2^{2z}\cdot\polylog\left(\frac{1}{\eps},\log(n\Delta)\right)$ bits of space.
\end{proof}

\section{Fast \texorpdfstring{$(k,z)$}{(k,z)}-Clustering}
In this section, we describe how our one-pass streaming algorithm on insertion-only streams for $(k,z)$-clustering that uses $\mathcal{O}_{k,d,\eps}(1)$ words of space can further be implemented using fast amortized update time. 
In \secref{sec:fast:quadratic}, we first implement our streaming algorithm using $dk\cdot\polylog(\log(n\Delta))$ amortized update time. 
The $o(\log(n\Delta))$ amortized update time of this algorithm is not only faster than the $k\cdot\polylog(n\Delta)$ amortized update time of \cite{BhattacharyaCLP23}, but also provides crucial structural properties that will be ultimately used for our algorithm that uses $d\log k\cdot\polylog(\log(n\Delta))$ amortized update time in \secref{sec:kz:filter}. 
One of the such crucial properties is the proof in \secref{sec:means:mediods} that the $(k,z)$-clustering sensitivities and the $(k,z)$-mediods sensitivities are within constant factors of each other. 

\subsection{\texorpdfstring{$(k,z)$}{(k,z)}-Clustering Sensitivities and \texorpdfstring{$(k,z)$}{(k,z)}-medoids Sensitivities}
\seclab{sec:means:mediods}
In this section, we show that the sensitivity $\tau(x)$ for a point $x$ with respect to the dataset $X$ for $k$-medoids is a constant-factor approximation to the sensitivity $s(x)$ for a point $x$ with respect to the dataset $X$ for $(k,z)$-clustering. 
Thus to acquire a constant-factor approximation to the sensitivity $s(x)$, it suffices to approximate the $(k,z)$-medoids sensitivity $\tau(x)$, and vice versa.

We first recall the well-known fact that the optimal $(k,z)$-clustering and $(k,z)$-medoids costs are within a constant factor of each other. 
For the sake of completeness, we include the proof. 
\begin{lemma}
\lemlab{lem:opt:median:medoid}
For each $X\subset[\Delta]^d$, the optimal $(k,z)$-clustering cost is a $2^{z+1}$-approximation of the optimal $(k,z)$-medoid clustering cost.
\end{lemma}
\begin{proof}
Let $C$ be an optimal $(k,z)$-clustering of $X$ and for a fixed $c\in C$, let $P$ be the subset of $X$ served by $c$. 
Let $c'\in P$ be the point closest to $c$. 
Then we have for all $p\in P$, 
\[\|c'-p\|_2^z\le2^z(\|c'-c\|_2^2+\|p-c\|_2^2)\le2^{z+1}\|p-c\|_2^2.\] 
Therefore, $\Cost(P,c')\le2^{z+1}\cdot\Cost(P,c)$. 
The conclusion then follows by iterating over all $c\in C$. 
\end{proof}

We now show that the $(k,z)$-clustering and $(k,z)$-medoids sensitivities are within a constant factor of each other. 
\begin{theorem}
\thmlab{thm:sens:medoids}
There exists a constant $\gamma\le 2^{z+1}\cdot 101$ such that the sensitivity of a point $x$ with respect $X$ for $(k,z)$-clustering  is a $\gamma$-approximation to the sensitivity of $x$ with respect to $X$ for $(k,z)$-medoids.
\end{theorem}
\begin{proof}
Let $C$ be any set of $k$ centers that achieves the sensitivity $s(x)$ for $(k,z)$-clustering and let $S$ be any set of $k$ centers that achieves the sensitivity $\tau(x)$ for $(k,z)$-medoids. 
Since $S\subseteq X\subseteq[\Delta]^d$ and $|S|\le k$, then by the maximality characterization of sensitivity, we have $\tau(x)\le s(x)$. 

Now, let $c\in C$ be the closest center to $x$. 
Let $R\subset X$ be the set of points served by $c$ and let $r=|R|$. 
Firstly, let $w\in R$ be the farthest point from $x$. 
Then 
\[\max_{c'\in X}\frac{\Cost{x,c'}}{\Cost(R,c')}\ge\frac{\|w-x\|_2}{\sum_{y\in R}\|w-y\|_2}\ge\frac{\|w-x\|_2}{\sum_{y\in R}\|w-x\|_2+\|y-x\|_2}\ge\frac{1}{2|R|},\]
since $\|w-x\|_2\ge\|y-x\|_2$ for all $y\in R$. 
Thus, $\frac{\|w-x\|_2}{\sum_{y\in R}\|w-y\|_2}\ge\frac{1}{2r}$. 

Let $P\subseteq R$ be the set of points $y$ that are closer to $x$ than to $c$. 
We perform casework on the size $|P|$ of $P$. 

First, suppose $|P|<0.99r$. 
Note that there are at most $0.01r$ points $y$ such that $\Cost(y,c)>\frac{100}{r}\cdot\Cost(R,c)$. 
Thus there are at least $0.99r$ points $y$ such that $\Cost(y,c)\le\frac{100}{r}\cdot\Cost(R,c)$. 
If $|P|<0.99r$, then there exists $y\in R$ such that $y\notin P$, so that $\Cost(y,c)\le\frac{100}{r}\cdot\Cost(R,c)$ but $\|y-c\|_2\le\|y-x\|_2$. 
But then if we set $c'=y$, we have that $\|x-c'\|_2\ge\frac{1}{2}\|x-c\|_2$ and moreover,
\[\Cost(R,c')\le2^z(\Cost(R,c)+|R|\cdot\Cost(C,c'))\le2^z(101\cdot\Cost(R,c)).\]
Hence we have 
\[\frac{\Cost(x,c')}{\Cost(R,c')}\ge\frac{1}{2^{z+1}\cdot101}\frac{\Cost(x,c)}{\Cost(R,c)}.\]
The claim then follows by clustering $X\setminus R$ with the remaining $k-1$ centers using \lemref{lem:opt:median:medoid}. 

Otherwise, suppose $|P|\ge 0.99r$. 
Furthermore, it suffices to assume that there exist $0.99r$ points $y\in P$ such that $\Cost(y,c)\le\frac{100}{r}\cdot\Cost(R,c)$, or else we could reach the same conclusion as the previous case where $|P|<0.99r$ by moving $c$ to such a point $y$. 
Since these points are closer to $x$ than to $c$ by definition of $P$, then we have $\|y-c\|_2\ge\frac{1}{2}\cdot\|x-c\|_2$. 
Thus we have $\Cost(R,c)\ge0.99r\cdot\frac{1}{2}\cdot\|x-c\|_2^2$. 
Hence, $\frac{\Cost(x,c)}{\Cost(R,c)}\le\frac{4}{r}$. 
On the other hand, by the above argument, we have $\max_{c'\in X}\frac{\Cost{x,c'}}{\Cost(R,c')}\ge\frac{1}{2r}$ and thus $\tau(x)$ is a $2^{z+1}\cdot 8$-approximation to $s(x)$, after clustering the remaining points in $X\setminus R$ with the other $k-1$ centers. 
\end{proof}
Finally, we show that the sensitivities of points are preserved under coresets. 
That is, the $(k,z)$-clustering sensitivity of $x$ with respect to $X$ is well-approximated by the sensitivity of $x$ with respect to a coreset $Z$ of $X$. 
\begin{lemma}
\lemlab{lem:sens:allset:coreset}
Let $\gamma\ge 1$ be fixed and let $Z$ be a $\gamma$-strong coreset for $X$. 
For $p\in X$, let $s_X(p)=\max_{C: |C|\le k}\frac{\Cost(p,C)}{\Cost(X,C)}$ and let $s_Z(p)=\max_{C: |C|\le k}\frac{\Cost(p,C)}{\Cost(Z,C)}$. 
Then 
\[\frac{1}{\gamma}\cdot s_Z(p)\le s_X(p)\le\gamma\cdot s_Z(p).\]
\end{lemma}
\begin{proof}
Let $S$ be a set of $k$ centers in $C$ that achieves the sensitivity with respect to $Z$, so that $s_Z(p)=\frac{\Cost(p,S)}{\Cost(X,S)}$. 
Since $Z$ is a $\gamma$-strong coreset for $X$, then we have $\frac{1}{\gamma}\cdot\Cost(Z,S)\le\cdot\Cost(X,S)\le\gamma\cdot(Z,S)$. 
Therefore, we have
\[\frac{1}{\gamma}\cdot s_Z(p)\le s_X(p)\le\gamma\cdot s_Z(p).\]
\end{proof}

\subsection{Fast Update Algorithm}
\seclab{sec:fast:quadratic}
In this section, we show that our streaming algorithm can be implemented using $dk\cdot\polylog(\log(n\Delta))$ amortized update time. 
We remark that by comparison, the algorithm of \cite{BhattacharyaCLP23} uses $k\cdot\polylog(n\Delta)$ amortized update time, which is exponentially larger in terms of the dependency in the $n$ and $\Delta$ factors. 
On the other hand, we remark that the main point of \cite{BhattacharyaCLP23} is to handle the fully dynamic case, where insertions and deletions of points are both permitted, which we cannot handle, though they also require the entire dataset to be stored. 

We first use the following guarantee about quadratic runtime for the local search algorithm for $(k,z)$-clustering. 
\begin{theorem}
\thmlab{thm:local:search}
\cite{GuptaT08}
For each constant $z\ge 1$, there exists a polynomial time algorithm that outputs a $\O{z}$-approximation to $(k,z)$-clustering in time $\O{dn^2}$.  
\end{theorem}
We now show how the approximately solve the constrained clustering problem $\min\Cost(X,C)$ across all sets $C\subset X$ of $k$ centers containing a center at a fixed center $x$. 
The algorithm is quite simple. 
We compute a constant-factor approximation $C$ to the optimal $(k,z)$-medoids clustering on $X$. 
We then swap one of the centers of $C$ with $\{x\}$, choosing the swap with the best subsequent clustering cost. 
\begin{lemma}
\lemlab{lem:constrain:swap:constant}
For a fixed constant $\gamma\ge 1$, let $C$ be a $\gamma$-approximation to the optimal $(k,z)$-medoids clustering on $X\subset[\Delta]^d$. 
Let $S$ be the optimal $(k,z)$-medoids clustering on $X\subset[\Delta]^d$ containing a center at a fixed $x\in X$. 
Let $W$ be the set of $k-1$ centers $W\subseteq C$ such that $W\cup\{x\}$ has the minimum $(k,z)$-medoids clustering cost on $X$, and let $C'=W\cup\{x\}$. 
Then 
\[\Cost(X,S)\le\Cost(X,C')\le(2^z+2^{2z}+2^{2z}\cdot\gamma)\cdot\Cost(X,S).\]
\end{lemma}
\begin{proof}
Since $C'$ is a set of $k$ centers containing $x$ and $S$ is the optimal clustering among such sets of $k$ centers containing $x$, then $\Cost(X,S)\le\Cost(X,C')$. 
It thus remains to prove the right-hand side of the inequality. 

Let $S'=C\cup\{p\}$ so that $|S'|=k+1$ and also $C'\subset S'$ with $|S'\setminus C'|=1$. 
Let $\pi_{S'}:S\to C'$ be the mapping that assigns each center $s\in S$ to the closest center in $S'$, breaking ties arbitrarily. 
Let $Q=\{\pi_{S'}(s)\,\mid\, s\in S\}$ and note that $x\in Q$ since $x\in S$ and $x\in S'$. 
Moreover, since $|S|=k$, then $|Q|=k$, and since $Q\subset S'$ and $C'$ is the best subset of $S'$, then $\Cost(X,C')\le\Cost(X,Q)$. 

Now for any $p\in X$, consider its contribution to $\Cost(X,Q)$. 
Let $\pi_S: X\to S$ be the mapping that assigns each point $x\in X$ to the closest center in $S$, breaking ties arbitrarily. 
Then by generalized triangle inequality, c.f., \factref{fact:triangle}, 
\[\Cost(p,Q)\le2^z\cdot\Cost(p,\pi_{S}(p))+2^z\cdot\Cost(\pi_{S}(p),Q).\]
Since $\pi_{S'}$ maps each center in $S$ to its closest center in $S'$, then we have 
\[\Cost(\pi_{S}(p),Q)=\Cost(\pi_{S}(p),S').\] 
Thus, 
\[\Cost(p,Q)\le2^z\cdot\Cost(p,\pi_{S}(p))+2^z\cdot\Cost(\pi_{S}(p),S').\]
We also have by generalized triangle inequality, c.f., \factref{fact:triangle},
\[\Cost(\pi_{S}(p),S')\le2^z\cdot\Cost(\pi_{S}(p),p)+2^z\cdot\Cost(p,S')=2^z\cdot\Cost(p,S)+2^z\cdot\Cost(p,S').\]
Therefore,
\begin{align*}
\Cost(p,Q)&\le2^z\cdot\Cost(p,\pi_{S}(p))+2^{2z}\cdot\Cost(p,S)+2^{2z}\cdot\Cost(p,S')\\
&=2^z\cdot\Cost(p,S)+2^{2z}\cdot\Cost(p,S)+2^{2z}\cdot\Cost(p,S').
\end{align*}
Summing across $p\in X$, we thus have
\[\Cost(X,Q)\le2^z\cdot\Cost(X,S)+2^{2z}\cdot\Cost(X,S)+2^{2z}\cdot\Cost(X,S').\]
Since $C$ provides a $\gamma$-approximation to the optimal $(k,z)$-medoids clustering on $X\subset[\Delta]^d$ and $S$ is the optimal $(k,z)$-medoids clustering on $X$ containing a center at $x$, then $\Cost(X,C)\le\gamma\cdot\Cost(X,S)$. 
Since $S'=C\cup\{p\}$, then $\Cost(X,S')\le\Cost(X,C)\le\gamma\cdot\Cost(X,S)$ and thus we have
\[\Cost(X,Q)\le(2^z+2^{2z}+2^{2z}\cdot\gamma)\cdot\Cost(X,S).\]
Finally, since $\Cost(X,C')\le\Cost(X,Q)$, then the desired claim follows. 
\end{proof}

\begin{algorithm}[!htb]
\caption{Approximation algorithm for finding $(k,z)$-medoids clustering containing a fixed center}
\alglab{alg:constrained:cost:approx}
\begin{algorithmic}[1]
\Require{Constant-factor approximation $C\subset X\subset[\Delta]^d$ for $(k,z)$--medoids clustering on $X$, query point $x\in X$}
\Ensure{Constant-factor approximation for $(k,z)$-medoids clustering on $X$ with a center at $x$}
\State{$\psi\gets nd\Delta$, $\psi_x\gets\emptyset$}
\For{$c\in C$}
\State{Let $n_c$ be the number of points in $X$ assigned to $c$ by $C$}
\If{$n_c\cdot\|c-x\|_2^z<\psi$}
\State{$\psi\gets n_c\cdot\min_{p\in C\cup\{x\}\setminus\{c\}}\|c-p\|_2^z$, $\psi_x\gets c$}
\EndIf
\EndFor
\State{\Return $(\psi_x,\Cost(X,C)+\psi)$}
\end{algorithmic}
\end{algorithm}

Unfortunately, we cannot afford to exactly compute the cost of the best replacement of a center in $C$ with $\{x\}$. 
Thus we show that by snapping the points served by each center $c\in C$ and moving these points all to the closest center if $c$ is removed gives a good approximation to the cost of the best swap. 
In particular, if $n_c$ is the number of points served by $c$ and $p$ is the closest center of $C\cup\{x\}$ to $c$, then $\Cost(X,C)+n_c\cdot\|c-p\|_2^z$ is a good approximation to the cost of replacing $c$ with $x$ in $C$. 
\begin{lemma}
Let $\Psi=\Cost(X,C)+\psi$ be the output of \algref{alg:constrained:cost:approx}. 
Let $S$ be the optimal $(k,z)$-medoids clustering on $X\subset[\Delta]^d$ containing a center at a fixed $x\in X$.
Then there exists a constant $\gamma=2^{\Theta(z)}$ such that
\[\Cost(X,S)\le\Psi\le\gamma\cdot\Cost(X,S).\]
\end{lemma}
\begin{proof}
Let $W$ be the set of $k-1$ centers $W\subseteq C$ such that $W\cup\{x\}$ has the minimum $(k,z)$-medoids clustering cost on $X$, and let $C'=W\cup\{x\}$.
By \lemref{lem:constrain:swap:constant}, there exists a constant $\gamma_1\ge 1$ such that
\[\Cost(X,S)\le\Cost(X,C')\le(2^z+2^{2z}+2^{2z}\cdot\gamma_1)\cdot\Cost(X,S).\]
We claim $\Psi$ gives a constant-factor approximation to $\Cost(X,C')$. 

By generalized triangle inequality,
\[\Cost(X,C')\le 2^z\cdot\Cost(X,C)+2^z\cdot\Cost(C,C')\le2^z\cdot\Cost(X,C)+2^{2z}\cdot\Cost(X,C)+2^{2z}\cdot\Cost(X,C').\]
Let $C$ be a $\gamma_2$-approximation to the optimal unconstrained $(k,z)$-medoids clustering. 
Then $\Cost(X,C)\le\gamma_2\cdot\Cost(X,C')$, so that
\[\Cost(X,C')\le 2^z\cdot\Cost(X,C)+2^z\cdot\Cost(C,C')\le(2^z\gamma_2+2^{2z}\gamma_2+1)\cdot\Cost(X,C').\]
Hence, it suffices to show that $\Psi$ gives a constant-factor approximation to $\Cost(X,C)+\Cost(C,C')$. 

Note that $|C\setminus C'|=1$, then $\Cost(C,C')=\min_{c\in C}n_c\cdot\min_{p\in C'}\|c-p\|_2^z$, where $n_c$ is the number of points assigned to $c$ by $X$. 
Thus, we have $\psi=\Cost(C,C')$, so that
\[\Psi=\Cost(X,C)+\Cost(C,C').\]
Putting things together, we have that there exists a constant $\gamma=2^{\Theta(z)}$ such that
\[\Cost(X,S)\le\Psi\le\gamma\cdot\Cost(X,S).\]
\end{proof}

Unfortunately, we may not have time to compute the closest center $p$ to each center $c$ among $C\cup\{x\}\setminus\{c\}$. 
Thus we show that it suffices to just update the information for the center $u\in C$ that is closest to $x$. 
\begin{lemma}
\lemlab{lem:most:centers:ok}
Let $C$ be a constant-approximation to the optimal $(k,z)$-medoids clustering on $X\subset[\Delta]^d$. 
For a point $x\in X$, let $u$ be the center of $C$ closest to $x$, breaking ties arbitrarily. 
Then for any center $c\in C\setminus\{u\}$, 
\[\min_{p\in C\cup\{x\}\setminus\{c\}}\|c-p\|_2^z\le\min_{p\in C\setminus\{c\}}\|c-p\|_2^z\le2^{z+1}\min_{p\in C\cup\{x\}\setminus\{c\}}\|c-p\|_2^z.\] 
\end{lemma}
\begin{proof}
Note that the left-hand side of the inequality immediately holds because the minimization is taken over a larger superset on the left-most term. 
We now justify the right-hand side of the inequality. 

Let $p$ be the closest center of $S\setminus\{c\}$ to $c$. 
Observe that if $\|c-p\|_2\le\|c-x\|_2$, then $p$ remains the closest center of $S\cup\{x\}\setminus\{c\}$ to $c$ and thus the claim holds. 

Hence, it remains to consider the setting where $\|c-p\|_2<\|c-x\|_2$. 
In this case, we have by the optimality of $p$ in $S$, $\|c-p\|_2\le\|c-u\|_2$. 
Thus by generalized triangle inequality, 
\[\|c-p\|_2^z\le\|c-u\|_2^z\le 2^z\cdot\|u-x\|_2^z+2^z\cdot\|c-x\|_2^z.\]
Since $u$ is the closest center in $S$ to $x$ and $c\in S$, then it follows that $\|u-x\|_2\le\|c-x\|_2$ and thus
\[\|c-p\|_2^z\le2^{z+1}\cdot\|c-x\|_2^z.\]
\end{proof}
To estimate the sensitivity of a point $x$, we enumerate over all possible centers $p$ that serve $x$. 
To that end, we show that $\Phi+n_b\cdot r^z$ can only be an underestimate to the optimal $(k,z)$-medoids clustering constrained to the center $p$ being the closest to $x$. 
\begin{lemma}
\lemlab{lem:sens:est:under}
Let $x\in X$ be a fixed point and $p\in X$ be a fixed center with $r=\|x-p\|_2$. 
Let $Q$ be a constant-factor approximation to the optimal $(k,z)$-medoids clustering on $X$ containing a center at $p$ and let $\Psi$ be a constant-factor approximation to $\Cost(X,Q)$. 
Let $S$ be the optimal $(k,z)$-medoids clustering on $X$ containing a center at $p$ and no center in the interior of $B_r(x)$. 
Let $n_b$ be the number of points assigned to centers in $B_{r/2}(x)$. 
There exists a constant $\gamma\ge 1$ such that
\[\Psi+n_b\cdot r^z\le\gamma\cdot\Cost(X,S).\]
\end{lemma}
\begin{proof}
Let $Q$ be a $\zeta_1$-approximation to the optimal $(k,z)$-medoids clustering on $X$ containing a center at $p$. 
Since $S$ is the optimal solution for a more constrained search space, i.e., sets of at most $k$ centers that contain a center at $p$ and no center in the interior of $B_r(x)$, then we have 
\[\Cost(X,Q)\le\zeta_1\cdot\Cost(X,S).\] 
Let $\Psi$ be a $\zeta_2$-approximation to $\Cost(X,Q)$, so that 
\[\Psi\le\zeta_1\zeta_2\cdot\Cost(X,S).\] 

Now, consider a point $y$ served by a center $q\in Q$ inside $B_{r/2}(x)$. 
If $\|y-q\|_2\ge\frac{r}{4}$, then we have 
\[\|y-q\|_2+r\le 5\cdot\|y-q\|_2.\] 
Otherwise, if $\|y-q\|_2<\frac{r}{4}$, then by triangle inequality, $\dist(y,S)\ge\dist(q,S)-\|y-q\|_2\ge\frac{3r}{4}$. 
Thus, 
\[\|y-q\|_2+r\le\|y-q\|_2+2\cdot\dist(y,S).\]
Taking both cases together, we have
\[\|y-q\|_2+r\le 5\cdot\|y-q\|_2+2\cdot\dist(y,S).\]
Summing across all $y$ and $q$, we have
\[\Cost(X,Q)+n_b\cdot r^z\le5\Cost(X,Q)+2\Cost(X,S)\le(5\gamma+2)\cdot\Cost(X,S).\]
Since $\Psi$ is a $\zeta_2$-approximation to $\Cost(X,Q)$, then there exists a constant $\gamma_2\ge 1$ such that
\[\Psi+n_b\cdot r^z\le\gamma\cdot\Cost(X,S).\]
\end{proof}
Now we show that if there is a large number of points served by centers inside the ball $B_{r/2}(x)$ of radius $\frac{r}{2}$ around $x$, where $r:=\|x-p\|_2$, then $\Phi+n_b\cdot r^z$ is a good estimate to the cost of the optimal $(k,z)$-medoids clustering constrained to the center $p$ being the closest to $x$. 
\begin{lemma}
\lemlab{lem:ball:approx}
Let $x\in X$ be a fixed point and $p\in X$ be a fixed center with $r=\|x-p\|_2$. 
Let $Q$ be a constant-factor approximation to the optimal $(k,z)$-medoids clustering on $X$ containing a center at $p$ and let $\Psi$ be a constant-factor approximation to $\Cost(X,Q)$. 
Let $S$ be the optimal $(k,z)$-medoids clustering on $X$ containing a center at $p$ and no center in the interior of $B_r(x)$. 
Let $n_b$ be the number of points assigned to centers in $B_{r/2}(x)$ and let $n_a$ be the number of points assigned to centers in the annulus $B_r(x)\setminus B_{r/2}(x)$. 
If $n_b\ge n_a$, then there exist constants $\gamma_1,\gamma_2\ge 1$ such that
\[\Cost(X,S)\le\gamma_1(\Psi+n_b\cdot r^z)\le(\gamma_1\gamma_2)\cdot\Cost(X,S).\]
\end{lemma}
\begin{proof}
We show that there exists a constant $\gamma_1\ge 1$ with $\Cost(X,S)\le\gamma_1(\Psi+n_b\cdot r^z)$. 
Observe that this would imply by \lemref{lem:sens:est:under} that there also exists a constant $\gamma_2\ge 1$ such that $\gamma_1(\Psi+n_b\cdot r^z)\le(\gamma_1\gamma_2)\cdot\Cost(X,S)$. 

Recall that $S$ is the optimal solution across sets of at most $k$ centers that contain a center at $p$ and no center in the interior of $B_r(x)$. 

Consider a point $y$ served by a center $q\in Q$ inside $B_{r/2}(x)$. 
Let $\pi_S(y)$ be the closest center in $S$ to $y$, breaking ties arbitrarily. 
By generalized triangle inequality,
\[\Cost(y,S)\le2^z\cdot\|y-q\|_2^z+2^z\cdot\Cost(q,\pi_S(y)).\]
Summing across all $y$ served by centers $q\in B_{r/2}(x)$ and noting that $\dist(q,S)\le 2r$, we have
\[\sum_{y:\pi_Q(y)\in B_{r/2}(x)}\Cost(y,S)\le 8^z\cdot n_b\cdot\left(\frac{r}{2}\right)^z+2^z\cdot\sum_{y:\pi_Q(y)\in B_{r/2}(x)}\Cost(y,Q).\]

Next, consider a point $y$ served by a center $q\in Q$ inside $B_r(x)\setminus(B_{r/2}(x))$. 
Note that if $\|y-q\|_2\ge\frac{r}{4}$, then $\|y-p\|_2\le\|y-q\|_2+\|q-p\|_2\le 9\cdot\|y-q\|_2$. 
Otherwise if $\|y-q\|_2<\frac{r}{4}$, then $\|y-p\|_2\le\|y-q\|_2+\|q-p\|_2<3r$. 
Summing across all $y\in X$ served by centers $q\in B_{r/2}(x)$ and noting that $n_b\ge n_a$, we have
\[\sum_{y: \pi_Q(y)\in B_r(x)\setminus(B_{r/2}(x))}\Cost(y,S)\le 12^z\cdot n_b\cdot\left(\frac{r}{2}\right)^z+9\cdot2^z\cdot\sum_{y: \pi_Q(y)\in B_r(x)\setminus(B_{r/2}(x))}\Cost(y,Q).\]

Finally, for the points $y$ served by a center $q\in Q$ outside $B_r(x)$, note that $S$ can only have more centers outside $B_r(x)$ than $C$ because $S$ cannot have centers inside $B_r(x)$ whereas $C$ can. 
Thus by optimality of $S$, 
\[\sum_{y: \pi_Q(y)\in X\setminus (B_r(x))}\Cost(y,S)\le\sum_{y: \pi_Q(y)\in X\setminus (B_r(x))}\Cost(y,C).\]

Putting these together by summing across all $y$ served by centers $q\in B_{r/2}(x)$, $q\in B_r(x)\setminus B_{r/2}(x)$, and $q\in X\setminus B_r(x)$, we have that there exists a constant $\zeta\ge 1$ such that
\[\Cost(X,S)\le\zeta\cdot2^{\O{z}}\cdot\left(\Cost(X,Q)+n_b\cdot\left(\frac{r}{2}\right)^z\right).\]
Since $\Phi$ is a $\zeta_2$-factor approximation to $\Cost(X,Q)$, then there exists a constant $\gamma_1$ such that
\[\Cost(X,S)\le\gamma_1(\Psi+n_b\cdot r^z).\]
Therefore, we have
\[\Cost(X,S)\le\gamma_1(\Psi+n_b\cdot r^z)\le(\gamma_1\gamma_2)\cdot\Cost(X,S),\]
as desired.
\end{proof}
On the other hand, if the number of points served by centers inside the ball $B_{r/2}(x)$ of radius $\frac{r}{2}$ around $x$ is small, then $\Phi+n_b\cdot r^z$ may not be a good estimate to the cost of the optimal $(k,z)$-medoids clustering $S$ constrained to the center $p$ being the closest to $x$. 
Nevertheless, we show that it can provide an upper bound on $\frac{\Cost(x,S)}{\Cost(X,S)}$, while still being at most the sensitivity $s(x)$ of $x$. 
\begin{lemma}
\lemlab{lem:annulus:approx}
Let $x\in X$ be a fixed point and $p\in X$ be a fixed center with $r=\|x-p\|_2$. 
Let $s(x)$ be the sensitivity of $(k,z)$-medoids clustering with respect to $X$. 
Let $S$ be the optimal $(k,z)$-medoids clustering on $X$ containing a center at $p$ and no center in the interior of $B_r(x)$. 
Let $n_b$ be the number of points assigned to centers in $B_{r/2}(x)$ and let $n_a$ be the number of points assigned to centers in the annulus $B_r(x)\setminus B_{r/2}(x)$. 
If $n_b<n_a$, then 
\[\frac{1}{\gamma_1}\frac{\Cost(x,S)}{\Cost(X,S)}\le\frac{r^z}{\Psi+n_b\cdot r^z}\le 2^z\cdot\gamma_2\cdot s(x),\]
for fixed constants $\gamma_1,\gamma_2\ge 1$. 
\end{lemma}
\begin{proof}
Let $Q$ be a $\zeta_1$-factor approximation to the optimal $(k,z)$-medoids clustering on $X$ containing a center at $p$ and let $\Psi$ be a $\zeta_2$-factor approximation to $\Cost(X,Q)$. 
By \lemref{lem:sens:est:under}, we have that there exists a constant $\gamma_1\ge 1$ such that
\[\Psi+n_b\cdot r^z\le\gamma_1\cdot\Cost(X,S).\]
Since $S$ contains a center at $p$ and no centers in the interior of $B_r(x)$, then $\Cost(x,S)=r^z$. 
Then we have
\[\frac{1}{\gamma_1}\cdot\frac{\Cost(x,S)}{\Cost(X,S)}\le\frac{r^z}{\Psi+n_b\cdot r^z}.\]
On other hand, for a set $W$ of $k$ centers that contains no center in the interior of $B_{r/2}(x)$ and a center at $p$, so that we have $\frac{r}{2}\le \dist(x,W)\le r$ and $\Cost(X,W)\le\gamma_2(\Psi+n_b\cdot r^z)$ for some constant $\gamma_2\ge 1$.  
Thus by the maximality of sensitivity, we have
\[\frac{r^z}{\Psi+n_b\cdot r^z}\le 2^z\cdot\gamma_2\cdot\frac{\Cost(x,W)}{\Cost(X,w)}\le 2^z\cdot\gamma_2\cdot s(x).\]
\end{proof}

\begin{algorithm}[!htb]
\caption{$\BatchSens$: Fast batch approximation of sensitivities}
\alglab{alg:batch:sens}
\begin{algorithmic}[1]
\Require{Input set $Z\subset[\Delta]^d$ for $(k,z)$-medoids clustering, batch $B$ of $k$ query points}
\Ensure{Constant-factor approximation for the sensitivity of $x$ for $(k,z)$-medoids clustering on $Z\cup B$, for all $x\in B$}
\State{Compute $S$ to be a constant-factor solution on $Z\cup B$ for $(k,z)$-medoids clustering}
\For{$x\in B$}
\State{$\widehat{s(x)}\gets 0$}
\For{$p\in Z\cup B$}
\State{Let $Q$ be a constant-factor approximation of the optimal $(k,z)$-medoids clustering on $X$ containing a center at $p$}
\Comment{Only need to pre-process and compute once per $p$ across all $x$}
\State{Use $S$ to compute a constant-factor approximation $\Psi$ of $\Cost(X,Q)$}
\Comment{\algref{alg:constrained:cost:approx}}
\State{$r\gets\|x-p\|_2$}
\State{Let $n_b$ be the number of points assigned by $Q$ to centers of $Q$ in $B_{r/2}(x)$}
\State{$\widehat{s(x)}\gets\max\left(\widehat{s(x)},\frac{r^z}{\Psi+n_b\cdot r^z}\right)$}
\EndFor
\State{\Return $\widehat{s(x)}$}
\EndFor
\end{algorithmic}
\end{algorithm}
It follows that our algorithm in $\BatchSens$ provides constant-factor approximations to all sensitivities $x$ in the batch $B$. 
\begin{lemma}
\lemlab{lem:batch:sens:approx}
For a fixed $x\in X$, let $s(x)$ be the sensitivity of $(k,z)$-medoids clustering with respect to $X$. 
Let $\widehat{s(x)}$ be the output of \algref{alg:batch:sens}, for each $x\in B$, where $B$ is the input batch.  
Then there exist constants $\gamma_1,\gamma_2\ge 1$ such that 
\[\frac{1}{\gamma_1}\cdot s(x)\le\widehat{s(x)}\le 2^z\cdot\gamma_2\cdot s(x).\]
\end{lemma}
\begin{proof}
For each fixed $p\in X$, let $S$ be the optimal $(k,z)$-medoids clustering on $X$ containing a center at $p$ and no center in the interior of $B_r(x)$. 
Let $n_b$ be the number of points assigned to centers in $B_{r/2}(x)$ and let $n_a$ be the number of points assigned to centers in the annulus $B_r(x)\setminus B_{r/2}(x)$. 
We perform casework on whether $n_b<n_a$ or $n_b\ge n_a$. 

In the case where $n_b\ge n_a$, then by \lemref{lem:ball:approx}, we have
\[\Cost(X,S)\le\gamma_1(\Psi+n_b\cdot r^z)\le(\gamma_1\gamma_2)\cdot\Cost(X,S).\]
Since $\Cost(x,S)=r^z$, then it holds that $\frac{r^z}{\Psi+n_b\cdot r^z}$ is a constant-factor approximation to $\frac{\Cost(x,S)}{\Cost(X,S)}$. 

Otherwise if $n_b<n_a$, then by \lemref{lem:annulus:approx}, 
\[\frac{1}{\gamma_1}\frac{\Cost(x,S)}{\Cost(X,S)}\le\frac{r^z}{\Psi+n_b\cdot r^z}\le 2^z\cdot\gamma_2\cdot s(x),\]
for fixed constants $\gamma_1,\gamma_2\ge 1$. 
In particular, for the center $p\in X$ that realizes the sensitivity $s(x)$, we have that $\frac{r^z}{\Psi+n_b\cdot r^z}$ is an $\O{2^z}$-approximation to $s(x)$. 
For centers $q\in X$ that do not realize the sensitivity, we have that $\frac{r^z}{\Psi+n_b\cdot r^z}\le\gamma_2\cdot s(x)$. 
Therefore, 
\[\frac{1}{\gamma_1}\cdot s(x)\le\widehat{s(x)}\le 2^z\cdot\gamma_2\cdot s(x).\]
\end{proof}
It remains to analyze the amortized runtime of the algorithm $\BatchSens$. 
\begin{lemma}
\lemlab{lem:batch:sens}
Given a constant-factor coreset $Z$ to $X\subset[\Delta]^d$, there exists an absolute constant $C>1$ and an algorithm $\BatchSens$ that outputs $C$-approximations to the sensitivities to each of the points $x\in B$ in a batch $B$ of $k$ points of $X$, using amortized $\frac{(|Z|+k\log k)^2}{k}\cdot\poly(d)$ time. 
\end{lemma}
\begin{proof}
Consider $\BatchSens$ in \algref{alg:batch:sens}. 
By \lemref{lem:sens:allset:coreset}, to compute the sensitivity of $x\in B$ with respect to $X$, it suffices to compute the sensitivity of $x\in Z\cup B$, since $Z$ is a constant-factor coreset to $Y$ and thus $Z\cup B$ is a constant-factor coreset to $X=Y\cup B$. 
By \thmref{thm:sens:medoids}, it suffices to compute a constant-factor approximation to the $(k,z)$-medoids clustering sensitivity of $x$ with respect to $Z\cup B$. 
Hence, we want to compute $\max_{C\subset Z\cup B:|C|\le k}\frac{\Cost(x,C)}{\Cost(X,C)}$. 
Correctness thus follows by \lemref{lem:batch:sens:approx}. 

It remains to analyze the amortized runtime of $\BatchSens$. 
First, note that it takes time $\O{d(|Z|+k)^2}$ to compute a set $S$ of $k$ centers that serve as a constant-factor approximation to $Z\cup B$. 
We can then pre-process $(Z\cup B)$ in $\O{d(|Z|+k)^2}$ time so that for any $x\in Z\cup B$, we can find the set $W$ of $k-1$ centers of $S$ that after adjoining with $\{x\}$ achieves the best $k$-medoids clustering on $(Z\cup B)$ in $\O{k}$ time. 
Specifically, for each center $s\in S$, we can store the number $\eta_s$ of points assigned to $s$, as well as the distance $d_s$ to the closest center in $S\setminus\{s\}$. 

We can then update the information in time $\O{d(k+|Z|)}$ for each query $\{p\}$ to add to $S$, to ultimately form $Q$.  
In particular, we find $u\in S$ that is the nearest center of $S$ to $p$ and update $d_u$ to be $\|u-p\|_2$ if $\|u-p\|_2<d_u$. 
By \lemref{lem:most:centers:ok}, the points served by the remaining centers all have the distances to their nearest centers preserved up to a $2^{z+1}$-factor. 
Thus, we can approximate $\psi$ in \algref{alg:constrained:cost:approx} up to a $2^{z+1}$-factor, along with the corresponding center $\psi_p$ to be removed from $S$. 
By \lemref{lem:constrain:swap:constant}, the resulting set $Q$ of $k$ centers is a $2^{\O{z}}$-approximation to the optimal set $S$ of centers for $(k,z)$-medoids clustering that contains a center at $x$. 

It remains to pre-process $(Z\cup B)$ and $S$ to efficiently compute the number of centers $n_b$ assigned to centers of $Q$ in $B_{r/2}(x)$ for each $x\in B$ and each radius $r:=\|x-p\|_2$ for all $p\in(Z\cup B)$. 
To that end, we can first sort the points $y\in(Z\cup B)$ by their distances from $x$. 
Now we can scan radially outward from $x$, finding the closest center in $Z\cup B$ and iterating outward. 

Given two centers $u,v\in(Z\cup B)$, let $Q_u$ (resp. $Q_v$) be the clustering $Q$ induced by removing the output of $\phi_u$ (resp. $\phi_v$) from $S\cup\{u\}$ (resp. $S\cup\{v\}$) by \algref{alg:constrained:cost:approx} with query center $u$ (resp. $v$). 
In particular, suppose that after scanning radially outward from $x$, we first see $u$, followed immediately by $v$. 
Since $B_{r/2}(x)$ is monotonically increasing as $r$ increases, it suffices to simply consider the additional centers in $Q_v\setminus Q_u$. 
Since $Q_u$ and $Q_v$ are both formed by swapping a center of $S$ with $\{u\}$ and $\{v\}$ respectively, then we have $|Q_v\setminus Q_u|\le 2$. 
Hence, to compute $n_b$ for $v$, it suffices to use the previous computation of $n_b$ for $u$ and consider the number of points assigned to at most two centers of $Q_u$, which can be done in $\O{1}$ time. 

Therefore, the total time to compute $S$ and perform additional pre-processing is $\O{d(|Z|+k)^2}$. 
For each possible center $p\in(Z\cup B)$, we use $\O{k+|Z|}$ time to update the pre-processing information to obtain $\Psi$. 
For each point $x\in B$, we use $\O{dk\log k}$ time to sort the points by their distances from $x$. 
For each point $x\in B$ and possible center $p\in(Z\cup B)$, we use $\O{1}$ time to compute $n_b$. 
Hence, the amortized runtime is $\frac{d(|Z|+k\log k)^2}{k}$. 
\end{proof}

\begin{algorithm}[!htb]
\caption{Fast $(k,z)$-clustering}
\alglab{alg:fast:kz}
\begin{algorithmic}[1]
\Require{Set $X=\{x_1,\ldots,x_n\}$ of points in $[\Delta]^d$ that arrive as a stream}
\Ensure{$(k,z)$-clustering coreset for $X$}
\State{$\calS\gets\emptyset$, $\lambda\gets\O{\frac{kd}{\eps^2}\log(n\Delta)}$}
\State{Batch $x_1,\ldots,x_n$ into blocks $B_1,\ldots,B_{n/k}$ of $k$ updates}
\For{$b\in[n/k]$}
\State{Let $Z$ be a coreset for block $B_{b-1}$}
\State{Call $\BatchSens$ to batch approximation of sensitivities for $x_t$ with $t\in B_b$}
\State{Sample points of $B_{b-1}$ into a stream $\calS'$ using sensitivity sampling}
\State{Update $Z$ by running merge-and-reduce on $\calS'$}
\EndFor
\end{algorithmic}
\end{algorithm}

Putting together \thmref{thm:cluster:stream} and \lemref{lem:batch:sens}, we have:
\begin{restatable}{theorem}{thmclusterktime}
\thmlab{thm:cluster:k:time}
Given a set $X$ of $n$ points on $[\Delta]^d$, let $f\left(n,d,\Delta,k,\frac{1}{\eps},z\right)$ be the number of points of a coreset construction with weights $[1,\poly(nd\Delta)]$ for $(k,z)$-clustering. 
There exists an algorithm that uses $dk\cdot\polylog(\log(nd\Delta))$ amortized update time and $\O{dk\log(n\Delta)}+f\left(n,d,\Delta,k,\frac{\polylog(\log(nd\Delta))}{\eps},z\right)\cdot\polylog\left(\frac{1}{\eps},\log(nd\Delta)\right)$ bits of space, and outputs a $(1+\eps)$-strong coreset of $X$.
\end{restatable}

\subsection{Filtering of Low-Sensitivity Points}
\seclab{sec:kz:filter}
In this section, we show that our one-pass streaming algorithm can be implemented in $d\log k\cdot\polylog(\log(n\Delta))$ amortized update time. 

We first define a quadtree embedding. 
Due to the desiderata of fast update time, our construction is somewhat non-standard. 
Given $s=(s_1,\ldots,s_d)\in\mathbb{Z}^d$, $t\in\{0,1,\ldots,\ell\}$, and a parameter $z>1$, we define the axis-aligned grid $\calG_{s,t,\zeta}$ over $\mathbb{Z}^d$ with side length $\zeta^t$, so that $s=(s_1,\ldots,s_d)$ lies on one of the corners of the grid. 
Then for $X\in\mathbb{R}^{[\Delta]^d}$, we define $G_{s,t,\zeta}(X)$ as the frequency vector over the hypercubes of the grid $\calG_{s,t,\zeta}$ that counts the total number or weight of points in each hypercube. 
Specifically, each cell of a grid of length $\zeta^t$ has closed boundaries on one side and open boundaries on the other side, i.e., a cell containing $(u_1,\ldots,u_d)$ cannot also contain $(v_1,\ldots,v_d)$ for any $v_i\ge u_i+\zeta^t$. 
We define $\TreeDist_{s,\zeta}(x,y)$ to be $\sqrt{d}\cdot\zeta^\alpha$, where $\alpha$ is the smallest integer such that $x$ and $y$ are in two different cells of $\calG_{s,\alpha,\zeta}$ or $\TreeDist_{s,\zeta}(x,y)=0$ if no such grid exists. 
For our purposes, we will set $\zeta=n^{1-c}$ for a fixed constant $c\in(0,1)$ so that $\ell=\O{\frac{1}{c}}$ for $\Delta=\poly(n)$. 

We first show that the distance between a pair of points cannot be underestimated by the quadtree. 
\begin{lemma}
\lemlab{lem:tree:contraction}
For every $x,y\in[\Delta]^d$ and every $s=(s_1,\ldots,s_d)\in\mathbb{Z}^d$ and $\zeta>1$, we have
\[\dist(x,y)\le\TreeDist_{s,\zeta}(x,y).\]
\end{lemma}
\begin{proof}
Since $x,y\in[\Delta]^d$, then by an averaging argument there exists some coordinate for which $x$ and $y$ are separated by distance $\frac{\|x-y\|_2}{\sqrt{d}}$. 
Then $x$ and $y$ must be separated in grid $\calG_{s,\alpha,\zeta}$ for $\zeta^\alpha\le\frac{\|x-y\|_2}{\sqrt{d}}$. 
Therefore, $\TreeDist_{s,\zeta}(x,y)=\sqrt{d}\cdot\zeta^\beta$ for $\zeta^\beta>\frac{\|x-y\|_2}{\sqrt{d}}$. 
Then it follows $\dist(x,y)\le\TreeDist_{s,\zeta}(x,y)$. 
\end{proof}

We next bound the probability the distance between a pair of points is overestimated by the quadtree for a fixed distortion rate.
\begin{lemma}
\lemlab{lem:tree:dilation}
Given a set $S$ and a parameter $t>2$, let $\calE$ be the event where no points of $S$ are within $\frac{1}{t}$ fraction toward the boundary of any cell in the grid that induces $\TreeDist$. 
Then conditioned on $\calE$, we have that for all $x,y\in S$,
\[\TreeDist_{s,\zeta}(x,y)\le t\cdot\sqrt{d}\cdot\zeta\cdot\dist(x,y).\]
Moreover, we have
\[\PPr{\calE}\ge1-\frac{d|S|}{t}.\]
\end{lemma}
\begin{proof}
Observe that if $x$ and $y$ are first in the same cell at level $i$, then $\TreeDist_{s,\zeta}(x,y)=\sqrt{d}\cdot\zeta^i$. 
However, because the cells at level $i-1$ have length $\zeta^{i-1}$ and conditioned on $\calE$, both of $x$ and $y$ are at least $\frac{1}{t}$ fraction away from the boundary, then we have
$\dist(x,y)\ge\frac{\zeta^{i-1}}{t}$, which implies that 
\[\TreeDist_{s,\zeta}(x,y)\le t\cdot\sqrt{d}\cdot\zeta\cdot\dist(x,y).\]
Next, we note that for each coordinate, the probability that $x$ is contained a distance that is within $\frac{1}{t}$ fraction toward the boundary is $\frac{1}{t}$. 
The desired claim then follows from a union bound over all $x\in S$ and all $d$ dimensions. 
\end{proof}

We now require the following generalization of the fast Euclidean $k$-means approximation algorithm by \cite{Cohen-AddadLNSS20} to $(k,z)$-clustering. 
We remark that the statement in \cite{Cohen-AddadLNSS20} provides a $\mathcal{O}_{\eps}(\log k)$-approximation to $k$-means in runtime $\tO{nd+(n\log\Delta)^{1+\eps}}$. 
We briefly describe the algorithm of \cite{Cohen-AddadLNSS20} and the necessary adjustments to provide \thmref{thm:fast:kz}. 

The algorithm of \cite{Cohen-AddadLNSS20} first applies the Johnson-Lindenstrauss transformation~\cite{johnson1984extensions} to all the points to reduce the dimension to $\O{\log k}$, which preserves all clustering costs within a constant factor~\cite{MakarychevMR19,IzzoSZ21}. 
It then constructs a standard quadtree embedding, which it uses to adaptively sample a set $C$ of $k$ centers. 
Specifically, adaptive sampling randomly selects a point of the input set to be the first center of $C$. 
It then iterates, iteratively choosing $k-1$ additional points from the input set, with probability proportional to the squared distance of the point from the current set of centers. 
To do this quickly, \cite{Cohen-AddadLNSS20} uses multiple independent instances of the quadtree embedding to estimate the distance of each point from the current set of centers.
Finally, to roughly estimate the clustering cost induced by $C$, \cite{Cohen-AddadLNSS20} applies a locality-sensitivity hashing procedure to assign each point to its approximate nearest neighbor in $C$. 

For $(k,z)$-clustering, we instead perform adaptive sampling by iteratively choosing each of the input points with probability proportional to the $z$-th power of distances of the point from the centers already opened in $C$. 
The $\O{\log\Delta}$ runtime of the algorithm of \cite{Cohen-AddadLNSS20} is due to the $\O{\log\Delta}$ levels in the quadtree used to estimate these distances. 
To achieve runtime independent of $\O{\log\Delta}$, we instead use our crude quadtree with $\O{1}$ levels, to provide $n^{\Theta(1)}$-approximations to these distances. 
Once $C$ is acquired from adaptive sampling, we then estimate the clustering cost by again using the crude quadtree to assign each point to its closest center in the quadtree, which uses constant time per point, since the crude quadtree has $\O{1}$ levels.  
Finally, we remark that unlike \cite{Cohen-AddadLNSS20}, we can perform the quadtree embedding multiple times to avoid any cases where the distortion between pairs of points is too large. 
Thus we have:
\begin{theorem}
\thmlab{thm:fast:kz}
\cite{Cohen-AddadLNSS20}
For any constant $\alpha\in(0,1)$ and $N\gg n$, there exists a $N^{\alpha}$-approximation algorithm to $(k,z)$-clustering on $X$ with $n$ weighted points, that uses $\O{nd\log(nd)}$ expected runtime. 
\end{theorem}

To estimate the sensitivity of a point $x$, we now enumerate over the distances induced by all levels of the quadtree for centers that serve $x$. 
To that end, we show that $\Psi+n_{\beta}\cdot d^{z/2}\cdot\zeta^{\beta z}$ can only be an underestimate to the optimal $(k,z)$-medoids clustering constrained to $\beta$ being the first level in which the center serving $x$ is in the same cell as $x$ in the quadtree. 
In particular, the following lemma should be considered the quadtree analog to \lemref{lem:sens:est:under} for the radial search. 
\begin{lemma}
\lemlab{lem:filter:under}
Let $x\in X$ be a fixed point and $p\in X$. 
Suppose $\TreeDist$ preserves all pairwise distances in $X$ up to a factor of $\kappa$ and suppose $\TreeDist(p,x)=\sqrt{d}\cdot\zeta^\beta$. 
Let $\kappa\ge 2$ be a parameter, so that $Q$ is a $\kappa$-factor approximation to the optimal $(k,z)$-medoids clustering on $X$ containing a center at $p$ and let $\Psi$ be is a $\kappa$-factor approximation to $\Cost(X,Q)$. 
Let $S$ be the optimal $(k,z)$-medoids clustering on $X$ containing a center at $p$ and no center in the interior of $B_r(x)$, where $r:=\|x-p\|_2$.  
Let $n_\beta$ be the weight of the points assigned to centers $q\in Q$ such that $\TreeDist(x,q)\le\sqrt{d}\cdot\zeta^{\beta-1}$. 
There exists a constant $\gamma\ge 1$ such that
\[\Psi+n_{\beta}\cdot d^{z/2}\cdot\zeta^{\beta z}\le\kappa^{\O{z}}\cdot\Cost(X,S).\]
\end{lemma}
\begin{proof}
Since $S$ is the optimal solution for a more constrained search space than $Q$, i.e., sets of at most $k$ centers that contain a center at $p$ and no center in the interior of $B_r(x)$, then we have 
\[\kappa^2\cdot\Cost(X,S)\ge\kappa\cdot\Cost(X,Q)\ge\Psi.\] 

Now, consider a point $y$ served by a center $q\in Q$ such that $\TreeDist(x,q)\le\sqrt{d}\cdot\zeta^{\beta-1}$. 
If $\|y-q\|_2\ge\frac{r}{4}$, then 
\[\|y-q\|_2+r\le 5\cdot\|y-q\|_2.\] 
Otherwise, for $\|y-q\|_2<\frac{r}{4}$, then by triangle inequality, we have
\[\dist(y,S)\ge\dist(q,S)-\|y-q\|_2\ge\frac{3r}{4}.\] 
Therefore,
\[\|y-q\|_2+r\le\|y-q\|_2+2\cdot\dist(y,S).\]
Putting the two cases together, it follows that
\[\|y-q\|_2+r\le 5\cdot\|y-q\|_2+2\cdot\dist(y,S).\]
Summing across all $y\in X$ served by $q$ such that $\TreeDist(x,q)\le\sqrt{d}\cdot\zeta^{\beta-1}$, 
\[\Cost(X,Q)+n_{\beta}\cdot r^z\le5\Cost(X,Q)+2\Cost(X,S)\le(5\kappa+2)\cdot\Cost(X,S).\]
Since $\Psi$ is a $\kappa$-approximation to $\Cost(X,Q)$ and conditioned on the $\TreeDist$ procedure preserving all pairwise distances in $X$ up to a factor of $\kappa$ and thus $\sqrt{d}\cdot\zeta^{\beta}$ is a $\kappa$-approximation to $r$, then
\[\Psi+n_{\beta}\cdot d^{z/2}\cdot\zeta^{\beta z}\le\kappa^{\O{z}}\cdot\Cost(X,S),\]
for $\kappa\ge 2$, i.e., bounded away from $1$. 
\end{proof}
We next show that if there is a large number of points served by centers that are first in the same cell as $x$ in a level of the quadtree before $\beta$, then $\Psi+n_{\beta}\cdot d^{z/2}\cdot\zeta^{\beta z}$ is a good estimate to the cost of the optimal $(k,z)$-medoids clustering constrained to $x$ being served at a center that is first in the same cell as $x$ at level $\beta$ in the quadtree. 
In particular, the following lemma should be considered the quadtree analog to \lemref{lem:ball:approx} for the radial search. 
\begin{lemma}
\lemlab{lem:small:cell:approx}
Let $x\in X$ be a fixed point and $p\in X$. 
Suppose $\TreeDist$ preserves all pairwise distances in $X$ up to a factor of $\kappa$ and suppose $\TreeDist(p,x)=\sqrt{d}\cdot\zeta^\beta$. 
Let $\kappa\ge 2$ be a parameter, so that $Q$ is a $\kappa$-factor approximation to the optimal $(k,z)$-medoids clustering on $X$ containing a center at $p$ and let $\Psi$ be is a $\kappa$-factor approximation to $\Cost(X,Q)$. 
Let $S$ be the optimal $(k,z)$-medoids clustering on $X$ containing a center at $p$ and no center in the interior of $B_r(x)$, where $r:=\|x-p\|_2$. 
Let $n_\beta$ be the weight of the points assigned to centers $q\in Q$ such that $\TreeDist(x,q)\le\sqrt{d}\cdot\zeta^{\beta-1}$ and let $n_\alpha$ be the weight of the points assigned to centers $q\in Q$ such that $\TreeDist(x,q)\ge\sqrt{d}\cdot\zeta^\beta$. 
If $n_{\beta}\ge n_{\alpha}$, then there exist constants $\gamma_1,\gamma_2\ge 1$ such that
\[\Cost(X,S)\le\kappa^{\O{\gamma_1 z}}\left(\Psi+n_{\beta}\cdot d^{z/2}\cdot\zeta^{\beta z}\right)\le\kappa^{\O{\gamma_1\gamma_2 z}}\cdot\Cost(X,S).\]
\end{lemma}
\begin{proof}
We first show that there exists a constant $\gamma_1\ge 1$ such that $\Cost(X,S)\le\kappa^{\O{\gamma_1 z}}\left(\Psi+n_{\beta}\cdot d^{z/2}\cdot\zeta^{\beta z}\right)$. 
Observe that this would imply by \lemref{lem:filter:under} that there also exists a constant $\gamma_2\ge 1$ such that $\kappa^{\O{\gamma_1 z}}\left(\Psi+n_{\beta}\cdot d^{z/2}\cdot\zeta^{\beta z}\right)\le\kappa^{\O{\gamma_1\gamma_2 z}}\cdot\Cost(X,S)$. 

Consider a point $y$ served by a center $q\in Q$ with $\TreeDist(x,q)\le\sqrt{d}\cdot\zeta^{\beta-1}$. 
Let $\pi_S(y)$ be the closest center in $S$ to $y$, breaking ties arbitrarily. 
By generalized triangle inequality,
\[\Cost(y,S)\le2^z\cdot\|y-q\|_2^z+2^z\cdot\Cost(q,\pi_S(y)).\]
Summing across all $y$ served by centers $q$ with $\TreeDist(x,q)\le\sqrt{d}\cdot\zeta^{\beta-1}$ and noting that $\dist(s,q)\le 2\sqrt{d}\cdot\zeta^\beta$ conditioned on the correctness of $\TreeDist$, as well as  and noting that $n_\beta\ge n_\alpha$, we have
\[\sum_{y:\TreeDist(x,\pi_Q(y))\le\sqrt{d}\cdot\zeta^{\beta-1}}\Cost(y,S)\le 2^z\cdot n_{\beta}\cdot\left(2\sqrt{d}\cdot\zeta^\beta\right)^z+2^z\cdot\sum_{y:\TreeDist(x,\pi_Q(y))\le\sqrt{d}\cdot\zeta^{\beta-1}}\Cost(y,Q).\]

Next, consider a point $y$ served by a center $q\in Q$ with $\TreeDist(x,q)\ge\sqrt{d}\cdot\zeta^\beta$. 
We have that if $\|y-q\|_2\ge\frac{r}{4}$, then $\|y-p\|_2\le\|y-q\|_2+\|q-p\|_2\le 9\cdot\|y-q\|_2$. 
Else if $\|y-q\|_2<\frac{r}{4}$, then $\|y-p\|_2\le\|y-q\|_2+\|q-p\|_2<3r$. 
Summing across all $y$ served by centers $q\in Q$ with $\TreeDist(x,q)\ge\sqrt{d}\cdot\zeta^\beta$,
\[\sum_{y:\TreeDist(x,\pi_Q(y))\ge\sqrt{d}\cdot\zeta^\beta}\Cost(y,S)\le 12^z\cdot n_\beta\cdot\left(\frac{r}{2}\right)^z+9\cdot2^z\cdot\sum_{y:\TreeDist(x,\pi_Q(y))\ge\sqrt{d}\cdot\zeta^\beta}\Cost(y,Q).\]

Conditioned on the correctness of $\TreeDist$, we have that $\sqrt{d}\cdot\zeta^\beta$ is a $\kappa$-approximation to $r$. 
Thus summing the two cases across all $y\in X$, we have that there exists a constant $\zeta\ge 1$ such that
\[\Cost(X,S)\le\kappa^{\zeta z}\cdot\left(\Cost(X,Q)+n_\beta\cdot\left(\sqrt{d}\cdot\zeta^\beta\right)^z\right).\]
Since $\Phi$ is a $\kappa$-factor approximation to $\Cost(X,Q)$, then there exists a constant $\gamma_1$ such that
\[\Cost(X,S)\le\kappa^{\O{\gamma_1 z}}\left(\Psi+n_{\beta}\cdot d^{z/2}\cdot\zeta^{\beta z}\right).\]
Thus,
\[\Cost(X,S)\le\kappa^{\O{\gamma_1 z}}\left(\Psi+n_{\beta}\cdot d^{z/2}\cdot\zeta^{\beta z}\right)\le\kappa^{\O{\gamma_1\gamma_2 z}}\cdot\Cost(X,S).\]
\end{proof}
On the other hand, if the number of points served by centers that are first in the same cell as $x$ in a level of the quadtree before $\beta$ is small, then $\Psi+n_{\beta}\cdot d^{z/2}\cdot\zeta^{\beta z}$ may not be a good estimate to the cost of the optimal $(k,z)$-medoids clustering constrained to $x$ being served at a center that is first in the same cell as $x$ at level $\beta$ in the quadtree. 
However, we next show that it can nevertheless provide an upper bound on the $\frac{\Cost(x,S)}{\Cost(X,S)}$ for any $S$ where $x$ is served by a center that is firs tin the same cell as $x$ at level $\beta$ in the quadtree. 
Importantly, the quantity is at most the sensitivity $s(x)$ of $x$. 
In particular, the following lemma should be considered the quadtree analog to \lemref{lem:annulus:approx} for the radial search. 
\begin{lemma}
\lemlab{lem:filter:over}
Let $x\in X$ be a fixed point and $p\in X$. 
Suppose $\TreeDist$ preserves all pairwise distances in $X$ up to a factor of $\kappa$ and suppose $\TreeDist(p,x)=\sqrt{d}\cdot\zeta^\beta$. 
Let $s(x)$ be the sensitivity of $(k,z)$-medoids clustering with respect to $X$. 
Let $\kappa\ge 2$ be a parameter, so that $Q$ is a $\kappa$-factor approximation to the optimal $(k,z)$-medoids clustering on $X$ containing a center at $p$ and let $\Psi$ be is a $\kappa$-factor approximation to $\Cost(X,Q)$. 
Let $S$ be the optimal $(k,z)$-medoids clustering on $X$ containing a center at $p$ and no center in the interior of $B_r(x)$, where $r:=\|x-p\|_2$. 
Let $n_\beta$ be the weight of the points assigned to centers $q\in Q$ such that $\TreeDist(x,q)\le\sqrt{d}\cdot\zeta^{\beta-1}$ and let $n_\alpha$ be the weight of the points assigned to centers $q\in Q$ such that $\TreeDist(x,q)\ge\sqrt{d}\cdot\zeta^\beta$.  
If $n_b<n_a$, then 
\[\frac{1}{\kappa^{\O{\gamma_1 z}}}\frac{\Cost(x,S)}{\Cost(X,S)}\le\frac{d^{z/2}\cdot\zeta^{\beta z}}{\Psi+n_\beta\cdot d^{z/2}\cdot\zeta^{\beta z}}\le\kappa^{\O{\gamma_2 z}}\cdot s(x),\]
for fixed constants $\gamma_1,\gamma_2\ge 1$.
\end{lemma}
\begin{proof} 
By \lemref{lem:filter:under}, there exists a constant $\gamma\ge 1$ such that
\[\Psi+n_{\beta}\cdot d^{z/2}\cdot\zeta^{\beta z}\le\kappa^{\gamma z}\cdot\Cost(X,S).\]
Since $S$ contains a center at $p$ and no centers in the interior of $B_r(x)$, then $\Cost(x,S)=r^z$. 
Then conditioned on the correctness of $\TreeDist$ giving $d^{z/2}\cdot\zeta^{\beta z}$ as a $\kappa^z$-approximation to $r^z$, we have that
\[\frac{1}{\kappa^{\O{\gamma_1 z}}}\frac{\Cost(x,S)}{\Cost(X,S)}\le\frac{r^z}{\Psi+n_\beta\cdot d^{z/2}\cdot\zeta^{\beta z}}.\]
However, for a set $W$ of $k$ centers that contains a center at $p$ but no center $q$ with $\TreeDist(x,q)<\sqrt{d}\cdot\zeta^{\beta-1}$, then we have $\frac{\sqrt{d}\cdot\zeta^{\beta-1}}{\kappa}\le\dist(x,W)$ and $\Cost(X,W)\le\kappa^{\gamma_2 z}(\Psi+n_{\beta}\cdot d^{z/2}\cdot\zeta^{\beta z})$ for some constant $\gamma_2\ge 1$.  
Thus by the maximality of sensitivity, we have
\[\frac{d^{z/2}\cdot\zeta^{\beta z}}{\Psi+n_b\cdot r^z}\le \kappa^{\gamma_2 z}\cdot\frac{\Cost(x,W)}{\Cost(X,w)}\le\kappa^{\gamma_2 z}\cdot s(x).\]
\end{proof}
Hence if $\TreeDist$ provides a $\kappa$-approximation to the pairwise distance of all points $x\in X$, then the algorithm in $\RoughSens$ provides a $\kappa^{\O{1}}$-factor approximations to all sensitivities $x$ in the batch $B$. 
\begin{lemma}
\lemlab{lem:rough:sens:approx}
For a fixed $x\in X$, let $s(x)$ be the sensitivity of $(k,z)$-medoids clustering with respect to $X$. 
Suppose $\TreeDist$ preserves all pairwise distances in $X$ up to a factor of $\kappa$. 
Let $\widehat{s(x)}$ be the output of \algref{alg:rough:sens}, for each $x\in B$, where $B$ is the input batch.  
Then there exist constants $\gamma_1,\gamma_2\ge 1$ such that 
\[\frac{1}{\kappa^{\gamma_1 z}}\cdot s(x)\le\widehat{s(x)}\le \kappa^{\gamma_2 z}\cdot s(x).\]
\end{lemma}
\begin{proof}
For each fixed $p\in X$, let $S$ be the optimal $(k,z)$-medoids clustering on $X$ containing a center at $p$ and no center in the interior of $B_r(x)$. 
Let $n_\beta$ be the weight of the points assigned to centers $q\in Q$ such that $\TreeDist(x,q)\le\sqrt{d}\cdot\zeta^{\beta-1}$ and let $n_\alpha$ be the weight of the points assigned to centers $q\in Q$ such that $\TreeDist(x,q)\ge\sqrt{d}\cdot\zeta^\beta$. 
We perform casework on whether $n_\beta<n_\alpha$ or $n_\beta\ge n_\alpha$. 

In the case where $n_\beta\ge n_\beta$, then by \lemref{lem:small:cell:approx}, 
\[\Cost(X,S)\le\kappa^{\O{\gamma_1 z}}\left(\Psi+n_{\beta}\cdot d^{z/2}\cdot\zeta^{\beta z}\right)\le\kappa^{\O{\gamma_1\gamma_2 z}}\cdot\Cost(X,S).\]
Since $\Cost(x,S)=r^z$, which is a $\kappa^z$-approximation to $(d^{z/2}\cdot\zeta^{\beta z})$ conditioned on the correctness of $\TreeDist$, then it holds that $\frac{d^{z/2}\cdot\zeta^{\beta z}}{\Psi+n_{\beta}\cdot d^{z/2}\cdot\zeta^{\beta z}}$ is a $\kappa^{\O{z}}$ approximation to $\frac{\Cost(x,S)}{\Cost(X,S)}$. 

Otherwise if $n_\beta<n_\alpha$, then by \lemref{lem:filter:over}, 
\[\frac{1}{\kappa^{\O{\gamma_1 z}}}\frac{\Cost(x,S)}{\Cost(X,S)}\le\frac{d^{z/2}\cdot\zeta^{\beta z}}{\Psi+n_\beta\cdot d^{z/2}\cdot\zeta^{\beta z}}\le\kappa^{\O{\gamma_2 z}}\cdot s(x),\]
for fixed constants $\gamma_1,\gamma_2\ge 1$. 
In particular, for the center $p\in X$ that realizes the sensitivity $s(x)$, we have that $\frac{d^{z/2}\cdot\zeta^{\beta z}}{\Psi+n_{\beta}\cdot d^{z/2}\cdot\zeta^{\beta z}}$ is a $\kappa^{\O{z}}$ approximation to $\frac{\Cost(x,S)}{\Cost(X,S)}$. 
Furthermore, for centers $q\in X$ that do not realize the sensitivity, we have that $\frac{d^{z/2}\cdot\zeta^{\beta z}}{\Psi+n_\beta\cdot d^{z/2}\cdot\zeta^{\beta z}}\le\kappa^{\O{\gamma_2 z}}\cdot s(x)$. 
Thus, 
\[\frac{1}{\kappa^{\gamma_1 z}}\cdot s(x)\le\widehat{s(x)}\le\kappa^{\gamma_2 z}\cdot s(x).\]
\end{proof}

\begin{algorithm}[!htb]
\caption{$\RoughSens$: Rough batch approximation of sensitivities}
\alglab{alg:rough:sens}
\begin{algorithmic}[1]
\Require{Input set $Z\subset[\Delta]^d$ for $(k,z)$-medoids clustering, batch $B$ of $k$ query points, approximation parameter $\alpha\in(0,1)$}
\Ensure{$n^\alpha$ approximation for the sensitivity of $x$ for $(k,z)$-medoids clustering on $Z\cup B$, for all $x\in B$}
\State{$\kappa\gets n^{\alpha}$, $\zeta\gets n^{\iota}$ for sufficiently small $\iota<\alpha\in(0,1)$}
\State{Compute $S$ to be a constant-factor solution on $Z\cup B$ for $(k,z)$-medoids clustering}
\For{$x\in B$}
\State{$\widehat{s(x)}\gets 0$}
\State{Build a tree-embedding on $Z\cup B$}
\While{there is a point of $Z\cup B$ within $\frac{1}{\kappa}$-fraction of the distance to a cell in the tree embedding}
\State{Create a tree-embedding on $Z\cup B$}
\EndWhile
\For{each level $\ell_\beta$ in the tree}
\State{Let $Q$ be a $\kappa$-factor approximation $Q$ of the optimal $(k,z)$-medoids clustering on $X$ initially separated from $x$ at level $\ell_\beta$}
\State{Use $\TreeDist$ to compute a constant-factor approximation $\Psi$ of $\Cost(X,Q)$}
\State{Let $n_\beta$ be the number of points served by centers of $Q$ in the same cell as $x$ before $\ell_\beta$}
\State{$\widehat{s(x)}\gets\max\left(\widehat{s(x)},\frac{d^{z/2}\cdot\zeta^{\beta z}}{\Psi+n_\beta\cdot(d^{z/2}\cdot\zeta^{\beta z}}\right)$}
\EndFor
\State{\Return $\widehat{s(x)}$}
\EndFor
\end{algorithmic}
\end{algorithm}
We now analyze the properties of $\RoughSens$. 
Since correctness is mostly given by \lemref{lem:rough:sens:approx}, the main consideration is the amortized update time. 
\begin{lemma}
\lemlab{lem:batch:filter}
Given any constant $\alpha\in(0,1)$ and a constant-factor coreset $Z$ of size $k\cdot\polylog(\log n)$ to $X\subset[\Delta]^d$, there exists an algorithm $\RoughSens$ that outputs $n^{\alpha}$-approximations to the sensitivities to each of the points $x\in B$ in a batch $B$ of $k$ points of $X$, using amortized $d\log(k)\cdot\polylog(\log(n\Delta))$ time, with probability at least $1-\frac{1}{\sqrt{n}}$.
\end{lemma}
\begin{proof}
By \lemref{lem:sens:allset:coreset}, we can compute the sensitivity of $x\in Z\cup B$ as a good approximation to computing the $x\in B$ with respect to $X$, since $Z$ is a constant-factor coreset to $Y$ so that $Z\cup B$ is a constant-factor coreset to $X=Y\cup B$. 
Moreover, it suffices to compute a constant-factor approximation to the $(k,z)$-medoids clustering sensitivity of $x$ with respect to $Z\cup B$, by \thmref{thm:sens:medoids}. 
Therefore, the goal is to compute a crude approximation to $\max_{C\subset Z\cup B:|C|\le k}\frac{\Cost(x,C)}{\Cost(X,C)}$. 
Let $\calE$ be the event that $\TreeDist$ approximates all pairwise distances within a factor of $\frac{1}{\poly(d)}\cdot n^{\frac{1}{100p}}$. 
Then conditioned on $\calE$, correctness follows by \lemref{lem:rough:sens:approx}. 
We remark that the algorithm iterates until $\calE$ occurs, so the dependency on $\calE$ is not in the correctness, but rather than runtime. 

Thus, we now analyze the amortized runtime of $\RoughSens$. 
By \thmref{thm:fast:kz}, we use $kd\log(kd)\cdot\polylog(\log(n\Delta))$ time to find a $\kappa$-approximation $S$ of $k$ centers that serve as a constant-factor approximation to $Z\cup B$. 
For $\Delta=\poly(n)$, the tree has height $\O{\frac{1}{\iota}}$ and thus by pre-processing the number of points in each cell, the total time form $Q$ across all levels is $\O{\frac{1}{\iota}}$. 
Similarly, we can approximate $\Psi$ up to a $\kappa$-factor, along with the corresponding center $\psi_\beta$ to be removed from $S$ that is consistent with $\Psi$. 

It remains to pre-process $(Z\cup B)$ and $S$ to efficiently compute the number of centers $n_{\beta}$ assigned to centers of $q\in Q$ in such that $\TreeDist(x,q)<\sqrt{d}\cdot\zeta^{\beta}$. 
Again, this is handled by pre-processing all the points $y\in(Z\cup B)$ in the quadtree by pre-computing the number of points in each cell at each level, which takes total time $k\cdot\polylog(\log n)$ due to the size of $X$. 
Then by starting at $\beta=0$ and iterating to $\beta=\O{\frac{1}{\iota}}$, we proceed up the tree to compute each $n_\beta$, using $\O{\frac{1}{\iota}}$ time in total. 

Therefore, the total time to compute $S$ and perform additional pre-processing is $kd\log(kd)\cdot\polylog(\log(n\Delta))$. 
Across all possible levels $\beta$, we use $\O{\frac{1}{\iota}}$ time to update the pre-processing information to obtain $\Psi$ and compute each value of $n_\beta$. 
Hence, the amortized runtime is $d\log(kd)\cdot\polylog(\log(n\Delta))$. 
Finally, we remark that if $d\ge\O{\log n}$, then we can apply the standard Johnson-Lindenstrauss dimensionality reduction technique to achieve dimension $\O{\log n}$, and thus the resulting amortized runtime is $d\log(k)\cdot\polylog(\log(n\Delta))$. 
\end{proof}

Putting everything together, we have:
\begin{restatable}{theorem}{thmcoresetmain}
\thmlab{thm:coreset:main}
Given a set $X$ of $n$ points on $[\Delta]^d$, let $f\left(n,d,\Delta,k,\frac{1}{\eps},z\right)$ be the number of points of a coreset construction with weights $[1,\poly(nd\Delta)]$ for $(k,z)$-clustering. 
There exists an algorithm that uses $d\log(k)\cdot\polylog(\log(n\Delta))$ amortized update time and $\O{dk\log(n\Delta)}+f\left(n,d,\Delta,k,\frac{\polylog(\log(nd\Delta))}{\eps},z\right)\cdot\polylog\left(\frac{1}{\eps},\log(nd\Delta)\right)$ bits of space, and outputs a $(1+\eps)$-strong coreset of $X$.
\end{restatable}
\begin{proof}
We perform induction on the time $t$ throughout the stream. 
Correctness at the first step is immediate, since the first point is inserted into the batch. 
Now we suppose for our inductive hypothesis that at time $t-1$, we have both a constant-factor coreset and a $(1+\eps)$-coreset at time $t-1$. 

Consider time $t$ for our inductive step. 
Either the new point is added to the batch of size $k$, in which case both coreset invariants are maintained, or we perform a sampling step on the batching using crude approximations to the sensitivities as an initial filtering step, because the batch has gotten too large. 
In the latter case, we first use the constant-factor coreset through \lemref{lem:batch:filter} to compute rough approximations to the $(k,z)$-sensitivities in a batch of $k$ points. 
This uses $d\log(k)\cdot\polylog(\log n)$ amortized runtime and produces an insertion-only stream $\calS$ of length $\O{n^{1-\Omega(1)}}$. 
We feed this insertion-only stream $\calS$ as the input to \algref{alg:fast:kz}. 
By \thmref{thm:cluster:k:time}, the algorithm uses $dk\polylog(\log(n\Delta))$ update time on the stream $\calS$ of length $o(n)$ to generate both a constant-factor coreset and a $(1+\eps)$-factor coreset at time $t$, which completes our induction. 
Since this is a lower-order term, then the total amortized runtime is $d\log(k)\cdot\polylog(\log n)$. 
\end{proof}

Finally, we note the implications of \thmref{thm:coreset:main} to dynamic $(k,z)$-clustering in the incremental setting. 
\begin{restatable}{theorem}{thmclustermain}
\thmlab{thm:cluster:main}
Given an insertion-only stream of $n$ points that defines a dataset $X$ on $[\Delta]^d$, there exists a one-pass streaming algorithm that uses $d\log(k)\cdot\polylog(\log(n\Delta))$ amortized update time and $\tO{dk\log(n\Delta)}$ bits of space, and outputs a $\O{z}$-approximation to $(k,z)$-clustering at all times in the stream. 
\end{restatable}
\begin{proof}
By \thmref{thm:offline:coreset:size}, there exists a coreset construction for Euclidean $(k,z)$-clustering that samples $\tO{k}$ weighted points of the input dataset. 
Thus, \thmref{thm:coreset:main} guarantees a constant-approximation coreset for $X$ using $d\log(k)\cdot\polylog(\log(n\Delta))$ amortized update time. 
In fact, because the proof of \thmref{thm:coreset:main} uses induction over the course of the stream, it actually achieves a constant-approximation coreset to every prefix of the stream. 
We can thus achieve a $\O{z}$-approximation at all times in the stream by applying a standard polynomial-time clustering approximation algorithm such as local search, c.f., \thmref{thm:local:search}, each time the merge-and-reduce data structure is updated. 
Moreover, note that by \thmref{thm:online:sens:cluster}, the final input stream induced by online sensitivity sampling to the merge-and-reduce data structure has length $\O{dk^2\log^3(n\Delta)}$. 
By applying a standard Johnson-Lindenstrauss transformation, c.f., \thmref{thm:jl}, we can view $d=\O{\log n}$. 
Therefore, the total runtime of all iterations of local search is $\O{k^2\log^4(n\Delta)}\cdot\tO{k^2\log n}$. 
Thus for $n\gg k^4$, the amortized runtime remains $d\log(k)\cdot\polylog(\log(n\Delta))$, due to the computation of the constant-approximation coreset. 
\end{proof}

\section{Subspace Embeddings}
\seclab{sec:subspace:embed}
In this section, we describe our one-pass streaming algorithm on insertion-only streams for $L_p$ subspace embeddings that uses $\mathcal{O}_{k,d,\eps}(1)$ words of space and amortized runtime $\O{d}$ per update. 
Recall that in problem of $L_p$ subspace embeddings, the rows of a matrix $\bA\in\mathbb{R}^{n\times d}$ arrive one-by-one, and the goal is to produce a matrix $\bM\in\mathbb{R}^{m\times d}$ such that for a fixed $\eps\in(0,1)$ and for all $\bx\in\mathbb{R}^d$, we have 
\[(1-\eps)\|\bA\bx\|_p\le\|\bM\bx\|_p\le(1+\eps)\|\bA\bx\|_p.\]
$L_p$ subspace embeddings are widely used in data compression applications where preserving the geometry of data is critical, e.g., dimensionality reduction, compressed sensing, and linear/robust regression. 
Hence, there is a long line of active work studying $L_p$ subspace embeddings~\cite{DrineasMM06a,DrineasMM06b,Sarlos06,DasguptaDHKM09,ClarksonW09,CohenP15,CohenMP20,BravermanDMMUWZ20,ParulekarPP21,MahabadiWZ22,MeyerMMWZ22,MuscoMWY22,CherapanamjeriSWZ23,MeyerMMWZ23,WoodruffY23}. 

We first describe an efficient encoding for coresets for $L_p$ embeddings, as well as a global encoding that be utilized to achieve our streaming algorithm that uses $\mathcal{O}_{k,d,\eps}(1)$ words of space. 
We then describe how our algorithm can be implemented using $\O{d}$ amortized update time. 

We first recall the following definition of leverage scores, which intuitively quantifies the importance of a row for $L_2$ subspace embeddings. 
 
\begin{definition}[Leverage scores]
\deflab{def:leverage:score}
The \emph{leverage score} of a row $\ba_i\in\mathbb{R}^d$ of a matrix $\bA\in\mathbb{R}^{n\times d}$ is  
\[\max_{\bx:\|\bA\bx\|_2=1}\langle\ba,\bx\rangle^2.\]
It is known that the leverage score for $\ba_i$ also admits the closed form $\ba_i^\top(\bA^\top\bA)^{-1}\ba_i$. 
\end{definition}
We first recall one of the possible generalizations of leverage scores to $L_p$ subspace embeddings. 
\begin{definition}[$L_p$ Lewis weights]
\deflab{def:lewis:weight}
The \emph{Lewis weight} of a row $\ba_i\in\mathbb{R}^d$ of a matrix $\bA\in\mathbb{R}^{n\times d}$ for $L_p$ subspace embedding is the $i$-th diagonal entry of the unique diagonal matrix $\bW$ such that for all $i\in[n]$
\[W_{i,i}=\ell_i(\bW^{\frac{1}{2}-\frac{1}{p}}\bA),\]
where $\ell_i$ denotes the $i$-th leverage score, i.e., the leverage score of the $i$-th row of the matrix $\bW^{\frac{1}{2}-\frac{1}{p}}\bA$. 
\end{definition}
The intuition for $L_p$ Lewis weights is not obvious; they are the leverage scores of the matrix after applying a proper change-of-density, where each row is weighted the number of times necessary to preserve $\|\bA\bx\|_p$ when consider $\|\bA'\bx\|_2$, where $\bA'$ is the reweighted matrix~\cite{CohenP15}. 
An advantage of using Lewis weights to sample rows over other quantities such as the $L_p$ leverage scores~\cite{DasguptaDHKM09} or the $L_p$ sensitivities~\cite{BravermanDMMUWZ20,ChenD21} is that the $L_p$ Lewis weights are known to admit the optimal sampling complexity. 
On the other hand, for the special case of $p=2$, \cite{Sarlos06} showed the existence of an oblivious subspace embedding with $\O{\frac{d}{\eps^2}}$ rows. 
Combining these methods, we have:
\begin{theorem}
\thmlab{thm:lewis:weights}
\cite{Sarlos06,CohenP15,WoodruffY23}
Given a matrix $\bA\in\mathbb{R}^{n\times d}$, accuracy $\eps\in(0,1)$, and $p\ge 1$, let 
\[f\left(d,\frac{1}{\eps},p\right)=\begin{cases}
\frac{d}{\eps^2}\log\frac{d}{\eps}\polylog\left(\frac{d}{\eps}\right),\qquad& p\in[1,2)\\
\O{\frac{d}{\eps^2}},\qquad& p=2\\
\frac{d^{p/2}}{\eps^2}\cdot\polylog\left(\frac{d}{\eps}\right),\qquad&p>2     
\end{cases}.\]
Then there exists a coreset construction for $L_p$ subspace embeddings that consists of $f\left(d,\frac{1}{\eps},p\right)$ rows, each with entry at most $\|\bA\|_\infty\cdot\poly(n)$, which can be constructed in time polynomial in $n$ and $d$. 
\end{theorem}
We next define the generalization of $L_p$ Lewis weights to the online setting. 
\begin{definition}[Online $L_p$ Lewis weights]
\deflab{def:online:lewis:weight}
The online $L_p$ \emph{Lewis weight} of a row $\ba_i\in\mathbb{R}^d$ of a matrix $\bA\in\mathbb{R}^{n\times d}$ is the Lewis weight of $\ba_i$ with respect to the matrix $\bA_i=\ba_1\circ\ldots\circ\ba_i$, i.e., the submatrix of $\bA$ consisting of the first $i$ rows of $\bA$. 
\end{definition}
We recall the following upper bounds on the sum of the online Lewis weights. 
\begin{theorem}
\thmlab{thm:total:online:lewis}
\cite{BravermanDMMUWZ20,WoodruffY23}
Let $\bA=\{\ba_1,\ldots,\ba_n\}\in\mathbb{R}^{n\times d}$ be a matrix of $n$ rows and let $\sigma(\ba_t)$ denote the online Lewis weight of $\ba_t$ for $t\in[n]$ for $L_p$-subspace embedding, where $p\ge 1$. 
Then for $p\in[1,2]$, 
\[\sum_{t=1}^n\sigma(\ba_t)=d\cdot\polylog(n\kappa),\]
and for $p>2$,
\[\sum_{t=1}^n\sigma(\ba_t)=d^{p/2}\cdot\polylog(n\kappa),\]
where $\kappa$ is the online condition number.
\end{theorem}
We now recall the full guarantees of online Lewis weight sampling. 
\begin{theorem}[Online Lewis weight sampling]
\thmlab{thm:online:lewis:embedding}
\cite{WoodruffY23}
Given a sequence $\ba_1,\ldots,\ba_n\in\mathbb{R}^{n\times d}$ of rows, suppose each point $\ba_t$ is sampled with probability $p_t\ge\min(1,\gamma\cdot\sigma(\ba_t))$, where $\sigma(\ba_t)$ is the online Lewis weight of $\ba_t$ and $\gamma=\O{\frac{\log n}{\eps^2}}$ and multiplied by $\frac{1}{p_t}$ if $\ba_t$ is sampled. 
Then with high probability, the rescaled sample is a $(1+\eps)$-strong coreset for $L_p$-subspace embedding that has a number of rows that is at most:
\[g\left(n,d,\frac{1}{\eps},\kappa,p\right)=\begin{cases}
\frac{d}{\eps^2}\cdot\polylog(n\kappa),\qquad& p\in[1,2]\\
\frac{d^{p/2}}{\eps^2}\cdot\poly(\log^{p/2}(n\kappa)),\qquad&p>2     
\end{cases}.\]
\end{theorem}

We use standard approaches, e.g., Lemma 4.1 in \cite{ClarksonW09} to upper bound the magnitude of the logarithm of any nonzero singular value of a matrix with bounded integer entries. 
This in turn upper bounds the logarithm of the online condition number. 
\begin{lemma}
Suppose $\bA\in\mathbb{R}^{n\times d}$ has integer entries bounded in magnitude by $M$. 
Then for any nonzero singular value $\sigma_i$ of $\bA$, we have $|\log\sigma_i|=\O{d\log(nM)}$.
\end{lemma}
\begin{proof}
We have that the characteristic polynomial of $\bA^\top\bA$ is $p(x)=\det(x\cdot\mathbb{I}_d-\bA^\top\bA)$. 
Note that if $\bA$ has rank $r$, then $p(x)=x^{d-r}\prod_{i=1}^r(x-\lambda_i)$, where $\lambda_i$ is the $i$-th eigenvalue of $\bA^\top\bA$, writing $\lambda_1\ge\lambda_2\ge\ldots\ge\lambda_r$. 
Since all entries of $\bA^\top\bA$ are integers, then the coefficients of $p(x)$ are also integers. 
Moreover, since the eigenvalues of $\bA^\top\bA$ are non-negative, then $\prod_{i=1}^r\lambda_i\ge 1$. 
Because all entries of $\bA^\top\bA$ are at most $\poly(ndM)$ in magnitude, then 
\[\lambda_i\le\lambda_1\le\|\bA^\top\bA\|_F\le\poly(ndM).\]
Since $\prod_{i=1}^r\lambda_i\ge 1$, then we have that $\lambda_i\ge\frac{1}{\poly(ndM)^d}$. 
Since $\sigma_i^2=\lambda_i$, then $|\log\sigma_i|=\O{d\log(nM)}$.
\end{proof}

We next recall the following definition of a well-conditioned basis.
\begin{definition}[Well-conditioned basis]
Given a matrix $\bA\in\mathbb{R}^{n\times d}$ of rank $r$, let $p\in[1,\infty)$ and $q$ be its dual norm, so that $\frac{1}{p}+\frac{1}{q}=1$. 
Then matrix $\bU\in\mathbb{R}^{n\times d}$ is an $(\alpha,\beta,p)$-well-conditioned basis for the column space of $\bA$ if:
\begin{enumerate}
\item 
The column space of $\bA$ is the column space of $\bU$
\item
$\sum_{i\in[n],j\in[d]} U_{i,j}^p\le\alpha^p$
\item
For all $\bx\in\mathbb{R}^d$, we have $\|\bz\|_q\le\beta\|\bU\bz\|_p$
\end{enumerate}
\end{definition}
One such construction of a well-conditioned basis gives the following properties:
\begin{theorem}
\cite{DasguptaDHKM09}
\thmlab{thm:well:cond:basis}
Given a matrix $\bA\in\mathbb{R}^{n\times d}$ of rank $r$, let $p\in[1,\infty)$ and $q$ be its dual norm, so that $\frac{1}{p}+\frac{1}{q}=1$. 
There exists an $(\alpha,\beta,p)$-well-conditioned basis $\bU$ for the column space of $\bA$ such that:
\begin{itemize}
\item 
If $p<2$, then $\alpha=r^{\frac{1}{p}+\frac{1}{2}}$ and $\beta=1$
\item 
If $p=2$, then $\alpha=\sqrt{r}$ and $\beta=1$
\item 
If $p>2$, then $\alpha=r^{\frac{1}{p}+\frac{1}{2}}$ and $\beta=r^{\frac{1}{q}-\frac{1}{2}}$
\end{itemize}
Moreover, $\bU$ can be computed in time $\O{ndr+nr^2\log n}$. 
\end{theorem}
For completeness, we briefly describe the construction of the well-conditioned basis of \thmref{thm:well:cond:basis}, given by \cite{DasguptaDHKM09}. 
Given a QR decomposition $\bA=\bQ\bR$, so that $\bQ$ is any $n\times r$ matrix that is an orthonormal basis for the span of $\bA$ and $\bR$ is an $r\times d$ matrix, define the set $S=\{\bx\in\mathbb{R}^d\,\mid\,\|\bQ\bx\|_p\le1\}$ and the $r\times r$ matrix $\bF$ so that $\calE_{LJ}=\{\bx\in\mathbb{R}^r\,\mid\,\bx^\top\bF\bx\le 1\}$ is the L\"{o}wner-John ellipsoid of $S$, let $\bG\in\mathbb{R}^{d\times r}$ be the full rank and upper triangular matrix such that $\bF=\bG^\top\bG$. 
Then $\bU:=\bQ\bG^{-1}$ is the desired $(\alpha,\beta,p)$-well-conditioned basis for the column span of $\bA$. 

We now define the $L_p$ sensitivity of a row of a matrix, for the purposes of $L_p$ subspace embeddings. 
\begin{definition}[$L_p$ sensitivity]
For a matrix $\bA\in\mathbb{R}^{n\times d}$, the \emph{$L_p$ sensitivity} of row $\ba_t$ is the quantity
\[\max_{\by\in\mathbb{R}^d: \|\bA\by\|_p=1}\left\lvert\langle\ba_t,\by\rangle\right\rvert^p.\]
\end{definition}
Finally, we recall the following upper bound on the total $L_p$ sensitivity for $L_p$ subspace embedding. 
\begin{lemma}[Upper bounds on the sum of the $L_p$ sensitivities]
\cite{CohenP15,ChenD21}
\lemlab{lem:total:lp:sens}
For a matrix $\bA\in\mathbb{R}^{n\times d}$, let $s(\ba_t)$ denote the $L_p$ sensitivity of the $t$-th row of $\bA$ for $L_p$ subspace embedding. 
Then $\sum_{t=1}^n s(\ba_t)\le d$ for $p\le 2$ and $\sum_{t=1}^n s(\ba_t)\le d^{p/2}$ for $p>2$. 
\end{lemma}

\subsection{Efficient Encoding for Coreset Construction for \texorpdfstring{$L_p$}{Lp} Subspace Embedding}
We now give our efficient encoding for a given coreset for $L_p$ subspace embeddings. 
Given a matrix $\bA$, which can be viewed as either the original matrix or a set of reweighted rows that forms a coreset of some underlying matrix, we first acquire a constant-factor approximation $\bM$ for $L_p$ subspace embedding on $\bA$. 
We then use the existence of a well-conditioned basis to compute a deterministic preconditioner $\bP\in\mathbb{R}^{d\times d}$, so that $\bM\bP^{-1}$ has condition number $\poly(d)$. 
Now for each row $\ba_t$ of $\bA$, let $\bb'_t$ be $\ba_t\bP$ with each coordinate rounded to a power of $(1+\eps')$ for $\eps'=\frac{\poly(\eps)}{\poly(d^p,\log(nd)}$. 
We then store the exponent of each rounded coordinate of $\bB'$, as well as $\bM$. 
The algorithm appears in full in \algref{alg:coreset:embedding}. 

\begin{algorithm}[!htb]
\caption{Efficient Encoding for Coreset Construction for $L_p$ Subspace Embedding}
\alglab{alg:coreset:embedding}
\begin{algorithmic}[1]
\Require{Matrix $\bA\subset\{-M,\ldots,-1,0,1,\ldots,M\}^{n\times d}$, accuracy parameter $\eps\in(0,1)$, failure probability $\delta\in(0,1)$}
\Ensure{$(1+\eps)$-coreset for $L_p$ subspace embedding}
\State{$\eps'\gets\frac{\poly(\eps)}{\poly(d^p,\log(nd\Delta)}$}
\State{Find a matrix $\bM$ that is a constant-factor $L_p$ subspace embedding on $\bA$}
\Comment{\thmref{thm:lewis:weights}}
\State{Compute a deterministic preconditioner $\bP$ of $\bM$ using well-conditioned bases, such that $\bM\bP^{-1}$ has condition number $\poly(d)$}
\For{each row $\ba\in\bA$}
\State{Let $\bb'$ be $\ba\bP$ with each coordinate rounded to a power of $(1+\eps')$}
\State{$\bB'\gets\bB'\circ\bb'$, storing the \emph{exponent} for each entry of $\bb'$}
\EndFor
\State{\Return $(\bM,\bB')$, from which estimate $\bA'$ can be constructed}
\end{algorithmic}
\end{algorithm}
Given a constant-factor subspace embedding $\bM$, we now show that the matrix $\bB'$ acquired by rounding each row of $\bA\bP$ can be used to construct a matrix $\bA'$ that is a strong coreset of $\bA$. 
\begin{lemma}
\lemlab{lem:x:to:xprime:embedding}
Let $\eps\in\left(0,\frac{1}{2}\right)$ and let $\bA'=\bB'\bP^{-1}$, where $\bB'$ is the output of \algref{alg:coreset:embedding} and $\bP$ is the deterministic preconditioner of $\bM$. 
Then for all $\bx\in\mathbb{R}^d$,
\[(1-\eps)\|\bA\bx\|_p^p\le\|\bA'\bx\|_p\le(1+\eps)\|\bA\bx\|_p^p.\]
\end{lemma}
\begin{proof}
Without loss of generality, we assume that $\bA$ is full rank. 
Note that we can write $\bA\bx=\bA\bP\bP^{-1}\bx$ and thus by rewriting $\by=\bP^{-1}\bx$, it suffices to show that for all $\by\in\mathbb{R}^d$
\[(1-\eps)\|\bA\bP\by\|_p^p\le\|\bA'\bP\by\|_p^p\le(1+\eps)\|\bA\bP\by\|_p^p.\]
We remark that if $\bA$ is not full rank, the task would be to show claim for all $\by$ in the column span of $\bP^{-1}$. 

Let $\bB=\bA\bP$ and observe that \algref{alg:coreset:embedding} rounds all coordinates of each row $\bb_t=\ba_t\bP$ to their closest power of $(1+\eps')$. 
Then we have 
\[\left\lvert\langle\bb_t,\by\rangle^p-\langle\bb'_t,\by\rangle^p\right\rvert\le\eps'\cdot\max_{\by\in\mathbb{R}^d,\|\by\|_p=1}\left\lvert\langle\bb_t,\by\rangle\right\rvert^p.\]
For each $t\in[n]$, let $s(\bb_t)$ denote the sensitivity of $\bb_t$ with respect to $\bB$, so that
\[s(\bb_t)=\max_{\by\in\mathbb{R}^d}\frac{\left\lvert\langle\bb_t,\by\rangle\right\rvert^p}{\|\bB\by\|_p^p}.\]
Recall that $\bB=\bA\bP$ and $\bP$ is a preconditioner for $\bM$, so that $\bM\bP$ is well-conditioned. 
That is, for any unit vector $\by\in\mathbb{R}^d$, we have
\[\frac{1}{d^{pC}}\cdot\|\bM\bP\by\|_p^p\le\|\by\|_p^p\le d^{pC}\cdot\|\bM\bP\by|_p^p,\]
for a fixed constant $C>0$. 
Since $\bM$ is a constant-factor subspace embedding for $\bA$, then we certainly have 
\[\frac{1}{d^{p(C+1)}}\cdot\|\bA\bP\by\|_p^p\le\|\by\|_p^p\le d^{p(C+1)}\cdot\|\bA\bP\by|_p^p.\]
Thus it follows that 
\[\frac{1}{d^{p(C+1)}}\cdot\|\bB\by\|_p^p\le\|\by\|_p^p\le d^{p(C+1)}\cdot\|\bB\by|_p^p.\]
Therefore, 
\[\max_{\by\in\mathbb{R}^d}\left\lvert\langle\bb_t,\by\rangle\right\rvert^p\le s(\bb_t)\cdot d^{p(C+1)},\]
so that
\[\left\lvert\|\langle\bb_t,\by\rangle^p-\langle\bb'_t,\by\rangle^p\right\rvert\le\eps'\cdot s(\bb_t)\cdot d^{p(C+1)}.\]
Then by triangle inequality, we have
\[\left\lvert\|\bB\by\|_p^p-\|\bB'\by\|_p^p\right\rvert\le\sum_{t=1}^n\left\lvert\|\langle\bb_t,\by\rangle-\langle\bb'_t,\by\rangle\right\rvert^p\le\eps'\cdot\sum_{t=1}^n s(\bb_t)\cdot d^{p(C+1)}.\]
Since the sum of $L_p$ sensitivities is at most $d$ by \lemref{lem:total:lp:sens}, then 
\[\left\lvert\|\bB\by\|_p^p-\|\bB'\by\|_p^p\right\rvert\le\eps'\cdot
d^{p(C+2)}.\]
The desired claim then follows by the setting of $\eps'=\frac{\poly(\eps)}{\poly(d^p,\log(nd\Delta)}$, recalling that $\bB\by=\bA\bP\bP^{-1}\bx$ and $\bB'\by=\bA'\bP\bP^{-1}\bx$. 
\end{proof}
We now give the full guarantees of \algref{alg:coreset:embedding}, which gives an efficient encoding of an input matrix $\bA$. 
\begin{lemma}
\lemlab{lem:efficient:coreset:embedding}
Let $\bA\in\mathbb{R}^{m\times d}$ be a coreset construction with entries in $\{0\}\cup\left[\frac{1}{\poly(n)},\poly(n)\right]$. 
Let 
\[f\left(d,\frac{1}{\eps},p\right)=\begin{cases}
\frac{d}{\eps^2}\log\frac{d}{\eps}\polylog\left(\frac{d}{\eps}\right),\qquad& p\in[1,2)\\
\O{\frac{d}{\eps^2}},\qquad& p=2\\
\frac{d^{p/2}}{\eps^2}\cdot\polylog\left(\frac{d}{\eps}\right),\qquad&p>2     
\end{cases}.\]
Then $\bA'$ is a $(1+\eps)$-coreset for $\bA$ for $L_p$ subspace embedding that uses 
$\O{f\left(d,1,p\right)}\cdot (d\log n)+md\cdot\polylog\left(d,\frac{1}{\eps},\log(nd\kappa),\log\frac{1}{\delta}\right)$ bits of space. 
\end{lemma}
\begin{proof}
Firstly, note that \algref{alg:coreset:embedding} outputs a pair $(\bM,\bB')$ that encodes $\bB'$ using $\bM$. 
By \lemref{lem:x:to:xprime:embedding}, we have that for all $\bx\in\mathbb{R}^d$, 
\[(1-\eps)\|\bA\bx\|_1\le\|\bA'\bx\|_1\le(1+\eps)\|\bA\bx\|_1,\]
where $\bA'=\bB'\bP^{-1}$. 
Thus, $\bA'$ is a coreset for $\bA$. 
 
It remains to analyze the space complexity of $(\bM,\bB')$. 
Since $\bM$ is a constant-factor $L_p$ subspace embedding, then by \thmref{thm:lewis:weights}, $\bM$ can be represented using $\O{f(d,1,p)}\cdot d\log(n)$ bits of space. 
Furthermore, each offset $\bb'$ in $\bB'$ is encoded using $\O{d\log\left(\frac{1}{\eps}+\log(nd\kappa)\right)}$ bits of space due to storing the $d$ coordinates of the offset rounded to the power of $(1+\eps')$. 
Specifically, we store the exponent of each offset after the rounding, rather than the explicit coordinates. 
Therefore, the output of the algorithm is encoded using $\O{f\left(d,1,p\right)}\cdot (d\log n)+md\cdot\polylog\left(d,\frac{1}{\eps},\log(nd\kappa),\log\frac{1}{\delta}\right)$. 
\end{proof}

\subsection{Subspace Embedding in the Streaming Model}
In this section, we show how to use a global encoding combined with the efficient encoding from the previous section in order to achieve our main algorithm for $L_p$ subspace embeddings in the row-arrival model. 

Similar to $(k,z)$-clustering, the main downfall of immediately applying the previous efficient encoding to each node of a merge-and-reduce tree is that the constant-factor approximation requires $\O{d^2\log(nd)}$ bits of storage per efficient encoding, and there can be $\polylog(\log(nd))$ such efficient encodings due to the height of the merge-and-reduce tree on a data stream produced by online Lewis weight sampling. 
As before, we circumvent this issue by showing that we can instead consider a single constant-factor approximation to the global dataset. 
We then show that the error of maintaining such a global constant-factor approximation does not compound too much over the course of the stream. 

We give the algorithm in full in \algref{alg:embedding:insert:stream}. 

\begin{algorithm}[!htb]
\caption{$L_p$-Embedding on Insertion-Only Stream}
\alglab{alg:embedding:insert:stream}
\begin{algorithmic}[1]
\Require{Matrix $\bA=\{\ba_1,\ldots,\ba_n\}\in\mathbb{R}^{n\times d}$ that arrives as a data stream, $\eps\in(0,1)$, parameter $p\ge 1$}
\Ensure{$(1+\eps)$-coreset for $L_p$-subspace embedding}
\State{$\bZ\gets\emptyset$, $\lambda\gets\O{\frac{d}{\eps^2}\cdot\log d\log n}$}
\For{each $t\in[n]$}
\State{Use $\bZ\cup\{\widehat{\ba_t}\}$ to compute a $\O{1}$-approximation $\widehat{\sigma(\ba_t)}$ to the online Lewis weight of $\ba_t$}
\State{$p(\ba_t)\gets\min(1,\lambda\cdot\widehat{\sigma(\ba_t)})$}
\State{With probability $p(\ba_t)$, sample $\frac{1}{p(\ba_t)}\cdot\ba_t$ into a stream $\calS'$}
\State{Let $\bZ$ be the running output of merge-and-reduce on $\calS'$ using the efficient encoding for coreset construction in \algref{alg:coreset:embedding} with a \emph{global} constant-factor embedding, for accuracy $\frac{\eps}{\poly(\log\log nd)}$ and failure probability $\frac{1}{\poly\left(\frac{1}{\eps},\log(nd)\right)}$} 
\EndFor
\State{\Return $\bZ$}
\end{algorithmic}
\end{algorithm}
We first show that given disjoint matrices $\bQ_1,\ldots,\bQ_m$ of $\bA$, we can apply \algref{alg:coreset:embedding} and although the output $\bQ'_i$ from our efficient encoding will no longer necessarily be a strong coreset for each corresponding $\bQ_i$, the resulting error will only be an additive $\eps'\cdot\|\bA\bx\|_p^p$ for a fixed $\bx\in\mathbb{R}^d$, which suffices for our purposes. 
Specifically, we will set $\eps'=\frac{\eps}{\poly(\log\log(nd)}$ and $m=\poly(\log\log(nd))$, provided the entries of $\bA$ are integers bounded in magnitude by $\poly(nd)$. 
\begin{lemma}
\lemlab{lem:each:coreset:embedding}
Let $\bQ_1\circ\ldots\circ\bQ_m=\bA$ be a partition of $\bA$. 
Let $\bM$ be a constant-factor $L_p$ subspace embedding on $\bA$. 
For each $i\in[m]$, let $\bQ'_i$ be the resulting matrix after applying the inverse of the preconditioner to the rounded matrix resulting from $\bQ_i$. 
Then for all $i\in[m]$ and all $\bx\in\mathbb{R}^d$, 
\[\left\lvert\|\bQ'_i\bx\|_p^p-\|\bQ_i\bx\|_p^p\right\rvert\le\frac{\eps}{\poly(\log\log nd)}\cdot\|\bA\bx\|_p^p.\]
\end{lemma}
\begin{proof}
Note that by \lemref{lem:x:to:xprime:embedding} with accuracy $\frac{\eps}{\poly(\log\log nd)}$, we have
\[\left\lvert\|\bA'\bx\|_p^p-\|\bA\bx\|_p^p\right\rvert\le\frac{\eps}{\poly(\log\log nd)}\cdot\|\bA\bx\|_p^p.\]
Since 
\[\left\lvert\|\bQ'_i\bx\|_p^p-\|\bQ_i\bx\|_p^p\right\rvert\le\sum_{i=1}^m\left\lvert\|\bQ'_i\bx\|_p^p-\|\bQ_i\bx\|_p^p\right\rvert=\left\lvert\|\bA'\bx\|_p^p-\|\bA\bx\|_p^p\right\rvert,\]
then the claim follows. 
\end{proof}
We now give the full guarantees of \algref{alg:embedding:insert:stream} by showing correctness of the global encoding and specifically, that there is no compounding error across the $\poly(\log\log(nd))$ iterations. 
\begin{theorem}
\thmlab{thm:subspace:per:update}
Given a matrix $\bA$ of $n$ rows on $[-M,\ldots,-1,0,1,\ldots,M]^d$, with $M=\poly(n)$ and online condition number $\kappa$, there exists a one-pass streaming algorithm in the row arrival model that outputs a $(1+\eps)$-strong coreset of $\bA$ and uses $\O{d}\cdot f\left(d,\frac{1}{\eps},p\right)$ words of space, where
\[f\left(d,\frac{1}{\eps},p\right)=\begin{cases}
\frac{d}{\eps^2}\log\frac{d}{\eps}\polylog\left(\frac{d}{\eps}\right),\qquad& p\in[1,2)\\
\O{\frac{d}{\eps^2}},\qquad& p=2\\
\frac{d^{p/2}}{\eps^2}\cdot\polylog\left(\frac{d}{\eps}\right),\qquad&p>2     
\end{cases}.\]
\end{theorem}
\begin{proof}
Consider \algref{alg:embedding:insert:stream}. 
We prove correctness by induction on $t$, claiming that after $j\cdot g\left(n,d,\frac{\polylog(\log(n\kappa))}{\eps},p\right)$ samples into $\calS'$, then $\bZ$ is a $\left(1+\O{\frac{\eps}{\poly(\log\log nd)}}\right)^{j+1}$ coreset of $\bA_t$. 
The online Lewis weight of the first row is always $1$ and thus $\bZ=\ba_1$ after the first step, and the base case is complete. 
For the inductive hypothesis, we condition on the correctness of $\bZ$ being a $\left(1+\O{\frac{\eps}{\poly(\log\log nd)}}\right)^j$-coreset of $\bA_{t-1}=\ba_1\circ\ldots\circ\ba_{t-1}$ after $j-1$ merges have occured. 
Then $\bZ\circ\ba_t$ is a $\left(1+\frac{\eps}{\poly(\log\log nd)}\right)$-coreset of $\bA_t=\ba_1\circ\ldots\circ\ba_t$, so we can use $\bZ\circ\ba_t$ to compute a constant-factor approximation to the online Lewis weight $\sigma(\ba_t)$ of $\ba_t$. 
Note that merge-and-reduce will be correct if $\calS'$ is a stream with length at most $g\left(n,d,\frac{\polylog(\log(n\kappa))}{\eps},p\right)$, since the failure probability is set to be $\frac{1}{\poly\left(\frac{1}{\eps},\log(n\Delta)\right)}$ and merge-and-reduce will consider at least $g\left(n,d,\frac{1}{\eps},p\right)$ stream updates before applying a new coreset construction. 
Thus by \lemref{lem:each:coreset:embedding} the correctness of merge-and-reduce on $\calS'$, we have that $\bZ$ after time $t$ will be a $\left(1+\O{\frac{\eps}{\poly(\log\log nd)}}\right)^{j+1}$-strong coreset for $X_t$ after $j$ merges, which completes the induction. 
Finally, we have that by \thmref{thm:total:online:lewis}, at most $g\left(n,d,\frac{\polylog(\log(n\kappa))}{\eps},p\right)$ samples will be inserted into $\calS$, so that $j\le\polylog(n\kappa)$. 
The correctness guarantee then follows by rescaling $\eps$. 

To analyze the space complexity, note that by \thmref{thm:online:lewis:embedding} through \thmref{thm:total:online:lewis}, we have that the total number of sampled rows into $\calS'$ is $n'=\O{g\left(n,d,\frac{1}{\eps},p\right)}$, with high probability. 
We do not maintain $\calS'$, but instead feed $\calS'$ as input to the merge-and-reduce subroutine. 
Hence, the input to merge-and-reduce has size $n'$ and so each coreset implementation has size $f\left(d,\frac{1}{\eps},p\right)\cdot\polylog(n')$. 
We use the efficient encoding for coreset construction given in \algref{alg:coreset:embedding} with accuracy $\frac{\eps}{\poly(\log\log(n\kappa))}$ and failure probability $\frac{1}{\poly\left(\frac{1}{\eps},\log(n\Delta)\right)}$. 
Similar to the analysis of \lemref{lem:efficient:coreset:embedding}, we recall that total representation of $\bZ$ requires a constant-factor $L_p$ subspace embedding $\bM$ that uses $\O{f\left(d,1,p\right)}\cdot (d\log n)$ bits of space. 
We then use \lemref{lem:each:coreset:embedding} to encode each of the coresets of size $f\left(d,\frac{1}{\eps},p\right)\cdot\polylog(n')$, using $d\cdot\polylog\left(d,\frac{1}{\eps},\log(nd\kappa),\log\frac{1}{\delta}\right)$ bits per row. 
Since $n'=\O{g\left(n,d,\frac{1}{\eps},p\right)}$, it follows that there are at most $\polylog\left(d,\frac{1}{\eps},\log(nd\kappa),\log\frac{1}{\delta}\right)$ such coresets. 
Therefore, the total space required is 
\[\O{d\log n}\cdot f(d,1,p)+f\left(d,\frac{1}{\eps},p\right)\cdot\polylog(n')\cdot d\cdot\polylog\left(\frac{1}{\eps},\log(nd\kappa)\right).\]
in bits, which is equivalent to $\O{d}\cdot f\left(d,\frac{1}{\eps},p\right)$ words of space, as desired.
\end{proof}

\subsection{Fast Algorithm for Subspace Embeddings}
In this section, we show that our one-pass streaming algorithm can be implemented in $\O{d}$ amortized update time. 
We first require a crude $n^\alpha$-approximation to the online Lewis weight for each row that arrives in the stream. 
To that end, we first recall that the $L_p$ sensitivity is a $\poly(d)$ approximation to the $L_p$ Lewis weights.

\begin{lemma}
\cite{CohenP15,ChenD21}
\lemlab{lem:lp:lewis}
For a matrix $\bA\in\mathbb{R}^{n\times d}$, let $s(\ba_t):=\max_{\by\in\mathbb{R}^d: \|\bA\by\|_p=1}\left\lvert\langle\ba_t,\by\rangle\right\rvert^p$ denote the $L_p$ sensitivity of the $t$-th row of $\bA$ and let $w(\ba_t)$ denote the $L_p$ Lewis weight of $\ba_t$. 
Then
\[\frac{w(\ba_t)}{d^{-1-p/2}}\le s(\ba_t)\le w(\ba_t).\]
\end{lemma}
We first relate the square root of the leverage score of a row $\ba_t$ in a matrix $\bA$ with the $p$-th root of the $L_p$ sensitivity of $\ba_t$. 
\begin{lemma}
\lemlab{lem:root:scores}
For a matrix $\bA\in\mathbb{R}^{n\times d}$, let $\xi(\ba_t)$ denote the square root of the leverage score of $\ba_t$ and let $\chi(\ba_t)$ denote the $L_p$ sensitivity of $\ba_t$ raised to the $\frac{1}{p}$ power, i.e., $(s(\ba_t))^{1/p}$, where $s(\ba_t)$ is the $L_p$ sensitivity of $\ba_t$. 
Then for $p\in[1,2]$, we have
\[s(\ba_t)\le\xi(\ba_t)\le n^{\frac{1}{p}-\frac{1}{2}}\cdot s(\ba_t),\]
and for $p>2$, we have
\[\xi(\ba_t)\le s(\ba_t)\le n^{\frac{1}{2}-\frac{1}{p}}\cdot\xi(\ba_t).\]
\end{lemma}
\begin{proof}
Note that for a fixed vector $\bx\in\mathbb{R}^d$ that is not in the kernel of $\bA$, we have that $\frac{\langle\ba_t,\bx\rangle}{\|\bA\bx\|_2}$ is within a $\sqrt{n}$ factor of $\frac{\langle\ba_t,\bx\rangle}{\|\bA\bx\|_p}$, for $p\in[1,\infty)$. 
The claim then follows from the definition that the square root of the leverage score is the maximum of $\frac{\langle\ba_t,\bx\rangle}{\|\bA\bx\|_2}$ across all possible such $\bx$, while the $\frac{1}{p}$-th root of the $L_p$ sensitivity is the maximum of $\frac{\langle\ba_t,\bx\rangle}{\|\bA\bx\|_p}$ over all such possible $\bx$. 
\end{proof}
We next show that we can quickly compute crude approximations to the $L_p$ Lewis weights by instead computing crude approximations to the square root of the leverage scores of the rows of the matrix. 
\begin{lemma}
\lemlab{lem:crude:lewis:computation}
Given $\bA\in\mathbb{R}^{n\times d}$, let $\bB\in\mathbb{R}^{m\times d}$ be a matrix such that $\frac{1}{C}\cdot\bA^\top\bA\preceq\bB^\top\bB\preceq C\cdot\bA^\top\bA$, for some constant $C$, there exists an algorithm that computes a $n^{\alpha}$ approximations to the $L_p$ Lewis weight using $\O{\nnz(\bA)}$ runtime.  
\end{lemma}
\begin{proof}
Recall that the leverage score of $\ba_t$ with respect to $\bA$ is $\ba_t^\top(\bA^\top\bA)^{-1}\ba^\top$. 
Since $\frac{1}{C}\cdot\bA^\top\bA\preceq\bB^\top\bB\preceq C\cdot\bA^\top\bA$, then $\ba_t^\top(\bB^\top\bB)^{-1}\ba^\top$ is a $C$-approximation to the leverage score. 
Let $\bZ=(\bB^\top\bB)^{-1/2}$, so that $\|\bZ\ba_t\|_2^2$ is a $C$-approximation to the leverage score of $\ba_t$. 
Thus it suffices to find an $n^{\alpha/2}$-approximation to $\|\bZ\ba_t\|_2$. 
To that end, observe that for a unit vector $\bv\in\mathbb{R}^d$ a random Gaussian vector $\bg\in\mathbb{R}^d$ whose entries are independently drawn from $\calN(0,1)$, we have that 
\begin{align*}
\PPr{\langle\bv,\bg\rangle\ge\log^2 n}&\le 1-\frac{1}{\poly(n)}\\
\PPr{\langle\bv,\bg\rangle\le\frac{1}{n^{1/{100p}}}}&\le\O{\frac{1}{n^{1/{100p}}}}.
\end{align*}
Thus by the median over $\O{\frac{1}{100p}}$ values of $\|\bg^{(i)}\bM\ba_t\|_2$, we can approximately the leverage score of $\ba_t$ within $n^{1/100p}$. 
Hence, we can also approximate the square root of the leverage score of $\ba_t$ within $n^{1/100p}$. 
By \lemref{lem:root:scores}, this results in a $n^\alpha$-approximation to the $L_p$ Lewis weight of $\ba_t$.

For the runtime, observe that $\bA^\top\bA$ over changes $\poly(d)$ times throughout the stream, so the computation of $\bg^{(i)}\bM$ is a lower order term than the stream length, i.e., total runtime over all values of $\bM$ is $o(n)$. 
To compute $\bg^{(i)}\bM\ba_t$ from $\ba_t$ uses $\O{d}$ time per row and in fact, input-sparsity time over the entire matrix
\end{proof}
Putting everything together, we have:
\begin{theorem}
\thmlab{thn:lp:main}
Given a matrix $\bA$ of $n$ rows on $[-M,\ldots,-1,0,1,\ldots,M]^d$, with $M=\poly(n)$ and online condition number $\kappa$, there exists a one-pass streaming algorithm in the row arrival model that outputs a $(1+\eps)$-strong coreset of $\bA$ and uses $\O{d}\cdot f\left(d,\frac{1}{\eps},p\right)$ words of space, where
\[f\left(d,\frac{1}{\eps},p\right)=\begin{cases}
\frac{d}{\eps^2}\log\frac{d}{\eps}\polylog\left(\frac{d}{\eps}\right),\qquad& p\in[1,2)\\
\O{\frac{d}{\eps^2}},\qquad& p=2\\
\frac{d^{p/2}}{\eps^2}\cdot\polylog\left(\frac{d}{\eps}\right),\qquad&p>2     
\end{cases}.\]
Moreover, the amortized runtime is $\O{d}$ per update. 
\end{theorem}
\begin{proof}
We proceed by induction on the time $t$. 
The correctness at the first time is clear, since the initial nonzero row must be sampled. 
Now assuming that the correctness of both a constant-factor coreset and a $(1+\eps)$-coreset at time $t-1$, then at time $t$, either the row is added to the batch of size $k$, in which case both coreset invariants are maintained, or the sampling procedure is performed on the batch because the batch has gotten too large. 
In the latter case, we use the constant-factor coreset through \lemref{lem:crude:lewis:computation} to compute a crude approximation to the Lewis scores of each row in a batch of $k$ rows. 
This uses $\O{\nnz(\bA)}$ total runtime and produces an insertion-only stream of length $\O{n^{1-\Omega(1/p)}}$. 
We then use this as the input to \algref{alg:embedding:insert:stream}. 
By \thmref{thm:subspace:per:update}, there is $\poly(d)$ update time on the stream of length $o(n)$ to generate both a constant-factor coreset and a $(1+\eps)$-factor coreset at time $t$, which completes our induction. 
Thus for $d\ll n^{1/p}$, the total runtime for \algref{alg:embedding:insert:stream} is $\O{nd}$, and so the overall runtime is $\O{nd+\nnz(\bA)}$, which is amortized runtime $\O{d}$. 
\end{proof}

\section*{Acknowledgments}
The work was initialized while David P. Woodruff and Samson Zhou were visiting the Institute for Emerging CORE Methods in Data Science (EnCORE) supported by the NSF grant 2217058. 
The work was conducted in part while David P. Woodruff and Samson Zhou were visiting the Simons Institute for the Theory of Computing as part of the Sublinear Algorithms program. 
David P. Woodruff is supported in part by Office of Naval Research award number N000142112647 and a Simons Investigator Award.
Liudeng Wang and Samson Zhou were supported in part by NSF CCF-2335411.

\def\shortbib{0}
\bibliographystyle{alpha}
\bibliography{references}

\end{document}